\documentclass{amsart}[12pt]
\usepackage{latexsym}
\usepackage{amsmath}
\usepackage{amsfonts}
\usepackage{amssymb}
\usepackage{graphicx}
\usepackage{amsthm}
\usepackage[top=2.5cm, bottom=2.5cm, left=2.5cm, right=2.5cm]{geometry}
\usepackage{hyperref}
\usepackage{slashed,mathtools,color}
\usepackage{bbm}
\usepackage{comment}

\newcommand{\pfstep}[1]{\vspace{.5em} {\it \noindent #1.}}

\newcommand{\beq}{\begin{equation}}
\newcommand{\eeq}{\end{equation}}

\def\f {\frac}
\def\rd {\partial}
\def\ls {\lesssim}
\def\de {\delta}
\def\i {\infty}
\def\alp {\alpha}
\def\bt {\beta}
\def\nab {\nabla}
\newcommand{\ud}{\mathrm{d}}

\newcommand{\ve}{\varepsilon}
\def\kn {k^{\bf [n]}}

\def\kd {k^{(d)}}
\def\wtDn {\widetilde{\mathfrak D}^{\bf [n]}}



\theoremstyle{plain}
\newtheorem{theorem}{Theorem}[section]
\newtheorem{proposition}[theorem]{Proposition}
\newtheorem{lemma}[theorem]{Lemma}
\newtheorem{corollary}[theorem]{Corollary}

\theoremstyle{definition}
\newtheorem{definition}[theorem]{Definition}

\newtheorem{remark}[theorem]{Remark}

\numberwithin{equation}{section}

\allowdisplaybreaks[4] 

\swapnumbers
\begin{document}

\title{Asymptotically Kasner-like singularities}

\author{Grigorios Fournodavlos}
\address{Laboratoire Jacques-Louis Lions, Sorbonne Universit\'e, 75005 Paris, France}
\email{grigorios.fournodavlos@sorbonne-universite.fr}
\author{Jonathan Luk}
\address{Department of Mathematics, Stanford University, Stanford, CA 94305, USA}
\email{jluk@stanford.edu}

\date{}

\maketitle

\begin{abstract}
We prove existence, uniqueness and regularity of solutions to the Einstein vacuum equations taking the form 
\begin{align*}
\begin{split} 
{^{(4)}g} =& -dt^2 + \sum_{i,j = 1}^3 a_{ij}t^{2 p_{\max\{i,j\}}}\, \ud x^i\, \ud x^j
\end{split} 
\end{align*}
on $(0,T]_t \times \mathbb T^3_x$, where $a_{ij}(t,x)$ and $p_i(x)$ are regular functions without symmetry or analyticity assumptions. These metrics are singular and asymptotically Kasner-like as $t\to 0^+$. These solutions are expected to be highly non-generic, and our construction can be viewed as solving a singular initial value problem with Fuchsian-type analysis where the data are posed on the ``singular hypersurface'' $\{ t = 0\}$. This is the first such result without imposing symmetry or analyticity.

To carry out the analysis, we study the problem in a synchronized coordinate system. In particular, we introduce a novel way to perform (weighted) energy estimates in such a coordinate system based on estimating the second fundamental forms of the constant-$t$ hypersurfaces.

\end{abstract}

\parskip = 0 pt

\tableofcontents

\section{Introduction}

The Kasner spacetime $((0,+\infty)\times \mathbb T^3, ^{(4)}g)$, where
\begin{equation}\label{eq:Kasner}
^{(4)}g = -\ud t^2 + \sum_{i=1}^3 t^{2p_i} (\ud x^i)^2
\end{equation}
(with $p_i$ being constants such that $\sum_{i=1}^3 p_i = \sum_{i=1}^3 p_i^2 = 1$) is an explicit solution to the Einstein vacuum equations 
\begin{equation}\label{eq:EVE}
Ric(^{(4)}g)=0.
\end{equation}
As long as all $p_i\neq 0$, the Kasner solution moreover represents a \emph{singularity} as $t\to 0^+$. This is manifested in particular by the blowup of the Kretschmann scalar $R_{\mu\nu\alpha\bt} R^{\mu\nu\alpha\bt}$.

In an influential paper \cite{eLimK63}, Lifshitz--Khalatnikov considered the class of spacetimes solutions to \eqref{eq:EVE} with the form
\begin{equation}\label{eq:AVTD}
^{(4)}g = -\ud t^2 + \sum_{i=1}^3 t^{2p_i} \omega_i^2
\end{equation}
where $\omega_i$ are spatial $1$-forms with a ``finite limit'' as $t\to 0^+$ and $p_i = p_i(x)$ are now \emph{spatially-dependent} functions satisfying $\sum_{i=1}^3 p_i(x) = \sum_{i=1}^3 p_i^2(x) = 1$. The spacetime metrics \eqref{eq:AVTD} are Kasner-like asymptotically as $t\to 0^+$ except that the Kasner exponents are now functions. They are also sometimes called \emph{asymptotically velocity term dominated} (AVTD), a terminology that is used to mean that the asymptotics near the singularity is described by a simpler system of velocity term dominated equations \cite{pCjIvM1990,dEeLrS1972}.  Importantly, it is argued in \cite{eLimK63} that this class of spacetime solutions to \eqref{eq:EVE} depend only on three ``functional degrees of freedom'', which is one fewer than that for the Cauchy problem of \eqref{eq:EVE}, and they are therefore expected to be highly non-generic.

In this paper, we construct a large class of solutions to \eqref{eq:EVE} with the asymptotically Kasner-like behavior of \eqref{eq:AVTD}. Our construction in fact has full three functional degrees of freedom and includes all the spacetimes considered in the heuristics in  \cite{eLimK63} (see~Remark~\ref{rmk:counting}). Some previous constructions are known with either analyticity or symmetry assumptions (see Section~\ref{sec:construction}); our construction is the first without such assumptions. 

More precisely, our goal will be to construct a metric taking the form
\begin{align}\label{metricansatz}
\begin{split} 
{^{(4)}g} :=&\:  -\ud t^2 + g\\
:=&\: -\ud t^2 + \sum_{i,j=1}^3 a_{ij} t^{2p_{\max\{i,j\}}} \ud x^i\, \ud x^j, 
\end{split} 
\end{align}
where $(t,x^1, x^2, x^3)\in (0,T]\times \mathbb T^3$ for some $T>0$, $p_i:\mathbb T^3\to \mathbb R$ are smooth, time-independent functions, and $a_{ij}:(0,T]\times \mathbb T^3\to \mathbb R$ are smooth functions (symmetric in $i$ and $j$) which extend to continuous functions $:[0,T]\times \mathbb T^3\to \mathbb R$. Moreover, $a_{ij}$ obey
\begin{equation}\label{eq:a.limit}
\lim_{t\to 0^+} a_{ij}(t,x) = c_{ij}(x),
\end{equation}
where $c_{ij}$ are some prescribed smooth functions (symmetric in $i$ and $j$).

Notice that in the language of \eqref{eq:AVTD}, the ansatz \eqref{metricansatz} imposes the condition $\omega_1 \wedge \ud \omega_1 = 0$ for $\omega_1 = a_{11}^{\f 12} dx^1$. As we will explain in Remark~\ref{rmk:counting}, this condition is what restricts the functional degrees of freedom in our construction.

We will prove existence, uniqueness and regularity of solutions of the form \eqref{metricansatz}. The following is our main \emph{existence} theorem:
\begin{theorem}[Existence of solution]\label{mainthm}
Suppose the following assumptions hold:
\begin{enumerate}
\item  {The (time-independent) functions $c_{ij},\,p_i:\mathbb T^3\to \mathbb R$ are smooth for $i,j=1,2,3$, and that $c_{ij} = c_{ji}$.}
\item $\sum_{i=1}^3 p_i(x) = \sum_{i=1}^3 p_i^2(x) = 1$ pointwise.
\item It holds that $p_1(x) <p_2(x)< p_3(x)<1$ pointwise.
\item It holds that $c_{11}(x),\,c_{22}(x),\,c_{33}(x) >0$.
\item The following three \underline{asymptotic differential constraint equations} are satisfied:
\begin{align}\label{eq:cij}
\sum_{\ell=1}^3\bigg[\frac{\partial_ic_{\ell\ell}}{c_{\ell\ell}}(p_\ell-p_i)
+2\partial_\ell\kappa_i{}^\ell
+\mathbbm{1}_{\{\ell>i\}}\frac{\partial_\ell(c_{11}c_{22}c_{33})}{c_{11}c_{22}c_{33}}\kappa_i{}^\ell\bigg]=0,\qquad i=1,2,3,
\end{align}
where $\kappa_i{}^i = -p_i$ (without summing), $\kappa_1{ }^2 = (p_1 - p_2 )\f{c_{12}}{c_{22}}$, $\kappa_2{ }^3 = (p_2 - p_3)\f{c_{23}}{c_{33}}$, $\kappa_1{ }^3 = (p_2-p_1)\f{c_{12} c_{23}}{c_{22} c_{33}} + (p_1-p_3)\f{c_{13}}{c_{33}}$ and $\kappa_i{ }^\ell = 0$ if $\ell<i$, $\mathbbm{1}_{\{\ell>i\}}=1$ if $\ell>i$, 
$\mathbbm{1}_{\{\ell>i\}}=0$ if $\ell\leq i$.
\end{enumerate}
Then there is a  {$C^2$} solution to the Einstein vacuum equations  {\eqref{eq:EVE}} of the form \eqref{metricansatz}, for a $T>0$ depending on $c_{ij},p_i$, which satisfies \eqref{eq:a.limit}.
\end{theorem}
\begin{remark}[\eqref{metricansatz} is a Lorentzian metric]
Notice that under condition (3), the eigenvalues of $g$ as in  \eqref{metricansatz} are approximately $t^{2p_i} c_{ii}$ ($i=1,2,3$) for small $t$. Hence, given $p_i$ as in the theorem and the condition \eqref{eq:a.limit}, it follows that \eqref{metricansatz} is a well-defined \emph{Lorentzian} metric in all of $(0,T_0]\times\mathbb{T}^3$, for some $T_0>0$.
\end{remark}
\begin{remark}[Localizing the assumptions]
For technical convenience, we assume that there is a \emph{global} system of coordinates on $\mathbb T^3$ so that the assumptions of Theorem~\ref{mainthm} hold. One may in principle hope to use a localization argument to construct more general spacetimes for which we require only that around every point in $\mathbb{T}^3$, there is a coordinate patch $(x_1,x_2,x_3)$ such that the assumptions of Theorem~\ref{mainthm} hold. This, however, is not carried out in the present paper.
\end{remark}
\begin{remark}[Asymptotic CMC condition and asymptotic constraints]
The conditions  {(2) and (5) in Theorem~\ref{mainthm}} guarantee that a metric of the form \eqref{metricansatz} satisfies asymptotically, along the level sets of $t$, 1) the constraint equations and 2) the CMC gauge to leading order, as $t\rightarrow 0^+$. More precisely, condition  {(2)} is equivalent to
\begin{align}\label{asymHamconst}
\lim_{t\rightarrow 0^+}t (\text{tr}k) =-1,&&\lim_{t\rightarrow 0^+}t^2[R(g)-|k|^2+(\text{tr}k)^2]=0,
\end{align}
while condition (5) is equivalent to
\begin{align}\label{asymmomconst}
\lim_{t\rightarrow 0^+}t(\nabla_jk_i{}^j-\nabla_i\text{tr}k)=0,&&i=1,2,3;
\end{align}
see Lemma~\ref{lem:asymconst}. Note that condition  {(2)} is algebraic in the Kasner exponents $p_i$'s, while condition  {(5)} is differential in the $c_{ij}$'s. 
\end{remark}
\begin{remark}[Functional degrees of freedom and considerations in \cite{eLimK63}]\label{rmk:counting}
Note that $c_{ij}$ and $p_i$ consist of $9$ functions. On the other hand, the assumptions (2) and (5) in Theorem~\ref{mainthm} impose a total of $5$ conditions, leaving $4$ functional degrees of function.

There is in fact an additional \emph{residual gauge freedom}, namely, we can introduce a change of coordinates
\begin{equation}\label{eq:residual.gauge}
\widetilde{x}^1 = x^1,\quad \widetilde{x}^2 = x^2,\quad \widetilde{x}^3 = f(x^1,x^2,x^3),
\end{equation}
for some smooth $f$ such that $\f{\rd f}{\rd x^3} \neq 0$, then the resulting metric will have the same form as \eqref{metricansatz} (in the sense that the new $\widetilde{g}_{11}$ term is $O(t^{2p_1})$, the new $\widetilde{g}_{12}$, $\widetilde{g}_{22}$ terms are $O(t^{2p_2})$, and the new $\widetilde{g}_{13}, \widetilde{g}_{23}$, $\widetilde{g}_{33}$ terms are $O(t^{2p_3})$.) 

Thus, there are a total of $3$ functional degrees of freedom, which is one fewer than that for the initial value problem for the Einstein vacuum equations. It is for this reason that \cite{eLimK63} argued that metrics of the form \eqref{metricansatz} are \emph{non-generic}.

Notice that while we only construct a non-generic class of spacetimes, we do construct a class that includes all the metrics considered in \cite{eLimK63} (modulo the endpoint case; see Remark~\ref{rmk:endpoint}). Indeed, using the change of coordinates in \eqref{eq:residual.gauge}, one can locally change coordinates to the form
$$g = a_{11} t^{2p_1} (\ud x^1)^2 + a_{22} t^{2p_2} (\ud x^2)^2 + a_{33} t^{2p_3}(\ud x^3)^2 + 2 a_{12} t^{p_2} \ud x^1 \ud x^2 + 2 a_{13} t^{p_3} \ud x^1 \ud x^3,$$
which is exactly the local form of the metrics considered in the work of Lifshitz--Khalatnikov; see \cite[equation (3.25)]{eLimK63}.
\end{remark}
\begin{remark}[Some limiting cases]\label{rmk:endpoint}
 Our analysis degenerates in any of the limits $p_3 \to 1$ or $p_{i+1} - p_i \to 0$ (see (3) in Theorem~\ref{mainthm}). A particularly interesting limiting case that we do not cover is when
$$\{x\in\mathbb{T}^3: p_1(x)=-\frac{1}{3},\;p_2(x)=p_3(x)=\frac{2}{3}\} \neq \emptyset,$$
but still assuming $p_3(x) <1, \, \forall x$.
While we do not cover this case, it is possible that \cite{sK2007} is relevant. Notice that to handle possible terms with $p_2(x) = p_3(x)$, we need a new argument in constructing the approximate solution in Section~\ref{sec:parametrix}, but the analysis in the subsequent sections could in principle be carried out along the same lines.

Finally, we note that allowing $p_1(x)=-\frac{1}{3},\;p_2(x)=p_3(x)=\frac{2}{3}$ would also be relevant to constructing Schwarzschild-like singularities since locally the Schwarzschild singularity could be modeled by the Kasner singularity with $p_1 = -\f 13$, $p_2 = p_3 = \f 23$ (cf.~\cite{Fo2016} and discussions in Section~\ref{sec:construction}).
\end{remark}
 {We now turn to \emph{uniqueness}.} It is hard to talk about geometric uniqueness in the above singular initial value problem, since the setup itself includes the expression \eqref{metricansatz} of the spacetime metric. However, we can obtain uniqueness {\it in our gauge}, i.e.~within the class of metrics satisfying \eqref{metricansatz} . More precisely, we prove that given two solutions of the form \eqref{metricansatz} which (1) obey the estimates \eqref{uniqueness.condition.0} and \eqref{uniqueness.condition.1} which is proven in Theorem~\ref{mainthm} and (2) converge to each other sufficiently fast as $t\to 0^+$, then they must in fact be the same.
\begin{theorem}[Uniqueness of solutions]\label{thm:uniq}
Given the assumptions of Theorem~\ref{mainthm}, there exists $M_u\in \mathbb N$ sufficiently large (depending on the given data $p_i$ and $c_{ij}$) such that the following holds.

Let ${^{(4)}g},{^{(4)}\tilde{g}}$ be two  {$C^3$} solutions to the Einstein vacuum equations \eqref{eq:EVE} of the form \eqref{metricansatz} in $(0,T]\times\mathbb{T}^3$ for some $T>0$, such that
\begin{itemize}
\item the corresponding $a_{ij}$ and $\tilde{a}_{ij}$ converge to $c_{ij}$ with the following rate
\begin{equation}\label{uniqueness.condition.0}
\sum_{|\alp|\leq 2} ( | \rd_x^\alp (a_{ij} - c_{ij})| +  | \rd_x^\alp (\tilde a_{ij} - c_{ij})| ) =O(t^{\ve});
\end{equation}
\item the corresponding $k_{i}{ }^j = -\f 12 (g^{-1})^{j\ell} \rd_t g_{j\ell}$ and $\tilde{k}_i{}^j = -\f 12 (\tilde g^{-1})^{j\ell} \rd_t \tilde g_{j\ell}$ obey the following estimates
\begin{equation}\label{uniqueness.condition.1}
\sum_{r=0}^{1} \sum_{|\alp|\leq 2-r} t^r( |\rd_t^r \rd_x^\alp (k_i{}^j - t^{-1} \kappa_i{}^j)| +  |\rd_t^r \rd_x^\alp (\tilde k_i{}^j - t^{-1} \kappa_i{}^j)| ) =O(\min\{ t^{-1+\ve},t^{-1+\ve-2p_j+2p_i}\});
\end{equation}
 and
\item the $g-\tilde{g}$ and $\rd_t(g-\tilde{g})$ converge to $0$ \underline{sufficiently fast} in the following sense:
\begin{equation}\label{uniqueness.condition.2}
\sum_{r=0}^{1} \sum_{|\alp|\leq  {3}-r} |\rd_t^r \rd_x^\alp (g - \tilde{g})| = O(t^{M_u}).
\end{equation}
\end{itemize}
Here, $\ve  = \min\{\min_x (p_3 - p_2)(x), \min_x (1-p_3)(x) \}>0$, and $\kappa_i{}^j$ as in Theorem~\ref{mainthm}.

Then ${^{(4)}g}={^{(4)}\tilde{g}}$  {on $(0,T]\times \mathbb T^3$.}
\end{theorem}

\begin{remark}[Asymptotics determined by approximate solutions]\label{rmk:uniqueness.1}
In the proof of our existence result (Theorem~\ref{mainthm}), we construct a sequence of smooth \emph{approximate solutions} $\{g^{\bf [n]}\}_{n=0}^{+\infty}$, for which we get more precise asymptotic information, as $t\to 0^+$, as $n$ increases; see already Sections~\ref{sec:ideas}, \ref{sec:parametrix} and \ref{sec:approx.constraints}. The actual solutions that we construct in Theorem~\ref{mainthm} then have asymptotics determined by an approximate solution $g^{\bf [n]}$(for some large $n$). From this point of view, one way to interpret our uniqueness result (Theorem~\ref{thm:uniq}) is to say that for $n$ sufficiently large, there is in fact only one solution whose asymptotics are governed by $g^{\bf [n]}$.
\end{remark}

\begin{remark}[Regularity implies asymptotic expansion]\label{rmk:uniqueness.2}
Given any $M_u \in \mathbb N$, there exists $A\in \mathbb N$ sufficiently large such that if \eqref{uniqueness.condition.0} and \eqref{uniqueness.condition.1} are replaced by the stronger regularity assumptions
\begin{equation}\label{uniqueness.condition.0.alt}
\sum_{|\alp|\leq A} ( |\rd_x^\alp (a_{ij} - c_{ij})| + |\rd_x^\alp (\tilde{a}_{ij} - c_{ij})|) =O(t^{\ve}),
\end{equation}
and
\begin{equation}\label{uniqueness.condition.0.alt.2}
\sum_{r=0}^{1} \sum_{|\alp|\leq A-r} t^r ( |\rd_t^r \rd_x^\alp (k_i{}^j - t^{-1} \kappa_i{}^j)| +  |\rd_t^r \rd_x^\alp (\tilde k_i{}^j - t^{-1} \kappa_i{}^j)| ) =O(\min\{t^{-1+\ve}, t^{-1+\ve-2p_j+2p_i}\}),
\end{equation}
then in fact the convergence condition \eqref{uniqueness.condition.2} \emph{follows as a consequence}. In fact, in this case both $g$ and $\tilde{g}$ have the leading asymptotics given by an approximate solution $g^{\bf [n]}$ for large $n$ (see Remark~\ref{rmk:uniqueness.1}), which then implies \eqref{uniqueness.condition.2}.  {This can be proven by revisiting the argument for constructing the approximate solutions in Theorem~\ref{thm:parametrix}. We omit the details.}
\end{remark}

Finally, we  {state} our main \emph{regularity} theorem.  {We remark that initially our proof of the existence theorem (Theorem~\ref{mainthm}) only constructs a solution with finite regularity. In order to obtain smoothness, we need an additional argument which relies on the uniqueness result (Theorem~\ref{thm:uniq}); see Section~\ref{sec:intro.uniqueness.regularity}.}
\begin{theorem}[Smoothness of solutions]\label{thm:smooth}
Given the assumptions of Theorem~\ref{mainthm}, there is a \underline{smooth} solution to the Einstein vacuum equations \eqref{eq:EVE} of the form \eqref{metricansatz} in $(0,T]\times \mathbb T^3$, for a $T>0$ depending on $c_{ij},p_i$, which satisfies \eqref{eq:a.limit}.
\end{theorem}

 {In the remainder of the introduction, we will briefly discuss the ideas of the proof (\textbf{Section~\ref{sec:ideas}}) and some related works (\textbf{Section~\ref{sec:related.works}}).}

\subsection{Ideas of the proof}\label{sec:ideas}

\subsubsection{Fuchsian analysis of a model wave equation}\label{sec:ideas.Fuchsian} 

As far as the singularity is concerned, our basic strategy (which is quite standard, see for instance \cite{sK2007}) 
can be most easily explained by a model semilinear equation.

Consider the following nonlinear wave equation
\begin{equation}\label{eq:intro.model}
\Box_g \phi = (\rd_t \phi)^2
\end{equation}
on a Kasner spacetime \eqref{eq:Kasner} with constants $p_1< p_2< p_3<1$ satisfying $\sum_{i=1}^3 p_i = \sum_{i=1}^3 p_i^2 = 1$. (Note that the structure of the nonlinear terms plays no role here,  and the nonlinearity $(\rd_t\phi)^2$ is chosen here for its simplicity.)

The analogue of our main result in this setup would be to construct \emph{bounded} solutions to the nonlinear model equation \eqref{eq:intro.model}. However, the results of  \cite{aAgFaF2018} imply that even for the linear wave equation, generic data on say, $\{t=1\}$, give rise to solutions that blow up as $O(\log \f 1t)$ as $t\to 0^+$. Thus, in order to obtain bounded solutions to \eqref{eq:intro.model}, the solution that we build has to be special. This is achieved by imposing the leading order behavior of $\phi(t,x) = \phi_0(x) + \mbox{error}$, where $\phi_0(x)$ is a prescribed smooth function which is the limit of $\phi(t,x)$ as $t\to 0^+$. In fact, we build our solution as $\phi(t,x) = \sum_{j=0}^n \phi_j(t,x) + \phi^{(d)}$, where $\phi_j$ are increasingly precise approximations of $\phi$, and $\phi^{(d)}$ is determined by the condition $\lim_{t\to 0} \phi^{(d)} = 0$.

Our strategy contains two steps:
\begin{enumerate}
\item (Approximate solution) It is easy to first build an \emph{approximate solution} by stipulating an ansatz $\phi^{\bf [n]}(t,x) = \sum_{j=0}^n \phi_j(t,x)$, where 
\begin{itemize}
\item $\phi_0(t,x) = \phi_0(x)$ is the prescribed leading order behavior, 
\item $\phi_j$ obeys the better estimates $|\rd_x^\alp \phi_j(t,x)|\lesssim_{\alp,j} t^{j\ve}$, and 
\item $| \rd_x^\alp \{ \Box_g\phi^{\bf [n]}(t,x) - (\rd_t\phi^{\bf [n]})^2(t,x) \}| \lesssim_{\alp,n} t^{-2+(n+1)\ve}$.
\end{itemize}
 This expansion can simply be obtained inductively by solving \eqref{eq:intro.model} iteratively as an ODE in $t$. Here, we have the flexibility to carry out the expansion to an arbitrary order $n$ so as to achieve an arbitrarily good (in terms of the $t$-rates as $t\to 0^+$) approximation to a solution to \eqref{eq:intro.model}.

With\underline{out} \emph{analyticity}, however, one can\underline{not} hope to show that this series \emph{converges}. Instead we perform energy estimates for the error.
\item (Energy estimates) First notice that for an energy defined by
$$
\mathcal E(\tau):= \sum_{|\bt|\leq 4} \int_{\{t=\tau \} } (|\rd_t \rd_j^\bt\phi|^2 + \sum_{i=1}^3 t^{-2p_i} |\rd_i \rd_j^\bt\phi|^2) \,\ud x,
$$
it is easy to obtain an estimate of the form (e.g.~with $C_0=2$)
$$ \f{d}{dt} \mathcal E(t)  \leq \f{C_0}{t} \mathcal E(t) + C_1(\mathcal E(t))^2.$$

The issue is with the borderline singular term $\f{C_0}{t} \mathcal E(t)$, which cannot be treated by Gr\"onwall's inequality (since $\limsup_{t\to 0^+} t^{-{C_0}} \mathcal E(t) = +\infty$). Nevertheless, this is where the approximation constructed in the previous step becomes useful: instead of controlling the full solution $\phi$, we bound the difference quantity $\phi^{(d)}:=\phi - \phi^{\bf [n]}$, which for $n$ sufficiently large
\begin{itemize}
\item can be made to approach $0$ with a fast polynomial rate as $t\to 0^+$, and 
\item satisfies an inhomogeneous nonlinear wave equation where the inhomogeneity also $\to 0$ with a fast polynomial rate.
\end{itemize}

Define now an energy $\mathcal E^{(d)}$ with $\phi$ replaced by $\phi^{(d)}$. For any large $N\in \mathbb N$, we can find $n\in \mathbb N$ large enough (corresponding to a good enough approximation) such that under appropriate bootstrap assumptions, 
$$ \f{d}{dt} \mathcal E^{(d)}(t) \leq (\f{C_0}{t} + \f{C_n}{t^{1-\ve}})\mathcal E^{(d)}(t)  + C_n t^N,$$
where $C_n$ may depend on $n$, but importantly, the constant $C_0$ in the borderline term is \underline{in}dependent of $n$. The inhomogeneous $C_n t^N$ term arises from the fact that $\phi^{(d)}$ satisfies an inhomogeneous equation, and $N$ can be arbitrarily chosen as long as $n$ is also taken to be large. Thus, we obtain an estimate
$$ \f{d}{dt} (t^{-N} \mathcal E^{(d)}(t)) + \f N t (t^{-N} \mathcal E^{(d)}) \leq (\f{C_0}{t} + \f{C_n}{t^{1-\ve}})(t^{-N} \mathcal E^{(d)}(t))  + C_n.$$
Recall now moreover that for $n$ sufficiently large we have $\lim_{t\to 0^+} (t^{-N} \mathcal E^{(d)}(t)) = 0$. Moreover, first choosing $N$ large (by taking $n$ large) and then taking $t$ small (depending on $n$), it follows that $\f N t (t^{-N} \mathcal E^{(d)})$ on the LHS dominates $(\f{C_0}{t} + \f{C_n}{t^{1-\ve}})(t^{-N} \mathcal E^{(d)}(t))$ on the RHS. This gives an estimate for $t^{-N} \mathcal E^{(d)}(t)$.
\end{enumerate}

Once such energy estimates can be proven for the error $\phi^{(d)}$, we can in fact deduce \emph{existence} of solutions as follows. Choosing a sequence $t_I \to 0^+$, we solve for a sequence of solutions $\{\phi_I\}_{i=0}^\infty$ to \eqref{eq:intro.model} with $(\phi_I, \rd_t \phi_I)\restriction_{t=t_I} = (\phi^{\bf [n]}, \rd_t \phi^{\bf [n]})\restriction_{t=t_I}$. The energy estimates above allows us to show that $\{\phi_I\}_{i=0}^\infty$ can be solved in $[t_I,T] \times \mathbb T^3$ for uniform $T>0$ and that there is a limit which solves  \eqref{eq:intro.model} in $(0,T] \times \mathbb T^3$.

\subsubsection{Construction of solutions to the Einstein vacuum equations in synchronized coordinates} While the Fuchsian analysis is quite robust, we must also address the quasilinear, tensorial nature, as well as the gauge invariance, of the Einstein equations.

If one were to prescribe a wave-coordinate-type gauge, then the construction of the approximate solution will be algebraically very complex. Instead, we consider a system of synchronized coordinates, i.e.~we impose that the metric takes the form
\begin{equation}\label{eq:intro.metric.gauge}
^{(4)}g = -\ud t^2 + g_{ij}\, \ud x^i\, \ud x^j = -\ud t^2 + t^{2p_{\max\{i,j\}}} a_{ij}\,\ud x^i \, \ud x^j.
\end{equation}
This gauge captures important anisotropic features of Kasner-like singularities. In particular, assuming that the $a_{ij}$'s are $C^2$ up to $\{t = 0\}$, we know that $|g_{ij}| \sim t^{2p_{\max\{i,j\}}}$, $|(g^{-1})^{ij}| \sim t^{- 2p_{\min\{i,j\}}}$; and importantly (see Lemma~\ref{lem:Ricci}) that
\begin{equation}\label{eq:intro.Ric.bd}
|Ric_i{ }^j (g)| \sim t^{-2+\ve},\quad | Ric(g) |_g \sim t^{-2+\ve}.
\end{equation} 

In such a gauge, the construction of an approximate solution becomes more tractable. The difficulty, however, is shifted to the estimates for the error terms. Indeed, even when no singularities are present, it is a priori not clear that the Einstein vacuum equations are hyperbolic in the gauge \eqref{eq:intro.metric.gauge}; see discussions in Section~\ref{sec:EEinsyncoord}..

\subsubsection{Constructing approximate solutions}

Following ideas laid out in Section~\ref{sec:ideas.Fuchsian}, we first construct approximate solutions and then use energy estimates to obtain actual solutions to the Einstein vacuum equation. In order to construct approximate solutions, the first step is to solve a system of first order \emph{evolutionary equations}. The evolutionary equations will be treated as a system of ODEs in $t$ (compare Step~1 in Section~\ref{sec:ideas.Fuchsian}). In order to close the ODE analysis, we crucially rely on the bounds \eqref{eq:intro.Ric.bd}, which show that the spatial Ricci curvature is slightly better than critical, but we also need to additionally make use of the structure of the full system. We outline some main points here:
\begin{itemize}
\item The main difficulty in solving the system of ODEs is that there are many \emph{borderline} terms, i.e.~linear terms with $O(t^{-1})$ coefficients. It turns out that these terms have a \emph{reductive} structure. By this we mean that we can consider different components in a sequence of steps. In each step, there are two type of terms with a borderline $O(t^{-1})$ coefficient with the following properties.
\begin{itemize}
\item One type can be handled by introducing an integration factor. The integration factor gives a power of $t$ which is consistent with the initial conditions that we impose.
\item Another type of terms with borderline coefficients involve \emph{only terms which have been controlled in previous steps}. 
\item Any other linear terms must have a coefficient that is better, at least $O(t^{-1+\ve})$. 
\end{itemize}
Such a structure is important both in estimating the metric components (Lemmas~\ref{lem:an.well.defined}, \ref{lem:an-an-1}) and the components of the (approximate) second fundamental form (Lemmas~\ref{lem:kn.gen.bounds}, \ref{lem:kn-kn-1}).
\item In anticipation of the energy estimates needed to construct an actual solution, we also need to treat different components on different footing in the ODE analysis. An example of this is that while for $i\leq j$, we only prove that $(k^{\bf [n]})_i{ }^j = O(t^{-1})$; for $i> j$, we need a better estimate and the improvement we need depends on the precise $i$, $j$ under consideration; see Lemma~\ref{lem:kn.gen.bounds}. Such estimates can be traced back to \eqref{eq:intro.Ric.bd}, but also require the precise structure of the system.
\item Another technical difficulty is that the variable $k^{\bf [n]}$ we work with is only approximately the second fundamental form.
\end{itemize}	

The evolutionary equations solved in the first step roughly asserts that the spacetime Ricci curvature components $Ric(^{(4)}g^{\bf [n]})_i{}^j$ vanish with a very fast rate. Our second step is then to show that
\begin{itemize} 
\item $k^{\bf [n]}$ is asymptotically (as $t\to 0^{+}$) approximately the second fundamental form of the constant-$t$ hypersurfaces, and 
\item all other spacetime Ricci curvature components also vanish sufficiently fast as $t\to 0^+$.
\end{itemize} 
Both of these are achieved again by ODE analysis. For the first point, we need again a reductive structure, which is similar to the type used for the evolutionary equations. For the second point, the constraints as manifested both in the conditions on the Kasner exponents and asymptotic constraint equations \eqref{eq:cij} play a crucial role. See already Lemmas~\ref{lem:D.est}--\ref{lem:k.II} and Proposition~\ref{prop:approx.G}.

\subsubsection{Energy estimates in synchronized coordinates}\label{sec:EEinsyncoord}

It is a priori unclear that under a gauge condition as in \eqref{eq:intro.metric.gauge}, the metric components themselves satisfy any hyperbolic system. The main new ingredient is to consider a ``wave-type equation'' satisfied by the second fundamental form $k_i{ }^j$ of the spatial slice. 
Since this is already new for a local existence problem \emph{without singularities}, we will indicate the ingredients needed only for a local existence result for \emph{regular} data, i.e.~for this subsubsection suppose we are given geometric data $(\Sigma, g, k)$ satisfying the (usual) constraint equations and we explain how to construct a spacetime solution to the Einstein vacuum equations in the gauge \eqref{eq:intro.metric.gauge}.

Assuming that a metric of the form \eqref{eq:intro.metric.gauge} obeys the Einstein vacuum equations, we can deduce that the second fundamental form $k_i{ }^j$ obeys the following system of second order equations:
\begin{equation}\label{eq:k.eqn.intro}
\rd_t^2 k_i{ }^j = \Delta_g k_i{ }^j -\nab_i\nab^j k_\ell{ }^\ell + (k\star k\star k)_i{ }^j + (\rd_t k \star k)_i{ }^j,
\end{equation}
where $k\star k\star k$ and $\rd_t k \star k$ are nonlinear terms to be specified in \eqref{kstark} in Section~\ref{sec:derivation.of.equations}.

Notice that \eqref{eq:k.eqn.intro} is \underline{not} actually a wave equation, due to the term $\nab_i \nab^j k_\ell{ }^\ell$ on the RHS. The key is that the trace of $k$, i.e.~$k_\ell{ }^\ell$ in fact can be proven to have additional regularity if we further use the Einstein vacuum equations. First, the Einstein vacuum equations imply that 
$$\rd_t k_\ell{ }^\ell = |k|^2.$$
Now we consider $h = k_\ell{ }^\ell$ to be a separate variable and consider the coupled system for $(g,h,k)$:
\begin{equation}\label{eq:intro.hyp.sys}
\begin{split}
\rd_t h =&\:  |k|^2,\\
\rd_t^2 k_i{ }^j =&\: \Delta_g k_i{ }^j -\nab_i\nab^j h + (k\star k\star k)_i{ }^j + (\rd_t k \star k)_i{ }^j,\\
\rd_t g_{ij} = & - 2k_{i}{ }^{\ell} g_{j\ell}.
\end{split}
\end{equation}
(This system must hold for $h = k_\ell{ }^\ell$ if the Einstein vacuum equations are satisfied.) We then attempt to solve \eqref{eq:intro.hyp.sys} with initial data where $(g_{ij}, k_i{ }^j)$ is as given, $h= k_\ell{ }^\ell$ and $\rd_t k_i{ }^j = Ric(g)_i{ }^j + k_\ell{ }^\ell k_i{ }^j$ (which is completely determined by the geometric data).

The apparent difficulty in solving \eqref{eq:intro.hyp.sys} is a potential loss of derivatives. 
For instance, energy estimates for the second equation requires two derivative of $h$ and give only first-derivative estimates for $k$. The first equation, however, does not seem to give two derivatives for $h$ if we only have one derivative for $k$. A similar issue arises for $g$ and $k$ when we consider commutators for the second equation.

This can nevertheless be resolved by a renormalization together with elliptic estimates. As an example, we illustrate how to obtain second derivative estimates for $h$ when only controlling one derivative of $k$. Commute the first equation with $\Delta_g$ so that we have, up to error terms,
$$\rd_t \Delta_g h = 2 k_i{ }^j \Delta_g k_j{ }^i + \ldots.$$
The idea now is to use the second equation in \eqref{eq:intro.hyp.sys} so that we obtain
$$\rd_t (\Delta_g h - 2 k_i{ }^j \rd_t k_j{ }^i) = 2 k_i{ }^j (-\rd_t^2 +\Delta_g) k_j^i + \ldots = \ldots.$$
This allows us to control $\Delta_g h$ even only controlling one derivative of $k$. The other second derivatives of $h$ can then be bounded by $\Delta_g h$ using elliptic estimates. This avoids the loss of derivatives.

Standard energy estimates together with this renormalization/elliptic estimates trick indeed give a solution to \eqref{eq:intro.hyp.sys}. Furthermore, the choice of initial data and the structure of the equations allow one to propagate the symmetry of $g_{ij}$ and $g_{i\ell} k_j{ }^\ell$. Using moreover the Hamiltonian constraint $R(g)-|k|^2+(\text{tr}k)^2=0$, it can be shown a posteriori that $h = k_\ell{ }^\ell$. In particular, we also have that $\rd_t k_\ell{ }^\ell = |k|^2$, which implies that $Ric(^{(4)}g)_{tt} = 0$.

Finally, we need to upgrade the existence result \eqref{eq:intro.hyp.sys} to a bona fide existence result of solutions to the Einstein vacuum equations in the gauge \eqref{eq:intro.metric.gauge}, i.e.~we need to show that all the Ricci components vanish (in addition to $Ric(^{(4)}g)_{tt}$). For this purpose, first note that (after accounting for symmetries) the second equation in \eqref{eq:intro.hyp.sys} gives a system of $6$ first order homogeneous equations in $Ric_i{ }^j(^{(4)}g)$ and $Ric_{ti}(^{(4)} g)$. At the same time, three of the (contracted) second Bianchi equations give another $3$ first order homogeneous equations in $Ric_i{ }^j(^{(4)}g)$ and $Ric_{ti}(^{(4)} g)$. (The fourth equation is redundant, and does not give us extra information.) It turns out that these $9$ equations form  {a coupled system of wave-transport equations (see \eqref{eq:G.high.2} and \eqref{eq:wave.eqn.for.G}). This wave-transport equations is similar in structure to \eqref{eq:intro.hyp.sys}, and can also be treated using energy estimates together with renormalization/elliptic estimates. Moreover,} the momentum constraint  and the choice of initial data, when solving \eqref{eq:intro.hyp.sys}, together, guarantee that $Ric_i{ }^j(^{(4)}g)$ and $Ric_{ti}(^{(4)} g)$ are initially vanishing.  Combining all these we obtain that 
$Ric_i{ }^j(^{(4)}g) =0$ and $Ric_{ti}(^{(4)} g) =0$ everywhere, implying that the constructed solution to \eqref{eq:intro.hyp.sys} indeed obeys the Einstein vacuum equations.

Obviously, in our setting, we need to handle simultaneously the existence theory and the fact that the metric becomes singular as $t\to 0^+$. For this we combine the ideas here and Section~\ref{sec:ideas.Fuchsian}. A few technical issues arise. For instance, the Kasner-type geometry dictate that we do not have uniform control of the isoperimetric constants as $t\to 0^{+}$. Some care is therefore needed in the application of Sobolev embedding; in particular we need to be careful which terms are to be put in $L^2$/$L^\i$ type spaces. Finally, we note that the Fuchsian ideas in Section~\ref{sec:ideas.Fuchsian} are used not only in solving the system \eqref{eq:intro.hyp.sys}, but are also used in verifying that the solution to \eqref{eq:intro.hyp.sys} is indeed a solution to the Einstein vacuum equations. 

\subsubsection{Uniqueness and regularity}\label{sec:intro.uniqueness.regularity}

To prove uniqueness, we again rely on the wave equation satisfied by the second fundamental form, and perform $t$-weighted energy estimates in a similar way as proving existence. The only subtlety here is that we must impose that the metrics converge to each other sufficiently fast as $t\to 0^+$ in order to close the estimate (cf.~the statement of Theorem~\ref{thm:uniq}).

Finally, we prove higher regularity \emph{relying on the uniqueness result}. The issue at stake here is that for each additional derivative we try to control, the estimate in terms of $t$ worsens by one power. Thus, the approximation we choose has to be successively better for higher and higher derivatives. We then redo the construction of solutions for better and better choices of the approximations. The uniqueness result ensures that we have in fact constructed the \emph{same} solution, thus showing that the already constructed solution has arbitrarily high derivative bounds.

\subsection{Related works}\label{sec:related.works}

\subsubsection{ {Fuchsian} constructions of singular spacetimes}\label{sec:construction}

Many works have been carried out to construct AVTD singularities in $(3+1)$-dimensional vacuum spacetimes. All previous works assume either symmetry or analyticity (or both). The symmetry classes are typically chosen so that AVTD singularities are expected to be stable \emph{within} that class. We give a sample of such results, but refer the reader also to the references therein for further details.

\textbf{Gowdy symmetry.} AVTD singularities in (unpolarized) Gowdy symmetry was first constructed by Kichenassamy and Rendall \cite{sKaR1998} in the analytic category, in part based on the formal expansion carried out in \cite{GrMo1993}. A similar analysis was carried out by Rendall without the analyticity assumption in \cite{aR2000b}. See also \cite{fS2002} for more general topologies, and \cite{eAfBjIpL2017} for a treatment in generalized wave gauges.

\textbf{Polarized $\mathbb T^2$ symmetry.} Analytic AVTD singularities under polarized $\mathbb T^2$ symmetry were first constructed in \cite{sKjI1999}; analyticity was later removed in \cite{eAfBjIpL2013}.

\textbf{$\mathbb U(1)$ polarized or half-polarized symmetry.} Analytic solutions with AVTD behavior in polarized or half-polarized symmetry with $\mathbb T^3$ topology were constructed by Isenberg--Moncrief in \cite{jIvM2002}. That for more general topology was later carried out in \cite{yCBjIvM2004}.

\textbf{Beyond $(3+1)$-dimensional vacuum spacetimes.} The first construction of \emph{analytic} solutions with AVTD behavior \emph{without symmetries} was carried out in \cite{lAaR2001}, albeit \emph{not for the Einstein vacuum equations}. Indeed, the construction in \cite{lAaR2001} was for the Einstein--scalar field or Einstein--stiff fluid system. An important difference is that in the presence of a scalar field or stiff fluid, one expects AVTD singularities to be stable \cite{vBiK1972, jB1978}. A similar stability phenomenon is expected to occur in vacuum for spacetime dimensions $\geq 11$ \cite{DeHeSp1985}. Correspondingly, there is a construction of AVTD singularities for high dimensional vacuum (and more general) solutions in \cite{DaHeReWe2002}. See also Section~\ref{sec:stable}.

\textbf{Analytic singular spacetimes without symmetry assumptions.} All the works above concern regimes (either in symmetry classes or with matter, or in high dimensions) which at least heuristically should generically have AVTD behavior near the spacelike singularity. In a recent work of Klinger \cite{pK2015}, \emph{analytic} vacuum AVTD spacetimes \emph{with no symmetry assumptions} have been constructed. The work \cite{pK2015} can be viewed as similar to our result except for requiring the analyticity assumption and some additional inequalities on the Kasner exponents $p_i$'s. (These additional inequalities were used in \cite{pK2015} to apply a black-box Fuchsian theorem.)

\textbf{Asymptotically Schwarzschild singularity up to a singular $2$-sphere.} Finally, we mention the work \cite{Fo2016} of the first author, who constructed a class of spacetimes approaching the Schwarzschild black hole singularity. The construction requires no symmetry or analyticity assumptions. While it does not include a full spacelike singular hypersurface, the construction does include a spacelike singular $2$-sphere.

\subsubsection{Stable singularities in general relativity}\label{sec:stable} By ``function-counting'' arguments (cf.~Remark~\ref{rmk:counting}), the class of spacetimes we construct are not expected to be stable. For the vacuum equations in $(3+1)$ dimensions, the only known stable singularities are in fact null; see \cite{aOeF1996, jL2013, mDjL2017}.  These singularities are in stark contrast with the AVTD ones, which are spacelike.

 {As we already mentioned in Section~\ref{sec:construction}}, it has been suggested that in the presence of a scalar field or stiff fluid \cite{vBiK1972, jB1978}, or in the vacuum case in spacetime dimensions $\geq 11$ \cite{DeHeSp1985}, there is an open set of initial data which give rise to asymptotically Kasner-like singularities.  {It is also for this reason that in these settings, the construction of spacelike singularities with AVTD behavior is simpler.}


Spectacular progress has recently been made which indeed proves \emph{stability} of  {spacelike singularities in the aforementioned settings}. In the case of Einstein--scalar field or Einstein--stiff fluid, this was carried out in the breakthrough work by Rodnianski--Speck \cite{RoSp2018a, RoSp2018b} and later generalized by Speck \cite{Sp2018}. In the case of high dimensions, assuming spacetime dimensions $\geq 39$, Rodnianski--Speck has recently also constructed a class of stable spacelike singularities  {in vacuum} \cite{jSiR2018c}. (Note that the remarkable works of Rodnianski--Speck do not cover the whole regimes in \cite{vBiK1972, jB1978, DeHeSp1985}. Whether \emph{all} of the solutions discussed in \cite{vBiK1972, jB1978, DeHeSp1985} are stable remains an open problem.)

Very recently, the first author and Alexakis considered the stability problem for the Schwarzschild singularity \cite{sAgF2020}. Unlike the settings studied by Rodnianski--Speck, the Schwarzschild singularity is \emph{unstable}, but nonetheless it was shown in \cite{sAgF2020} to be stable \emph{within the class of polarized axisymmetric perturbations}.

\subsubsection{Strong cosmic censorship}

The understanding of AVTD singularities played an important role in understanding the strong cosmic censorship conjecture, at least under Gowdy symmetry. 

The strong cosmic censorship conjecture has first been resolved in the polarized Gowdy case in \cite{pCjIvM1990}. The work relies in particular on \cite{jIvM1990}, in which AVTD singularities in this setting were studied.

The more general case of the strong cosmic censorship conjecture in \emph{unpolarized} Gowdy symmetry turned out to be significantly more difficult in view of the so-called ``spikes''. This has been treated in the seminal work of Ringstr\"om  {\cite{hR2009b} (see also \cite{hR2008a}). Here, a form of asymptotic velocity term domination has been established \cite{hR2006} and plays an important role.}

It should again be stressed that outside symmetry classes (Gowdy, polarized $\mathbb T^2$, polarized $\mathbb U(1)$, etc.), AVTD singularities are most likely not generic, and the role of the study of AVTD singularities in the ultimate resolution of strong cosmic censorship conjecture is quite unclear. 

\subsubsection{Numerical works}

 {A discussion of the large number of related numerical works will take us too far afield. For this we will refer the reader to \cite{bB2002} and the many references therein.}

\subsubsection{Linear wave equations on singular spacetimes}

A closely related thread of works concerns solving the \emph{linear} wave equation on a spacetime with a spacelike singularity, including Kasner, FLRW and Schwarzschild. See for instance \cite{sKpS1981, pAaR2010, hR2017, gFjB2018, aAgFaF2018, hR2019, aB2019, pGjNjS2019}.

\subsubsection{Einstein equations in transport coordinates}

At the heart of our approach is the ability to perform energy estimates in the gauge \eqref{eq:intro.metric.gauge}, corresponding to a choice of coordinates such that $(t,x^i)$ are all transported by the unit normal to the spacelike hypersurfaces $\{t=\mathrm{constant}\}$; recall Section~\ref{sec:EEinsyncoord}. We highlight previous works where smooth solutions to the Einstein equations are constructed in gauge where the spatial $x^i$ coordinates are transported, i.e~the metric takes the form
\begin{equation}\label{eq:spatially.transport}
-\alp^2 \ud t^2 + g_{ij} \ud x^i \ud x^j.
\end{equation}

The first is the work of Rodnianski--Speck \cite{RoSp2018a, RoSp2018b} (in which they constructed stable spacelike singularities; see discussions in Section~\ref{sec:stable} above), where $\alp$ is determined by stipulating that each constant-$t$ hypersurface has constant mean curvature. See also \cite{gFjS2019} for a different approach in handling this gauge. (Constant mean curvature foliations, but without spatially transported coordinates, have been previously used. See for instance \cite{lAvM2003}, which used spatially harmonic instead of spatially transported coordinates.)

The second is the work of Choquet-Bruhat--Ruggeri \cite{yCBtR1983}, in which the authors consider the spacetime metric of the form \eqref{eq:spatially.transport} and impose the condition $\alp = \sqrt{\f{\det g}{\det e}}$, where $e$ is some arbitrary but fixed (i.e.~$t$-independent) Riemannian metric. They show that in such a gauge, the Einstein equations are hyperbolic.

\subsection{Outline of the paper}  {We end the introduction with an outline of the remainder of the paper.

The first part of the existence proof will be to construct an \emph{approximate} solution. This will be carried out in \textbf{Section~\ref{sec:parametrix}}, where we give the construction and show that \emph{evolutionary} equations are approximately satisfied. In \textbf{Section~\ref{sec:approx.constraints}} we then show that the \emph{constraint} equations are also approximately satisfied.

In Section~\ref{sec:actual} and \ref{subsec:vanEVE} we then construct an \emph{actual} solution, thus completing the proof of Theorem~\ref{mainthm}. This will be carried out in two steps: in \textbf{Section~\ref{sec:actual}} we will solve an appropriate system of reduced equations, then in \textbf{Section~\ref{subsec:vanEVE}} we show that the solutions to the reduced equations that we have constructed in fact obey the Einstein vacuum equations.

Finally, in \textbf{Section~\ref{sec:uni}}, we end with the proofs of \emph{uniqueness} (Theorem~\ref{thm:uniq}) and \emph{smoothness} (Theorem~\ref{thm:smooth}).}

\subsection{Acknowledgements} G.F.~would like to thank Lars Andersson, Satyanad Kichenassamy, Jacques Smulevici and Jared Speck for useful communications.

G.F. is supported by the ERC grant 714408 GEOWAKI, under the European Union's Horizon 2020 research and innovation program.  J.L.~gratefully acknowledges the support of the NSF grant DMS-1709458.

\section{Construction of an approximate solution}\label{sec:parametrix}

We work under the assumptions of Theorem~\ref{mainthm}. In particular, we fix $p_i$ and $c_{ij}$ to be as in Theorem~\ref{mainthm}. 

Unless explicitly stated otherwise, all the implicit constants (given either in the $\ls$ or the big-O or the $\cdot \leq C \cdot$ notation) that we have in our arguments, from now on, may depend on $p_i$ and $c_{ij}$. Many estimates in this section will involve an $n\in \mathbb N$ or a multi-index $\alp$. Unless otherwise stated, all constants may depend also on $n$ and $\alp$.

 Our goal in this section is to construct an approximate solution, i.e.~we will construct inductively a metric ${^{(4)}g}^{\bf [n]}$ ($n \in \mathbb N\cup \{0\}$), which takes the form \eqref{metricansatz}, but with $a_{ij}^{\bf [n]}$ in place of $a_{ij}$; as well as an approximate second fundamental form $(k^{\bf [n]})_i{ }^j$. These $a_{ij}^{\bf [n]}$ are constructed so that $\lim_{t\to 0^+} a^{[{\bf n}]}_{ij}(t,x) = c_{ij}(x)$.  We will moreover show that the pairs $(g^{\bf [n]},\,k^{\bf [n]})$ we construct indeed form an approximate solution to the evolution equation, i.e.~as $n$ becomes larger,  $\rd_t (k^{\bf [n]})_i{ }^j - Ric(g^{\bf [n]})_i{ }^j - (k^{\bf [n]})_\ell{ }^\ell(k^{\bf [n]})_i{ }^j$ tends to $0$ faster as $t\to 0^+$; see already Theorem~\ref{thm:parametrix}.

 Unless otherwise stated, we will also be using the Einstein summation convention for repeated indices, with lower case Latin indices running through $1,2,3$. It should be noted that sometimes we will still write out the sum explicitly in situations that confusion might arise (e.g.~when one has factors of $t^{p_{\max\{i,j\}}}$).

 {\textbf{Definition of ${^{(4)}g}^{\bf [n]}$ and $k^{\bf [n]}$.}} Define ${^{(4)}g}^{\bf [0]}$ by setting 
\begin{equation}\label{eq:a0.def}
a_{ij}^{\bf [0]} = c_{ij}.
\end{equation} 
Now given $g^{\bf [n-1]}$, $n\in\mathbb{N}$  {(and assuming that it is a Riemannian metric on $(0,t_n]\times \mathbb T^3$)}, define $k^{\bf [n]}$ by
\begin{align}
\rd_t (k^{\bf [n]})_i{ }^j =&\: Ric( {g^{\bf [n-1]}})_i{ }^j + (k^{\bf [n]})_\ell{ }^\ell(k^{\bf [n]})_i{ }^j, \label{eq:k.transport}
\end{align}
subject to the following condition at $t=0$:
\begin{equation}\label{eq:initial.parametrix.k}
|(k^{\bf [n]})_i{ }^j -  {t^{-1}\kappa_i{ }^j}|(t,x) = O(t^{-1+\ve}),
\end{equation}
where $ {\kappa}$ is defined by $ {\kappa}_i{ }^i = - p_i$ ( {for every $i=1,2,3$,} without summing), $ {\kappa}_1{ }^2 = (p_1 - p_2 )\f{c_{12}}{c_{22}}$, $ {\kappa}_2{ }^3 = (p_2 - p_3)\f{c_{23}}{c_{33}}$, $ {\kappa}_1{ }^3 = (-p_1+p_2)\f{c_{12} c_{23}}{c_{22} c_{33}} + (p_1-p_3)\f{c_{13}}{c_{33}}$ and $ {\kappa}_i{ }^j = 0$ if $i>j$;
and given $k^{\bf [n]}$, $n \in \mathbb N$, define $g^{\bf [n]}$ by
\begin{align}
\label{eq:g.transport}\rd_t g^{\bf [n]}_{ij} = &\: -   (k^{\bf [n]})_i{ }^\ell g^{\bf [n]}_{\ell j} -  (k^{\bf [n]})_j{ }^\ell g^{\bf [n]}_{\ell i}, 
\end{align}
subject to the following condition at $t=0$:
\begin{equation}\label{eq:initial.parametrix.g}
|a^{\bf [n]}_{ij} - c_{ij}|(t,x) = O(t^\ve),
\end{equation}
where we recall that $a^{\bf [n]}$ is related to $g^{\bf [n]}$ via \eqref{metricansatz}.
It readily follows from \eqref{eq:g.transport} that the inverse metric $(g^{[\bf n]})^{-1}$ satisfies the equation:
\begin{equation}\label{eq:gn-1.transport}
\rd_t ((g^{[\bf n]})^{-1})^{ij} =  (k^{[\bf n]})_\ell{ }^j ((g^{[\bf n]})^{-1})^{i\ell} + (k^{\bf[n]})_\ell{ }^i ((g^{[\bf n]})^{-1})^{j\ell}.
\end{equation}

 {Our goal in this section is to establish} the properties of the above sequences $\{g^{\bf [n]}\}_{n=0}^{+\infty}$, $\{k^{\bf [n]}\}_{n=1}^{+\infty}$ given in the following theorem:
\begin{theorem}\label{thm:parametrix}
Let $p_i$ and $c_{ij}$ be as in the Theorem \ref{mainthm}. Define 
$$ {\ve = \min\{\min_x (p_3 - p_2)(x), \min_x (1-p_3)(x)\} > 0}.$$
Then for $n\in \mathbb N$, there exist $t_n>0$  {(depending on $p_i$, $c_{ij}$ and $n$),} a  {smooth Lorentzian} metric ${^{(4)}g}^{\bf [n]}$ and a $(1,1)$-tensor $(\kn)_i{ }^j$ on $(0,t_n]\times \mathbb T^3$ such that the following holds. 
\begin{enumerate}
\item ${^{(4)}g} ^{\bf [n]}$ takes the following form for some smooth functions $a_{ij}^{\bf [n]}:(0,t_n]\times \mathbb T^3\to \mathbb R$ (symmetric in $i,j$):
\begin{align*}
\begin{split} 
{^{(4)}g}^{\bf [n]} =  -\ud t^2 + \sum_{i,j=1}^3 g^{\bf [n]}_{ij} \ud x^i\, \ud x^j
=&\: - \ud t^2 +  \sum_{i,j=1}^3 a^{\bf [n]}_{ij} t^{2p_{\max\{i,j\}}} \ud x^i\, \ud x^j. 
\end{split} 
\end{align*}
\item  {(Convergence to initial data)} For every multi-index $\alp$, every $i$, $j$ and every $n \in \mathbb N$, the functions $a_{ij}^{\bf [n]}$ and $(k^{\bf [n]})_i{ }^j$ satisfy
\begin{equation}\label{eq:main.parametrix.a.bd}
\sup_{x\in \mathbb T^3} |\rd_x^\alp (a_{ij}^{\bf [n]}(t,x) - c_{ij}(x))|\leq C_{\alp,n}t^{\ve}, 
\end{equation}
\begin{equation}\label{eq:main.parametrix.k.bd}
\sup_{x\in \mathbb T^3}| \rd_x^\alp [(k^{\bf [n]})_i{}^j(t,x)-t^{-1}\kappa_i{ }^j(x)]| \leq C_{\alp,n} \min\{ t^{-1+\ve},\, t^{-1+\ve-2p_j+2p_i} \},
\end{equation}
for some $C_{\alp,n}>0$ depending on $p_i$, $c_{ij}$, in addition to $\alp$ and $n$. (Recall the definition of $\kappa_i{}^j$ immediately after \eqref{eq:initial.parametrix.k}.)
\item  (Estimates for spatial curvature) For every multi-index $\alp$, every $i$, $j$ and every $n \in \mathbb N$, the spatial Ricci curvature satisfies
\begin{equation}\label{eq:main.parametrix.Ric.bd}
\sup_{x\in \mathbb T^3} \sum_{r=0}^1 t^r |\rd_x^\alp \rd_t^r Ric(g^{\bf [n]})_i{}^j(t,x) | \leq C_{\alp,n} \min\{ t^{-2+\ve},\, t^{-2+\ve-2p_j+2p_i} \},
\end{equation}
for some $C_{\alp,n}>0$ depending on $p_i$, $c_{ij}$, in addition to $\alp$ and $n$. 
\item  ($k^{\bf [n]}$ is an approximate second fundamental form) For every multi-index $\alp$, every $i$, $j$ and every $n \in \mathbb N$,  
\begin{equation}\label{eq:main.parametrix.k.and.II.1}
\sum_{r=0}^{2} t^r | \rd_x^{\alp} \rd_t^r (2(k^{\bf [n]})_i{ }^j +  (g^{\bf [n]})^{j\ell} \rd_t g^{\bf [n]}_{i\ell}) |(t,x) \leq C_{n,\alp} t^{-1+(n+1)\ve},
\end{equation}
for some $C_{\alp,n}>0$ depending on $p_i$, $c_{ij}$, in addition to $\alp$ and $n$.
\item  {(Evolution equations approximately satisfied) For every multi-index $\alp$, t}he tensors $(k^{[\bf n]})_i{}^j$, $g^{[\bf n]}_{ij}$ also satisfy
\begin{equation}\label{eq:main.parametrix.Ric}
\begin{split}
\sup_{x\in \mathbb T^3} \sum_{r=0}^1 t^r  \left|\rd_x^\alp \rd_t^r \left(\rd_t (k^{\bf [n]})_i{ }^j - Ric(g^{\bf [n]})_i{ }^j - (k^{\bf [n]})_\ell{ }^\ell(k^{\bf [n]})_i{ }^j\right)\right|(t,x) \leq &\:C_{\alp,n} t^{-2+(n+1)\ve},
\end{split}
\end{equation}
for some $C_{\alp,n}>0$ depending on $p_i$, $c_{ij}$, in addition to $\alp$ and $n$.
\end{enumerate}
\end{theorem}

\begin{remark}\label{rem:tvarepsilon}
All the $\ve$ in the error terms in Theorem~\ref{thm:parametrix} can be improved almost to $2\ve$ (or exactly to $2\ve$ if we allow some powers of $\log t$ in the error terms). Some estimates can even be further sharpened. We will be content with the weaker estimates for the sake of simplicity of the exposition.
\end{remark}

\begin{remark}\label{rem:varepsilon}
The definition of $\varepsilon$, together with conditions (2)--(3) in Theorem \ref{mainthm}, imply that 
\begin{align}\label{pi.range}
-\frac{1}{3}\leq p_1\leq -\varepsilon,\qquad
\varepsilon\leq p_2\leq \f 23,\qquad
\f 23 \leq p_3\leq 1-\varepsilon,\qquad p_3-p_2\ge\varepsilon.
\end{align}
This can be easily checked by using the following parametric form of the Kasner exponents $p_1,p_2,p_3$:
\begin{align}\label{eq:pi.param}
p_1=\frac{-u}{1+u+u^2},\qquad p_2=\frac{1+u}{1+u+u^2},\qquad p_3=\frac{u(1+u)}{1+u+u^2},\qquad u\in[1,+\infty),
\end{align}
which is valid at each point $x\in\mathbb{T}^3$, $u=u(x)$.
\end{remark}

In the rest of the section, we will prove  {Theorem~\ref{thm:parametrix}}; see the conclusion of the proof at the end of the section. (In particular, in the course of the proof, it can be seen that $g^{\bf [n]}$ and $k^{\bf [n]}$ are well-defined.)

\subsection{ {Estimates for $g^{\bf[0]}$}}

\begin{lemma}\label{lem:inverse.0}
There exists $t_0>0$ (depending on $c_{ij}$ and $p_i$) such that the following are true for $(t,x)\in (0,t_0 {]} \times \mathbb T^3$:
\begin{enumerate}
\item The determinant of $\det g^{\bf [0]}$ satisfies, for some $C>0$ (depending on $c_{ij}$ and $p_i$),
\begin{align}\label{detg0}
| \det g^{\bf [0]}(t,x) - c_{11} c_{22} c_{33} t^{2}|\leq C t^{2+\ve}.
\end{align}
\item  {The eigenvalues $\lambda_1 \leq \lambda_2 \leq \lambda_3$ of $g^{\bf [0]}$ satisfy, for some $C>0$ (depending on $c_{ij}$ and $p_i$),
$$|\lambda_i - t^{2p_i} c_{ii}| \leq C t^{2p_i + \ve}.$$
In particular, choosing $t_0$ smaller if necessary, $g^{\bf [0]}$ is a Lorentzian metric on $(0,t_0]\times \mathbb T^3$.}
\item For every multi-index $\alp$, the inverse metric $(g^{\bf [0]})^{-1}$ satisfies, for some $C_\alp>0$ (depending on $\alp$, $c_{ij}$ and $p_i$),
\begin{align}\label{invg0}
(g^{\bf [0]})^{-1} = 
\begin{bmatrix}
\f {t^{-2p_1}}{c_{11}} & -\f{c_{12} c_{33} t^{-2p_1}}{c_{11} c_{22} c_{33}} & \f{(c_{12} c_{23}-c_{13}c_{22}) t^{-2p_1}}{c_{11} c_{22} c_{33}} \\
-\f{c_{12} c_{33} t^{-2p_1}}{c_{11} c_{22} c_{33}} & \f {t^{-2p_2}}{c_{22}} & \f{(c_{12} c_{13}-c_{11}c_{23}) t^{-2p_2}}{c_{11} c_{22} c_{33}}\\
\f{(c_{12} c_{23}-c_{13}c_{22}) t^{-2p_1}}{c_{11} c_{22} c_{33}} & \f{(c_{12} c_{13}-c_{11}c_{23}) t^{-2p_2}}{c_{11} c_{22} c_{33}} & \f {t^{-2p_3}}{c_{33}}
\end{bmatrix} + (g^{-1})^{\bf [0]}_{error},
\end{align}
where $|\rd_x^\alp ((g^{-1})^{\bf [0]}_{error})^{ij}|\leq C t^{-2p_{\min\{i,j\}}+\ve}$.
\end{enumerate}
\end{lemma}
\begin{proof}
This is a simple computation and the proof is omitted.
\end{proof}

It will be convenient to define 
also 
\begin{equation}\label{eq:k0.def}
(k^{\bf [0]})_i{ }^j:= -\f 12 ((g^{\bf [0]})^{-1})^{j\ell}\rd_t g_{i\ell}^{\bf [0]}.
\end{equation}
The following lemma gives an estimate for $(k^{\bf [0]})_i{ }^j$.
\begin{lemma}\label{lem:k0}
 {F}or every multi-index $\alp$, there exists $C_{\alp}>0$ (depending on $\alp$, in addition to $c_{ij}$ and $p$) such that the following estimate holds for all $(t,x)\in (0,t_0  {]}\times \mathbb T^3$:
$$|\rd_x^\alp[(k^{\bf [0]})_i{ }^j -  {t^{-1}\kappa}_i{ }^j]|(t,x)\leq C_\alp   t^{-1+\ve}.$$
\end{lemma}
\begin{proof}
By the definition of $g^{\bf [0]}$, it is easy to see that 
$$ \rd_t g^{\bf [0]} = 
\begin{bmatrix}
2p_1 t^{2p_1-1} c_{11} & 2p_2 t^{2p_2-1} c_{12} & 2p_3 t^{2p_3-1} c_{13} \\
2p_2 t^{2p_2-1} c_{12} & 2p_2 t^{2p_2-1} c_{22} & 2p_3 t^{2p_3-1} c_{23}\\
2p_3 t^{2p_3-1} c_{13} & 2p_3 t^{2p_3-1} c_{23} & 2p_3 t^{2p_3-1} c_{33}
\end{bmatrix} + (\rd_t g)^{\bf [0]}_{error},
$$
where $|\rd_x^\alp((\rd_t g)^{\bf [0]}_{error})_{ij}|\leq C_\alp t^{2p_{\max\{i,j\}}-1+\ve}$. Recalling that $(k^{\bf [0]})_i{ }^j:= -\f 12 ((g^{\bf [0]})^{-1})^{j\ell}\rd_t g_{i\ell}^{\bf [0]}$, the conclusion of the lemma can be achieved by combining the above computation with Lemma~\ref{lem:inverse.0}. \qedhere
\end{proof}

The next lemma estimates the Ricci curvature of a general metric $g = \sum_{i,j=1}^3 a_{ij} t^{2p_{\max\{i,j\}}} \ud x^i\, \ud x^j$ when $a_{ij}$ satisfies some basic bounds. This in particular gives an estimate for $Ric(g^{\bf [0]})_i{}^j$.
\begin{lemma}\label{lem:Ricci}
Suppose ${^{(4)}g}$ is a metric on $(0,T]\times \mathbb T^3$ taking the form ${^{(4)}g} = - \ud t^2 + \sum_{i,j=1}^3 a_{ij} t^{2p_{\max\{i,j\}}} \ud x^i\, \ud x^j$, where $a_{ij}$ are smooth, symmetric and obey the estimates 
$$|\rd_x^\alp a_{ij}|(t,x) \leq C_{\alp},\qquad |\rd_x^\alp\partial_t a_{ij}|(t,x) \leq C_{\alp}t^{-1+\varepsilon},$$
for some $C_\alp >0$. 

Then for every multi-index $\alp$, there exists $C'_{\alp}>0$ (depending on $C_{\alp}$, in addition to $c_{ij}$ and $p$) such that the following estimate holds for all $(t,x)\in (0, {T]} \times \mathbb T^3$:
\begin{equation}\label{eq:Ricci.without.logs}
\sum_{r=0}^1 t^r |\rd_x^\alp \rd_t^r Ric(g)_i{ }^j|(t,x)\leq C'_{\alp} { \min\{ t^{-2+\ve},\, t^{-2+\ve -2p_j + 2p_i} \} }.
\end{equation}
In fact, the following slightly stronger estimate holds:
\begin{equation}\label{eq:Ricci.improved}
 \sum_{r=0}^1 t^r|\rd_x^\alp \rd_t^r Ric(g)_i{ }^j|(t,x) \leq C'_\alp \min\{t^{-2+2\ve} |\log t|^{2+|\alp|},\, t^{-2 + 2\ve - 2p_j + 2p_i}|\log t|^{2+|\alp|} \}.
\end{equation}
\end{lemma}
\begin{proof}
Clearly \eqref{eq:Ricci.improved} implies \eqref{eq:Ricci.without.logs}; from now on we focus on the proof of \eqref{eq:Ricci.improved}.

For notational convenience, in this proof we write $g^{ab} = (g^{-1})^{ab}$.

Here is the basic observation. For a pairing $g^{ab} \rd_c g_{ae}$ (note the one contracted index), we have 
$$g^{ab} \rd_c g_{ae} = O(t^{-2p_{\min\{a,b\}} + 2p_{\max\{a,e\}}}) |\log t| \leq O(|\log t|).$$
Similarly,
$$(\rd_c g^{ab})( \rd_{d} g_{ae}),\, g^{ab} \rd^2_{cd} g_{ae} = O(|\log t|^2).$$
So in order to give an estimate for the Ricci curvature, we will find pairs of $g^{-1}$ and derivatives of $g$ which \emph{share at least one index}.

To make the algebraic structure clear, we will focus on proving the estimate with $|\alp|=0$ and $r =0$ in Steps~1 and 2, and then indicate the necessary changes in Steps~3 and 4.

\pfstep{Step~1: Proof of the upper bound $t^{-2+\ve}$}
We recall the formula for the Ricci curvature:
\begin{align}\label{Ricciformula}
Ric(g)_i{ }^j= g^{ab}\partial_i \Gamma_{ab}^j- g^{ab}\partial_a \Gamma_{bi}^j + g^{ab} \Gamma_{ic}^j \Gamma_{ab}^c - g^{ab} \Gamma_{ac}^j \Gamma_{ib}^c
\end{align} 
and that for the Christoffel symbols
\begin{align}\label{Chformula}
\Gamma_{ab}^c=\frac{1}{2} g^{c\ell} (\partial_a g_{b\ell}+\partial_b g_{a\ell} - \partial_\ell g_{ab}) {.}
\end{align}
Hence, we notice that every term in \eqref{Ricciformula} has either of the forms
\begin{align}\label{Ricciterms}
g^{ab}\partial_{\ell_1} [g^{\ell_2\ell_3}\partial_{\ell_4} g_{\ell_5\ell_6}],\qquad g^{ab}g^{\ell_1\ell_2}\partial_{\ell_3} g_{\ell_4\ell_5}g^{\ell_6\ell_7}\partial_{\ell_8} g_{\ell_9\ell_{10}}
\end{align}
where among the $\ell_i$'s there is an upper $j$ and a lower $i$ index, while the rest are contractions among themselves and with respect to $a,b$. 

For the first kind of terms in \eqref{Ricciterms}, using Lemma \ref{lem:inverse.0}, we notice that they are of order 
\begin{align}
|g^{ab}\partial_{\ell_1} [g^{\ell_2\ell_3}\partial_{\ell_4} g_{\ell_5\ell_6}]|\lesssim |\log t|^2t^{-2p_{\min\{a,b\}}-2p_{\min\{\ell_2,\ell_3\}}+2p_{\max\{\ell_5,\ell_6\}}},
\end{align}
where the pair $\{\ell_5,\ell_6\}$ contains at least one of the indices $a,b,\ell_2,\ell_3$. Hence, we have either $-2p_{\min\{a,b\}}+2p_{\max\{\ell_5,\ell_6\}}\ge0$ or $-2p_{\min\{\ell_2,\ell_3\}}+2p_{\max\{\ell_5,\ell_6\}}\ge0$, leaving 
\begin{align*}
|g^{ab}\partial_{\ell_1} [g^{\ell_2\ell_3}\partial_{\ell_4} g_{\ell_5\ell_6}]|\lesssim |\log t|^2t^{-2p_\ell}\lesssim |\log t|^2 t^{-2+2\varepsilon},
\end{align*}
for some $\ell$. On the other hand, the second term in \eqref{Ricciterms} satisfies:
\begin{align*}
&|g^{ab}g^{\ell_1\ell_2}\partial_{\ell_3} g_{\ell_4\ell_5}g^{\ell_6\ell_7}\partial_{\ell_8} g_{\ell_9\ell_{10}}|\\
\lesssim&\, |\log t|^2t^{ -2p_{\min\{a,b\}}-2p_{\min\{\ell_1,\ell_2\}}+2p_{\max\{\ell_4,\ell_5\}}-2p_{\min\{\ell_6,\ell_7\}}+2p_{\max\{\ell_9,\ell_{10}\}}  } {,}
\end{align*}
where at least three from the indices $a,b,\ell_1,\ell_2,\ell_6,\ell_7$ are contracted against three of the indices $\ell_4,\ell_5,\ell_9,\ell_{10}$. This implies that at least two pairs of exponents having opposite signs, among $$\{-2p_{\min\{a,b\}},-2p_{\min\{\ell_1,\ell_2\}},2p_{\max\{\ell_4,\ell_5\}},-2p_{\min\{\ell_6,\ell_7\}},2p_{\max\{\ell_9,\ell_{10}\}}\},$$ yield non-negative sums, thus, leaving only 
\begin{align*}
|g^{ab}g^{\ell_1\ell_2}\partial_{\ell_3} g_{\ell_4\ell_5}g^{\ell_6\ell_7}\partial_{\ell_8} g_{\ell_9\ell_{10}}|\lesssim |\log t|^2t^{-2p_\ell}\lesssim |\log t|^2 t^{-2+2\varepsilon}.
\end{align*}

\pfstep{Step~2: Proof of the upper bound $t^{-2+\ve-2p_j+2p_i}$}
We now move on to prove the improved estimates when $i>j$ (when $i\leq j$ the desired estimate follows from that proven in Step~1). As we are now familiar with this type of argument, let us just consider the contribution from the second type of term in \eqref{Ricciterms} (the first type of terms can be treated similarly). We now separate out the factor of $g^{j\ell}$ (which gives a contribution of at worst of $O(t^{-2p_j})$), i.e.~we write
$$g^{j b} g^{\ell_1\ell_2}\partial_{\ell_3} g_{\ell_4\ell_5}g^{\ell_6\ell_7}\partial_{\ell_8} g_{\ell_9\ell_{10}},$$
where exactly one of the $\ell_m$ is $b$ and exactly one of the $\ell_m$ is $i$. It is easy to check that at least one of the following must hold:
\begin{itemize}
\item After relabelling $g^{\ell_1\ell_2}\partial_{\ell_3} g_{\ell_4\ell_5}g^{\ell_6\ell_7}\partial_{\ell_8} g_{\ell_9\ell_{10}} = g^{\ell_1 c}\partial_{\ell_3} g_{\ell_4c}g^{\ell_6 d}\partial_{\ell_8} g_{\ell_9 d}$, so that by our basic observation  $g^{\ell_1\ell_2}\partial_{\ell_3} g_{\ell_4\ell_5}g^{\ell_6\ell_7}\partial_{\ell_8} g_{\ell_9\ell_{10}}  = O(|\log t|^2)$. As a result, the whole term contributes $O(t^{-2p_j}|\log t|^2)$, which is better than $O(t^{-2+\ve-2p_j+2p_i})$.
\item After relabelling, we have one of the following
$$g^{j b} g^{ac} g^{df} \rd_a g_{df} \rd_c g_{bi}, \quad g^{jb} g^{ac} g^{df} \rd_a g_{db} \rd_c g_{fi}.$$
For the first term, after noting $g^{df} \rd_a g_{df} = O(|\log t|)$, $g^{ac} = O(t^{-2+2\ve})$, $g^{jb} = O(t^{-2p_j})$ and $\rd_c g_{bi} = O(t^{2p_i}|\log t|)$, we have $g^{j b} g^{ac} g^{df} \rd_a g_{df} \rd_c g_{bi} = O(t^{-2+2\ve -2p_j+2p_i}|\log t|^2)\leq O(t^{-2+\ve-2p_j+2p_i})$. For the second term, note that $g^{df} \rd_a g_{db} = O(|\log t|)$, $g^{ac} = O(t^{-2+2\ve })$, $g^{jb} = O(t^{-2p_j})$ and $\rd_c g_{bi} = O(t^{2p_i}|\log t|)$, which then again gives the desired estimate.
\end{itemize}

 \pfstep{Step~3: Higher derivative bounds} It is easy to see that after differentiating by $\rd_x^\alp$, we at worst pick up additional powers of $|\log t|^{|\alp|}$, we then obtain the desired estimate also for higher derivatives of $Ric(g)_i{ }^j$.

 \pfstep{Step~4: Time derivative}  For $\partial_x^\alpha \rd_t Ric(g)_i{}^j$, the argument is almost identical. Indeed, exploiting the form of the metric and using the estimate for $\partial_x^\alpha\partial_t a_{ij}$, we notice that $\partial_tg_{ij}=O(t^{2p_{\max\{i,j\}}-1}), \partial_tg^{ij}=O(t^{-2p_{\min\{i,j\}}-1})$ and similar behaviors for their spatial derivatives (up to logarithms). Hence, a power of $t^{-1}$ can be factored out, leaving terms with factors that behave as in the previous steps.  
This completes the proof of the lemma.
\end{proof}

\subsection{ {Estimates for $k^{\bf [n]}$}}\label{sec:parametrix.induction}

\begin{lemma}\label{lem:transport}
Consider the nonlinear transport equation 
$$\rd_t u = f + \f {u^2}{t^2},$$
where $f:(0,1)\times \mathbb T^3\to \mathbb R$ is a function such that $|f|(t,x)\ls t^\delta$ for some $\de>0$. Then there exist $t_*\in (0,1)$ and a unique solution $u:(0,t_*)\times \mathbb T^3\to \mathbb R$ such that $|u|(t,x)\ls t^{1+\de}$.

Assuming moreover that $|\rd_x^\alp f|(t,x)\ls_\alp t^{\delta}$. It also follows that $|\rd_x^\alp u|\ls_\alp t^{1+\de}$.
\end{lemma}
\begin{proof}
This is proven by a standard Picard iteration, with some  {extra} care tracing the $t$ dependence; we omit the details.
\end{proof}

\begin{lemma}\label{lem:kn.gen.bounds}
Suppose the following holds for some $N\geq 1$: there exists $t_{N-1}>0$ such that for every $0\leq n \leq N-1$ and every multi-index $\alp$,  $g^{\bf [n]}$ satisfies the following estimate for some $C_{\alp,n}>0$ (depending on $\alp$, $n$, in addition to $c_{ij}$ and $p_i$) for all $(t,x)\in (0,t_{N-1})\times \mathbb T^3$:
\begin{equation}\label{eq:a.diff.with.data.assump}
|\rd_x^{\alp} (a^{\bf [n]}_{ij} -c_{ij}) |(t,x) \leq C_{\alp, n} t^{\ve}.
\end{equation}

Then, there exists $t_N\in (0,t_{N-1})$ sufficiently small  {such that for every $1\leq n\leq N$} and  {every} multi-index $\alp$, the  {following holds for all $(t,x)\in (0,t_{N-1})\times \mathbb T^3$ for some $C'_{\alp,n}>0$ (depending on $\alp$, $n$, in addition to $c_{ij}$ and $p_i$):}
$$|\rd_x^\alp[(k^{\bf [n]})_i{ }^j - t^{-1}\kappa_i{ }^j]|(t,x)\leq C_{\alp,n}' \min\{ t^{-1+\ve},\, t^{-1+\ve-2p_j+2p_i} \}.$$
\end{lemma}
\begin{proof}

The key difficulty in solving \eqref{eq:k.transport} is that there are borderline terms with $O(t^{-1})$ coefficients so that we cannot directly apply Gr\"onwall's lemma. One can nevertheless analyze the precise structure of the equations.

\pfstep{Step~1: Solving an auxiliary system} We first solve an auxiliary system 
\begin{equation}\label{eq:aux.1}
\begin{cases}
\rd_t  {h^{\bf [n]}} = R(g^{\bf [n-1]}) +  {(h^{\bf [n]})^2} \\
\rd_t  {(k^{\bf [n]})}_i{ }^j = Ric(g^{\bf [n-1]})_i{ }^j + { h^{\bf [n]}(k^{\bf [n]})_i{ }^j} 
\end{cases}.
\end{equation}

The first equation in \eqref{eq:aux.1} can be rearranged to
\begin{equation}\label{eq:aux.1.2}
\rd_t [t^2( {h^{\bf [n]}}+ \f 1t)] = t^2 R(g^{\bf [n-1]}) + t^2( {h^{\bf [n]}}+\f 1t)^2.
\end{equation}
Using the bound $|R(g^{\bf [n-1]})|\ls t^{-2+\ve}$ from the assumptions on $g^{\bf [n-1]}$ together with Lemma~\ref{lem:Ricci}, \eqref{eq:aux.1.2} can be solved using Lemma~\ref{lem:transport} with  {$h^{\bf [n]}$} satisfying
\begin{equation}\label{eq:h1.est}
|\rd_x^\alp ( {h^{\bf [n]}} + \f 1t)|\ls t^{-1+\ve}.
\end{equation}

Now the second equation in \eqref{eq:aux.1} can be rearranged to
$$\rd_t [t(k^{\bf [n]})_i{ }^j] = t Ric(g^{\bf [n-1]})_i{ }^j + (h^{\bf [n]}+ \f 1t) t(k^{\bf [n]})_i{ }^j .$$
Using \eqref{eq:h1.est}, Gr\"onwall's inequality  {and the estimate in Lemma~\ref{lem:Ricci}}, it follows that there is a unique solution $(k^{\bf [ {n}]})_i{ }^j$ that obeys the initial condition \eqref{eq:initial.parametrix.k} and satisfies 
\begin{equation}\label{eq:k1.est}
|\rd_x^\alp [(k^{\bf [n]})_i{ }^j -  {t^{-1} \kappa}_i{ }^j]|\ls  {\min\{t^{-1+\ve}, \, t^{-1+\ve -2p_j + 2p_i}\}}.
\end{equation}

\pfstep{Step~2: Finishing the argument} Now that we have solved \eqref{eq:aux.1}  {and obtained estimates} \eqref{eq:h1.est} and \eqref{eq:k1.est}, in order to conclude the argument, it suffices to show that in fact $h^{\bf [n]} = (k^{\bf [n]})_\ell{ }^\ell$. To this end, it suffices to note that
$$\rd_t [t^2((k^{\bf [n]})_\ell{ }^\ell+ \f 1t)] = t^2 R(g^{\bf [n-1]}) + t^2((k^{\bf [n]})_\ell{ }^\ell+\f 1t)^2.$$
Hence, comparing this equation with \eqref{eq:aux.1.2}, we obtain $h^{\bf [n]} = (k^{\bf [n]})_\ell{ }^\ell$ by the uniqueness statement in Lemma~\ref{lem:transport}. \qedhere
\end{proof}

\begin{lemma}\label{lem:kn-kn-1}
Suppose  {the following holds for some $N\geq 2$:} there exists $t_{N-1}>0$ such that for every  {$1\leq n \leq N-1$} and every multi-index $\alp$,  $g^{\bf [n]}$ satisfies the following estimate for some $C_{\alp,n}>0$ (depending on $\alp$, $n$, in addition to $c_{ij}$ and $p_i$) for all $(t,x)\in (0,t_{N-1})\times \mathbb T^3$:
\begin{equation}\label{eq:a.diff.assump}
|\rd_x^{\alp} (a^{\bf [n]}_{ij} - a^{\bf [n-1]}_{ij})|(t,x) \leq C_{\alp, n} t^{n \ve}.
\end{equation}

Then,  {taking} $t_N\in (0,t_{N-1})$  {smaller (compared to Lemma~\ref{lem:kn.gen.bounds}) if necessary,}  {for every $2\leq n\leq N$} and  {every} multi-index $\alp$, the  {following holds for all $(t,x)\in (0,t_{N-1})\times \mathbb T^3$ for some $C'_{\alp,n}>0$ (depending on $\alp$, $n$, in addition to $c_{ij}$ and $p_i$):}
\begin{equation}\label{eq:k.diff.est}
|\rd_x^\alp [(k^{\bf [n]})_i{ }^j - (k^{\bf [n-1]})_i{ }^j]|(t,x) \leq C'_{\alp, n} t^{-1+n\ve}.
\end{equation}
\end{lemma}
\begin{proof}

\pfstep{Step~1: Estimates on the Ricci curvature} The estimate \eqref{eq:a.diff.assump} implies that 
\begin{equation}\label{eq:Ric.diff.est.more.precise}
|\rd_x^\alp [Ric(g^{\bf [n]})_i{ }^j - Ric(g^{\bf [n-1]})_i{ }^j]|(t,x)\ls t^{-2+(n+2)\ve}|\log t|^{2+|\alp|}
\end{equation}
for every $0\leq n \leq N-1$. Indeed, arguing as in the proof of Lemma~\ref{lem:Ricci}, we notice that the difference of the $\rd_x^\alp$ derivative of the Ricci curvatures can be bounded by the differences $a^{\bf [n]}_{ij} - a^{\bf [n-1]}_{ij}$ (and their spatial derivatives), multiplied by a term that is controlled by $t^{-2+2\varepsilon}|\log t|^{2+|\alp|}$. In particular, \eqref{eq:Ric.diff.est.more.precise} implies
\begin{equation}\label{eq:Ric.diff.est}
|\rd_x^\alp [Ric(g^{\bf [n]})_i{ }^j - Ric(g^{\bf [n-1]})_i{ }^j]|(t,x)\ls t^{-2+(n+1)\ve}.
\end{equation}

\pfstep{Step~2: Estimates on $(k^{\bf [n]})_i{ }^j$} The assumption \eqref{eq:a.diff.assump} implies the assumption of Lemma~\ref{lem:kn.gen.bounds} holds. Hence by Lemma~\ref{lem:kn.gen.bounds},
\begin{equation}\label{eq:kn.est}
|\rd_x^\alp [(k^{\bf [n]})_i{ }^j -  {t^{-1}\kappa}_i{ }^j]|(t,x)\ls t^{-1+\ve}
\end{equation}
for every $2\leq n \leq N$.

In particular, since (by definition) $t^{-1}\kappa_i{ }^i = \f 1t$, \eqref{eq:kn.est} implies that
\begin{equation}\label{eq:hn.est}
|\rd_x^\alp [(k^{\bf [n]})_i{ }^i +\f 1t]|(t,x)\ls t^{-1+\ve}.
\end{equation}

\pfstep{Step~3: Estimates on the difference $(k^{\bf [n]})_i{ }^j- (k^{\bf [n-1]})_i{ }^j$} Using \eqref{eq:k.transport}, we obtain, for $2\leq n \leq N$, that
\begin{equation}\label{eq:kn.diff}
\begin{split}
&\: \rd_t [(k^{\bf [n]})_i{ }^j - (k^{\bf [n-1]})_i{ }^j] \\
= &\: Ric(g^{\bf [n-1]})_i{ }^j -Ric(g^{\bf [n-2]})_i{ }^j + [(k^{\bf [n]})_\ell{ }^\ell- (k^{\bf [n-1]})_\ell{ }^\ell](k^{\bf [n]})_i{ }^j + (k^{\bf [n-1]})_\ell{ }^\ell[(k^{\bf [n]})_i{ }^j - (k^{\bf [n-1]})_i{ }^j].
\end{split}
\end{equation}

It turns out to be useful to first control the trace of $k^{\bf [n]} - k^{\bf [n-1]}$. Taking the trace of \eqref{eq:kn.diff}, we obtain
$$\rd_t ((k^{\bf [n]})_i{ }^i - (k^{\bf [n-1]})_i{ }^i) = R(g^{\bf [n-1]}) - R(g^{\bf [n-2]}) + ((k^{\bf [n]})_i{ }^i + (k^{\bf [n-1]})_i{ }^i)((k^{\bf [n]})_i{ }^i - (k^{\bf [n-1]})_i{ }^i).$$
This implies
$$\rd_t [t^2((k^{\bf [n]})_i{ }^i - (k^{\bf [n-1]})_i{ }^i)] = t^2(R(g^{\bf [n-1]}) - R(g^{\bf [n-2]})) + ((k^{\bf [n]})_i{ }^i + (k^{\bf [n-1]})_i{ }^i+ \f 2 t)t^2((k^{\bf [n]})_i{ }^i - (k^{\bf [n-1]})_i{ }^i).$$
By \eqref{eq:Ric.diff.est} in Step~1, the estimate \eqref{eq:hn.est} in Step~2, the condition \eqref{eq:initial.parametrix.k} and Gr\"onwall's inequality, it easily follows that
\begin{equation}\label{eq:kn.diff.est}
|\rd_x^\alp (k^{\bf [n]} - k^{\bf [n-1]})_i{ }^i|(t,x) \ls t^{-1+n\ve}
\end{equation}
 for every $2\leq n\leq N$.

We now return to \eqref{eq:kn.diff}, which we rewrite as follows. 
\begin{equation*}
\begin{split}
&\: \rd_t [t((k^{\bf [n]}- k^{\bf [n-1]})_i{ }^j)] \\
= &\: t(Ric(g^{\bf [n-1]})_i{ }^j -Ric(g^{\bf [n-2]})_i{ }^j) + t(k^{\bf [n]}- k^{\bf [n-1]})_\ell{ }^\ell (k^{\bf [n]})_i{ }^j + [(k^{\bf [n-1]})_\ell{ }^\ell+ \f 1t]t(k^{\bf [n]} - k^{\bf [n-1]})_i{ }^j .
\end{split}
\end{equation*}
By \eqref{eq:Ric.diff.est} in Step~1, the estimates \eqref{eq:kn.est} and \eqref{eq:hn.est} in Step~2, the estimate \eqref{eq:kn.diff.est} that we just proved, the condition \eqref{eq:initial.parametrix.k} and Gr\"onwall's inequality, we obtain
$$|\rd_x^\alp (k^{\bf [n]} - k^{\bf [n-1]})_i{ }^j|(t,x) \ls t^{-1+n\ve}$$
for every $2\leq n\leq N$, which is what we want to prove. \qedhere
\end{proof}

\subsection{Estimates for $a^{\bf [n]}_{ij}$}

\begin{lemma}\label{lem:a.eq}
For $n \in \mathbb N$ and $g^{\bf [n]}_{ij}$ defined by \eqref{eq:g.transport}--\eqref{eq:initial.parametrix.g}, the corresponding $a^{\bf [n]}_{ij}$ obeys the equation
\begin{align}\label{eq:a.n}
\begin{split}
 \rd_t a^{\bf [n]}_{ij} 
=&\: - \sum_\ell t^{2 p_{\max\{ \ell,j\}} - 2 p_{\max\{i,j\}}} \bigg( (k^{\bf [n]} - k^{\bf [0]})_i{}^\ell a_{\ell j}^{\bf [n]}
 +  (k^{\bf [0]})_i{}^\ell(a^{\bf [n]}_{\ell j} -c_{\ell j}) \bigg)\\
 &\: - \sum_\ell t^{2 p_{\max\{ \ell,i\}} - 2 p_{\max\{i,j\}}} \bigg( (k^{\bf [n]} - k^{\bf [0]})_j{}^\ell a_{\ell i}^{\bf [n]}
 + (k^{\bf [0]})_j{}^\ell(a^{\bf [n]}_{\ell i} -c_{\ell i}) \bigg)
- {\f{2p_{\max\{i,j\}}}{t}}(a^{\bf [n]}_{ij}-c_{ij}),
\end{split}
\end{align}
where $k^{\bf [0]}$ is as defined in \eqref{eq:k0.def}.
\end{lemma}
\begin{proof}
By  {\eqref{metricansatz} and \eqref{eq:g.transport}}, with repeated indices \underline{not} summed, we have  {on the one hand}
$$\rd_t g^{\bf [n]}_{ij} = 2 p_{\max\{i,j\}} t^{2 p_{\max\{i,j\}}-1} a^{\bf [n]}_{ij} + t^{2 p_{\max\{i,j\}}} \rd_t a^{\bf [n]}_{ij},$$
 {and on the other hand}
$$ { \rd_t g^{\bf [n]}_{ij} =  - \sum_{\ell} (k^{\bf [n]})_i{}^\ell t^{2p_{\max\{\ell,j\}}} a^{\bf [n]}_{\ell j} - \sum_{\ell} (k^{\bf [n]})_j{}^\ell t^{2p_{\max\{\ell,i\}}} a^{\bf [n]}_{\ell i} }.$$
 {Similarly, by \eqref{metricansatz}, \eqref{eq:a0.def} and \eqref{eq:k0.def},}
$$-  {\sum_{\ell} (k^{\bf [0]})_i{}^\ell t^{2p_{\max\{\ell,j\}}}c_{\ell j} -\sum_{\ell} (k^{\bf [0]})_j{}^\ell t^{2p_{\max\{\ell,i\}}}c_{\ell i} }=\partial_tg^{\bf [0]}_{ij}= 2  {p_{\max\{i,j\}} t^{2 p_{\max\{i,j\}}-1}} c_{ij}.$$

Therefore,  {we obtain}
\begin{align*}
 {t^{2 p_{\max\{i,j\}}} } \rd_t a^{\bf [n]}_{ij} = &\:  {\rd_t g^{\bf [n]}_{ij} - 2 p_{\max\{i,j\}} t^{2 p_{\max\{i,j\}}-1} a^{\bf [n]}_{ij} = \rd_t (g^{\bf [n]}_{ij} -g^{\bf [0]}_{ij}) - 2 p_{\max\{i,j\}} t^{2 p_{\max\{i,j\}}-1} (a^{\bf [n]}_{ij} - c_{ij})}\\
=&\: -  \sum_{\ell} (k^{\bf [0]})_i{}^\ell t^{2p_{\max\{\ell, j \}}}(a^{\bf [n]} -c)_{\ell j} -  \sum_{\ell} (k^{\bf [0]})_j{}^\ell t^{2p_{\max\{\ell, i\}}}(a^{\bf [n]} -c )_{\ell i} \\
&\: -  { \sum_{\ell} t^{2p_{\max\{\ell,j\}}} (k^{\bf [n]}-k^{\bf [0]})_i{}^\ell a_{\ell j}^{\bf [n]} - \sum_{\ell} t^{2p_{\max\{\ell,i\}}} (k^{\bf [n]}-k^{\bf [0]})_j{}^\ell a_{\ell i}^{\bf [n]}}\\
&\: -2  {p_{\max\{i,j\}} t^{2 p_{\max\{i,j\}}-1}} (a^{\bf [n]}_{ij}-c_{ij}).
\end{align*}
Canceling $ {t^{2 p_{\max\{i,j\}}} }$ on both sides, we obtain the desired equation.
\end{proof}

\begin{lemma}\label{lem:an.well.defined}
Suppose the following holds for some $N \geq 1$: there exists $t_{N}>0$ such that  {for every $1\leq n \leq N$} and every multi-index $\alp$, $k^{\bf [n]}$ satisfies the estimate for some $C_{\alp,n}>0$ (depending on $\alp$, $n$, in addition to $c_{ij}$ and $p_i$) for all $(t,x)\in (0,t_{N-1})\times \mathbb T^3$:
$$|\rd_x^\alp[(k^{\bf [n]})_i{ }^j - t^{-1}\kappa_i{ }^j]|(t,x)\leq C_{\alp,n} t^{-1+\ve}.$$

Then, after choosing $t_N>0$ smaller if necessary, $a^{\bf [n]}_{ij}(t,x)$ is well-defined and symmetric for all $(t,x)\in (0,t_{N}] \times \mathbb T^3$ and for every $1\leq n \leq N$. In addition, by reducing $t_N>0$ further, $g^{\bf [n]}_{ij}(t,x)$ is a Lorentzian metric.

Moreover, for every multi-index $\alp$ and $1\leq n \leq N$, there exists $C'_{\alp,n}>0$ such that
\begin{align}\label{eq:a-c.est}
|\rd_x^\alp (a^{\bf [n]}_{ij} - c_{ij})|(t,x) \leq C'_{\alp,n} t^{\ve},\qquad |\rd_x^\alp \partial_t a^{\bf [n]}_{ij}|(t,x) \leq C'_{\alp,n} t^{-1+\varepsilon}
\end{align}
for all $(t,x)\in (0,t_{N} ] \times \mathbb T^3$.
\end{lemma}
\begin{proof}

Clearly $\rd_t (a_{ij}^{\bf [n]} - a_{ji}^{\bf [n]}) = 0$. Moreover, at $\{ t = 0\}$, $a_{ij}^{\bf [n]} = c_{ij}$ which is symmetric. It follows that $a_{ij}^{\bf [n]}$ is symmetric.

Now given that $a_{ij}^{\bf [n]}$ is symmetric, we will only estimate the six components $\{a_{ij}^{\bf [n]}: i\leq j \}$. Using the equation in Lemma~\ref{lem:a.eq} and the bounds in Lemmas~\ref{lem:k0} and \ref{lem:kn.gen.bounds} (and implicitly using the symmetry of $a_{ij}^{\bf [n]}$ in the derivation), we obtain the following schematic equations:
\begin{align}
\rd_t (a^{\bf [n]}- c)_{33} 
=&\,  O(t^{-1+\varepsilon}) (a^{\bf [n]} - c) + O(t^{-1+\varepsilon}) a^{\bf [n]}, \label{eq:basic.eq.for.a-c.1}\\
\rd_t (a^{\bf [n]}- c)_{22} 
=&\,  O(t^{-1+\varepsilon}) (a^{\bf [n]} - c) + O(t^{-1+\varepsilon}) a^{\bf [n]}, \label{eq:basic.eq.for.a-c.2}\\
\rd_t (a^{\bf [n]}- c)_{11} 
=&\,  O(t^{-1+\varepsilon}) (a^{\bf [n]} - c) + O(t^{-1+\varepsilon}) a^{\bf [n]}, \label{eq:basic.eq.for.a-c.3}\\
\rd_t (a^{\bf [n]}- c)_{23} 
=&\,  \frac{p_2 - p_3}{t} (a^{\bf [n]} - c)_{23} -\frac{\kappa_2{}^3}{t} (a^{\bf [n]} - c)_{33} +O(t^{-1+\varepsilon}) (a^{\bf [n]} - c) + O(t^{-1+\varepsilon}) a^{\bf [n]}, \label{eq:basic.eq.for.a-c.4}\\
\rd_t (a^{\bf [n]}- c)_{12} 
=&\,  \frac{p_1 - p_2}{t} (a^{\bf [n]} - c)_{12} -\frac{\kappa_1{}^2}{t} (a^{\bf [n]} - c)_{22} +O(t^{-1+\varepsilon}) (a^{\bf [n]} - c) + O(t^{-1+\varepsilon}) a^{\bf [n]}, \label{eq:basic.eq.for.a-c.5}\\
\label{eq:basic.eq.for.a-c.6}\rd_t (a^{\bf [n]}- c)_{13} 
=&\, \frac{p_1-p_3}{t} (a^{\bf [n]} - c)_{13} - \frac{\kappa_1{}^2}{t}(a^{\bf [n]} - c)_{23} - \frac{\kappa_1{}^3}{t} (a^{\bf [n]} - c)_{33} \\
&\,+O(t^{-1+\varepsilon}) (a^{\bf [n]} - c) + O(t^{-1+\varepsilon}) a^{\bf [n]}. \notag
\end{align}
Here, we have used the schematic notation that when we write $(a^{\bf [n]} - c)$ or $a^{\bf [n]}$ without explicit indices, it can represent any component.

The key point is a reductive structure for terms with $O(t^{-1})$ coefficients: The diagonal $(a^{\bf [n]}- c)_{ii}$ terms do not see any terms with $O(t^{-1})$ coefficients on the right hand side. For the remaining terms, we make the observations that (1) the linear term has coefficients which is \emph{negative} and (2) by estimating the terms in the order as listed above, the only terms with $O(t^{-1})$ coefficients have already been estimated in the previous step.

Indeed, the first three equations (\eqref{eq:basic.eq.for.a-c.1}--\eqref{eq:basic.eq.for.a-c.3}) give
\begin{equation}\label{eq:a-c.est.1}
|(a^{\bf [n]}- c)_{33}|(t) + |(a^{\bf [n]}- c)_{22}|(t) + |(a^{\bf [n]}- c)_{11}|(t) \ls t^{\ve} \sup_{[0,t]}( |a^{\bf [n]} - c| + |c|),
\end{equation}
where we have used the initial condition \eqref{eq:initial.parametrix.g}.

Using the fourth and fifth equations (\eqref{eq:basic.eq.for.a-c.4}--\eqref{eq:basic.eq.for.a-c.5}) and plugging in \eqref{eq:a-c.est.1}, we obtain
\begin{equation}\label{eq:a-c.est.2}
\begin{split}
t^{p_3-p_2}|(a^{\bf [n]}- c)_{23}|(t) 
\ls &\: t^{p_3-p_2}|(a^{\bf [n]}- c)_{33}|(t) + t^{p_3-p_2+\ve} \sup_{[0,t]}( |a^{\bf [n]} - c| + |c|) \\
\ls &\: t^{p_3-p_2+\ve} \sup_{[0,t]}( |a^{\bf [n]} - c| + |c|)
\end{split}
\end{equation}
and
\begin{equation}\label{eq:a-c.est.3}
\begin{split}
 t^{p_2-p_1}|(a^{\bf [n]}- c)_{12}|(t) 
\ls &\: t^{p_2-p_1} |(a^{\bf [n]}- c)_{22}|(t) + t^{p_2-p_1+\ve} \sup_{[0,t]}( |a^{\bf [n]} - c| + |c|) \\
\ls &\: t^{p_2-p_1+\ve} \sup_{[0,t]}( |a^{\bf [n]} - c| + |c|).
\end{split}
\end{equation}
 {The estimates} \eqref{eq:a-c.est.2} and \eqref{eq:a-c.est.3} imply
\begin{equation}\label{eq:a-c.est.4}
|(a^{\bf [n]}- c)_{23}|(t) + |(a^{\bf [n]}- c)_{12}|(t) \ls  t^{\ve} \sup_{[0,t]}( |a^{\bf [n]} - c| + |c|).
\end{equation}
Finally, we consider the last equation, argue as above and plug in \eqref{eq:a-c.est.1} and \eqref{eq:a-c.est.4} to obtain
\begin{equation}\label{eq:a-c.est.5}
|(a^{\bf [n]}- c)_{13}|(t) \ls  t^{\ve} \sup_{[0,t]}( |a^{\bf [n]} - c| + |c|).
\end{equation}
Combining \eqref{eq:a-c.est.1}, \eqref{eq:a-c.est.4}, \eqref{eq:a-c.est.5}, and choosing $t_N$ to be sufficiently small, we obtain
$$\sup_{[0,t]} |a^{\bf [n]} - c| \ls t^{\ve} \sup |c| \ls t^{\ve}.$$
This proves that $a^{\bf [n]}_{ij}$ is well-defined and moreover shows the first inequality in \eqref{eq:a-c.est} in the case $|\alp| =0$. 

The second inequality in \eqref{eq:a-c.est} (that for $\rd_t a^{\bf [n]}_{ij}$) follows by applying the already derived bounds to the RHS of the system \eqref{eq:basic.eq.for.a-c.1}--\eqref{eq:basic.eq.for.a-c.6}.

We then obtain the desired higher order estimates by induction on $|\alpha|$. For example, differentiating the equation \eqref{eq:a.n} by $\partial_x^{\alp}$ for $|\alp| =1$, we may treat the zeroth order terms in the differences $a^{\bf [n]} - c$ as already estimated inhomogeneous terms and repeat the above argument. The same goes for $\partial_x^\alp$ with $|\alp| = 2$ etc. From this we deduce the estimate \eqref{eq:a.diff.est} in general. We omit the details. \qedhere

\end{proof}

\begin{lemma}\label{lem:an-an-1}
Suppose the following holds for some $N \geq 1$: there exists $t_N >0$ such that for every $1\leq n \leq N$ and for every multi-index $\alp$,  $k^{\bf [n]}$ satisfies the following estimate for some $C_{\alp,n}>0$ (depending on $\alp$, $n$, in addition to $c_{ij}$ and $p_i$) for all $(t,x)\in (0,t_{N}] \times \mathbb T^3$:
\begin{equation}\label{eq:k.diff.assump}
|\rd_x^{\alp} (k^{\bf [n]} - k^{\bf [n-1]})_i{ }^j|(t,x) \leq C_{\alp, n} t^{-1+ n \ve}
\end{equation}
for every $1\leq n\leq N$. 

Then, after choosing $t_N>0$ smaller if necessary, for every multi-index $\alp$ and $1\leq n \leq N$, there exists $C'_{\alp,n}>0$ such that
\begin{align}\label{eq:a.diff.est}
|\rd_x^\alp (a^{\bf [n]}_{ij} -a^{\bf [n-1]}_{ij})|(t,x) \leq C'_{\alp,n} t^{n\ve},\qquad | \rd_x^\alp \rd_t(a^{\bf [n]}_{ij} -a^{\bf [n-1]}_{ij})|(t,x) \leq C'_{\alp,n} t^{-1+n\ve},
\end{align}
for all $(t,x)\in (0,t_{N}] \times \mathbb T^3$ and for every $1\leq n \leq N$.
\end{lemma}
\begin{proof}

First, we note that by Lemma~\ref{lem:k0} and \eqref{eq:k.diff.assump}, 
\begin{equation}\label{eq:kn.-.0}
|\rd_x^{\alp} (k^{\bf [n]} - k^{\bf [0]})_i{ }^j|(t,x)\ls t^{-1+\ve}
\end{equation}
for every $1\leq n \leq N$.

Subtracting the $n$ and $n-1$ versions of \eqref{eq:a.n}, for $i\leq j$, we have
\begin{align}\label{dtan-an-1}
\begin{split}
&\: \rd_t (a^{\bf [n]}- a^{\bf [n-1]})_{ij} \\
= &\: -  \sum_{\ell} t^{2 p_{\max\{ \ell,j\}} -2p_{\max\{i,j\}}} [ (k^{\bf [n]} - k^{\bf [0]})_i{ }^\ell (a^{\bf [n]} - a^{\bf [n-1]})_{\ell j} + (k^{\bf [n]} - k^{\bf [n-1]})_i{ }^\ell a^{\bf [n-1]}_{\ell j}]\\
&\: -  \sum_{\ell} t^{2 p_{\max\{ \ell,i\}} -2p_{\max\{i,j\}}} [ (k^{\bf [n]} - k^{\bf [0]})_j{ }^\ell (a^{\bf [n]} - a^{\bf [n-1]})_{\ell i} + (k^{\bf [n]} - k^{\bf [n-1]})_j{ }^\ell a^{\bf [n-1]}_{\ell i}]\\
&\:- \sum_{\ell} t^{2 p_{\max\{ \ell,j\}} -2p_{\max\{i,j\}}} (k^{\bf [0]})_i{}^\ell(a^{\bf [n]} -a^{\bf [n-1]})_{\ell j}  - \sum_{\ell} t^{2 p_{\max\{ \ell,i\}} -2p_{\max\{i,j\}}} (k^{\bf [0]})_j{}^\ell(a^{\bf [n]} -a^{\bf [n-1]})_{\ell i} \\
&\: -\frac{2p_{\max\{i,j\}}}{t}(a^{\bf [n]}_{ij}-a^{\bf [n-1]}_{ {i} j}) {.}
\end{split}
\end{align}

Using the equation \eqref{dtan-an-1} and the estimates in Lemmas~\ref{lem:k0}, \ref{lem:kn.gen.bounds} and \eqref{eq:k.diff.assump}, we deduce a system of schematic equations in a similar manner as \eqref{eq:basic.eq.for.a-c.1}--\eqref{eq:basic.eq.for.a-c.6}, namely,
\begin{align}
\rd_t (a^{\bf [n]}- a^{\bf [n-1]})_{33} 
=&\,  O(t^{-1+\varepsilon}) (a^{\bf [n]} - a^{\bf [n-1]}) + O(t^{-1+n\varepsilon}) a^{\bf [n-1]}, \label{eq:basic.eq.for.an-an-1.1}\\
\rd_t (a^{\bf [n]}- a^{\bf [n-1]})_{22} 
=&\,  O(t^{-1+\varepsilon}) (a^{\bf [n]} - a^{\bf [n-1]}) + O(t^{-1+n\varepsilon}) a^{\bf [n-1]}, \label{eq:basic.eq.for.an-an-1.2}\\
\rd_t (a^{\bf [n]}- a^{\bf [n-1]})_{11} 
=&\,  O(t^{-1+\varepsilon}) (a^{\bf [n]} - a^{\bf [n-1]}) + O(t^{-1+n\varepsilon}) a^{\bf [n-1]}, \label{eq:basic.eq.for.an-an-1.3}\\
\rd_t (a^{\bf [n]}- a^{\bf [n-1]})_{23} 
=&\,  \frac{p_2 - p_3}{t} (a^{\bf [n]} - a^{\bf [n-1]})_{23} -\frac{\kappa_2{}^3}{t} (a^{\bf [n]} -a^{\bf [n-1]}c)_{33} \label{eq:basic.eq.for.an-an-1.4}\\
&\: +O(t^{-1+\varepsilon}) (a^{\bf [n]} - a^{\bf [n-1]}) + O(t^{-1+n\varepsilon}) a^{\bf [n-1]}, \notag\\
\rd_t (a^{\bf [n]}- a^{\bf [n-1]})_{12} 
=&\,  \frac{p_1 - p_2}{t} (a^{\bf [n]} - a^{\bf [n-1]})_{12} -\frac{\kappa_1{}^2}{t} (a^{\bf [n]} - a^{\bf [n-1]})_{22} \label{eq:basic.eq.for.an-an-1.5}\\
&\: +O(t^{-1+\varepsilon}) (a^{\bf [n]} - a^{\bf [n-1]}) + O(t^{-1+n\varepsilon}) a^{\bf [n01]}, \notag\\
\label{eq:basic.eq.for.an-an-1.6}\rd_t (a^{\bf [n]}- a^{\bf [n-1]})_{13} 
=&\, \frac{p_1-p_3}{t} (a^{\bf [n]} - a^{\bf [n-1]})_{13} - \frac{\kappa_1{}^2}{t}(a^{\bf [n]} - a^{\bf [n-1]})_{23} - \frac{\kappa_1{}^3}{t} (a^{\bf [n]} - a^{\bf [n-1]})_{33} \\
&\,+O(t^{-1+\varepsilon}) (a^{\bf [n]} - a^{\bf [n-1]}) + O(t^{-1+n\varepsilon}) a^{\bf [n-1]}. \notag
\end{align}

From this point on we can argue as in Lemma~\ref{lem:an.well.defined}, using the reductive structure of the system. Note that the system \eqref{eq:basic.eq.for.an-an-1.1}--\eqref{eq:basic.eq.for.an-an-1.6} is better than the system \eqref{eq:basic.eq.for.a-c.1}--\eqref{eq:basic.eq.for.a-c.6} in that the inhomogeneous terms $O(t^{-1+n\varepsilon}) a^{\bf [n-1]} = O(t^{-1+n\varepsilon})$. As a result, the argument in Lemma~\ref{lem:an.well.defined} gives the better estimate $|\rd_x^\alp \partial_t^r(a^{\bf [n]}_{ij} -a^{\bf [n-1]}_{ij})|(t,x) \leq C'_{\alp,n} t^{-r+n\ve}$, $r=0,1$. \qedhere

\end{proof}

 {\textbf{Now a straightforward induction argument using Lemmas~\ref{lem:kn.gen.bounds}, \ref{lem:kn-kn-1}, \ref{lem:an.well.defined}, \ref{lem:an-an-1} shows there exists a decreasing sequence of positive times $t_n$, such that $g^{\bf [n]}$ and $\kn$ are well-defined and smooth in $(0,t_n]\times \mathbb T^3$, for all $n\in \mathbb N$. Moreover, all the estimates in the conclusions (and proofs) of Lemmas~\ref{lem:kn.gen.bounds}, \ref{lem:kn-kn-1}, \ref{lem:an.well.defined}, \ref{lem:an-an-1} hold. In particular, points (1), (2) in Theorem~\ref{thm:parametrix} hold true; and after using also Lemma~\ref{lem:Ricci}, it can be checked that (3) in Theorem~\ref{thm:parametrix} is also verified.}}

In the remaining subsections, we prove points (4) and (5) in Theorem~\ref{thm:parametrix}, thus completing the proof of Theorem~\ref{thm:parametrix}.

\subsection{Comparing $\kn$ with the second fundamental form}
In this subsection, we prove point (4) of Theorem~\ref{thm:parametrix}; see the main estimate in Lemma~\ref{lem:k.II}.

The heart of the matter is the following estimates for $\mathfrak D^{\bf [n]}_{ij}:= (k^{\bf [n]})_{i}{ }^\ell g^{\bf [n]}_{\ell j} - (k^{\bf [n]})_{j}{ }^\ell g^{\bf [n]}_{\ell i}$.

\begin{lemma}\label{lem:D.est}
For each $n \in \mathbb N$, define $\mathfrak D^{\bf [n]}_{ij} = (k^{\bf [n]})_{i}{ }^\ell g^{\bf [n]}_{\ell j} - (k^{\bf [n]})_{j}{ }^\ell g^{\bf [n]}_{\ell i}$. Then if $(n+1)\ve >2$, after choosing $t_n$ smaller if necessary, the following estimate holds for some $C_{\alp,n}>0$ (depending on $\alp$, $n$, in addition to $c_{ij}$ and $p_i$):
$$|\rd_x^\alp \mathfrak D^{\bf [n]}_{ij}|(t,x) \leq C_{\alp,n} t^{-1+(n+2)\ve+2p_{\max\{i,j\}}} |\log t|^{2+|\alp|},\quad |\rd_x^\alp \rd_t \mathfrak D^{\bf [n]}_{ij}|(t,x) \leq C_{\alp,n} t^{-2+(n+2)\ve+2p_{\max\{i,j\}}} |\log t|^{2+|\alp|}$$
for all $(t,x) \in (0,t_n]\times \mathbb T^3$.
\end{lemma}
\begin{proof}
\pfstep{Step~1: Derivation of an equation for $\mathfrak D^{\bf [n]}_{ij}$} By \eqref{eq:g.transport},
\begin{equation}\label{eq:D.est.prelim}
\begin{split}
&\: (\rd_t g^{\bf [n]}_{\ell j}) (k^{\bf [n]})_i{ }^\ell - (\rd_t g^{\bf [n]}_{\ell i}) (k^{\bf [n]})_j{ }^\ell \\
= &\: - g^{\bf [n]}_{\ell b} (k^{\bf [n]})_j{ }^b (k^{\bf [n]})_i{ }^\ell - g^{\bf [n]}_{j b} (k^{\bf [n]})_\ell{ }^b (k^{\bf [n]})_i{ }^\ell + g^{\bf [n]}_{\ell b} (k^{\bf [n]})_i{ }^b (k^{\bf [n]})_j{ }^\ell + g^{\bf [n]}_{i b} (k^{\bf [n]})_\ell{ }^b (k^{\bf [n]})_j{ }^\ell \\
= &\: - (g^{\bf [n]}_{j b} (k^{\bf [n]})_\ell{ }^b -g^{\bf [n]}_{\ell b} (k^{\bf [n]})_j{ }^b) (k^{\bf [n]})_i{ }^\ell + (g^{\bf [n]}_{i b} (k^{\bf [n]})_\ell{ }^b -g^{\bf [n]}_{\ell b} (k^{\bf [n]})_i{ }^b) (k^{\bf [n]})_j{ }^\ell.
\end{split}
\end{equation}
Therefore, \eqref{eq:D.est.prelim} and the equation \eqref{eq:k.transport} that define $k^{\bf [n]}$, it follows that
\begin{equation}\label{D.est.Ricci.diff.exact}
\begin{split}
\rd_t \mathfrak D_{ij}^{\bf [n]} = &\: \rd_t [(k^{\bf [n]})_{i}{ }^\ell g^{\bf [n]}_{\ell j} - (k^{\bf [n]})_{j}{ }^\ell g^{\bf [n]}_{\ell i}]  \\
= &\: Ric(g^{\bf [n-1]})_i{ }^\ell g^{\bf [n]}_{\ell j} - Ric(g^{\bf [n-1]})_j{ }^\ell g^{\bf [n]}_{\ell i} + (k^{\bf [n]})_a{}^a[(k^{\bf [n]})_i{ }^\ell g^{\bf [n]}_{\ell j} - (k^{\bf [n]})_j{ }^\ell g^{\bf [n]}_{\ell i}] \\
&\: - (g^{\bf [n]}_{j b} (k^{\bf [n]})_\ell{ }^b -g^{\bf [n]}_{\ell b} (k^{\bf [n]})_j{ }^b) (k^{\bf [n]})_i{ }^\ell + (g^{\bf [n]}_{i b} (k^{\bf [n]})_\ell{ }^b -g^{\bf [n]}_{\ell b} (k^{\bf [n]})_i{ }^b) (k^{\bf [n]})_j{ }^\ell \\
= &\: Ric(g^{\bf [n-1]})_i{ }^\ell g^{\bf [n]}_{\ell j} - Ric(g^{\bf [n-1]})_j{ }^\ell g^{\bf [n]}_{\ell i} + (k^{\bf [n]})_a{}^a\mathfrak D^{\bf [n]}_{ij} - \mathfrak D^{\bf [n]}_{\ell j} (k^{\bf [n]})_i{ }^\ell + \mathfrak D^{\bf [n]}_{\ell i} (k^{\bf [n]})_j{ }^\ell.
\end{split}
\end{equation}

Now since $Ric(g^{\bf [n-1]})_i{ }^\ell g^{\bf [n-1]}_{\ell j}$ is symmetric in $i$ and $j$, we have
\begin{equation}\label{D.est.Ricci.diff}
\begin{split}
&\: Ric(g^{\bf [n-1]})_i{ }^\ell g^{\bf [n]}_{\ell j} - Ric(g^{\bf [n-1]})_j{ }^\ell g^{\bf [n]}_{\ell i} \\
= &\: Ric(g^{\bf [n-1]})_i{ }^\ell ( g^{\bf [n]} - g^{\bf [n-1]})_{\ell j} - Ric(g^{\bf [n-1]})_j{ }^\ell ( g^{\bf [n]} - g^{\bf [n-1]})_{\ell i} =O(t^{-2+(n+2)\ve+ 2p_{\max\{i,j\}}}|\log t|^2),
\end{split}
\end{equation}
where the final estimate follows from the form of the metric, Lemmas \ref{lem:Ricci}, \ref{lem:an-an-1}, and the 
fact that
\begin{equation*}
\begin{split}
&\: O(\min\{t^{-2+2\ve}, t^{-2+2\ve-2p_\ell+2p_i}\}|\log t|^2 \times t^{2p_{\max\{j,\ell \}}}) \\
= &\: O(\min\{t^{-2+2\ve+2p_j},\, t^{-2+2\ve+2p_i}\}|\log t|^2 ) = O(t^{-2+2\ve+2p_{\max\{i,j\}}} |\log t|^2).
\end{split}
\end{equation*}

Therefore, combining \eqref{D.est.Ricci.diff.exact} and \eqref{D.est.Ricci.diff}, we have obtained that \begin{equation}\label{eq:D.original}
\rd_t \mathfrak D^{\bf [n]}_{ij} = -\mathfrak D^{\bf [n]}_{\ell j} (k^{\bf [n]})_i{ }^\ell + \mathfrak D^{\bf [n]}_{\ell i} (k^{\bf [n]})_j{ }^\ell + (k^{\bf [n]})_a{}^a \mathfrak D^{\bf [n]}_{ij} + O(t^{-2+(n+2)\ve + 2p_{\max\{i,j\}}} |\log t|^2).
\end{equation}

\pfstep{Step~2: Estimating $\mathfrak D^{\bf [n]}_{ij}$} Since $\mathfrak D^{\bf [n]}_{ij}$ is manifestly anti-symmetric, it suffices to estimate $\mathfrak D^{\bf [n]}_{23}$, $\mathfrak D^{\bf [n]}_{13}$ and $\mathfrak D^{\bf [n]}_{12}$. By \eqref{eq:D.original}, they satisfy the following equations: 
$$\rd_t \mathfrak D^{\bf [n]}_{23} = [\f {p_2+p_3-1}t+O(t^{-1+\varepsilon})] \mathfrak D^{\bf [n]}_{23} - (k^{\bf [n]})_2{ }^1 \mathfrak D^{\bf [n]}_{13} + (k^{\bf [n]})_3{ }^1 \mathfrak D^{\bf [n]}_{12} + O(t^{-2+(n+2)\ve+ 2p_{\max\{i,j\}}} |\log t|^2),$$
$$\rd_t \mathfrak D^{\bf [n]}_{13} = [\f {p_1+p_3-1}t+O(t^{-1+\varepsilon})]  \mathfrak D^{\bf [n]}_{13} - (k^{\bf [n]})_3{ }^2 \mathfrak D^{\bf [n]}_{12} - (k^{\bf [n]})_1{ }^2 \mathfrak D^{\bf [n]}_{23} + O(t^{-2+(n+2)\ve+ 2p_{\max\{i,j\}}} |\log t|^2),$$
$$\rd_t \mathfrak D^{\bf [n]}_{12} = [\f {p_1+p_2-1}t+O(t^{-1+\varepsilon})] \mathfrak D^{\bf [n]}_{12} - (k^{\bf [n]})_2{ }^3 \mathfrak D^{\bf [n]}_{13} - (k^{\bf [n]})_1{ }^3 \mathfrak D^{\bf [n]}_{23} + O(t^{-2+(n+2)\ve+ 2p_{\max\{i,j\}}} |\log t|^2).$$

Applying the estimates for $k^{\bf [n]}$ from Lemma~\ref{lem:kn.gen.bounds}, we obtain
\begin{align}
\rd_t (t^{p_1}\mathfrak D^{\bf [n]}_{23}) =&\, O(t^{-1+\ve}) t^{p_1}\mathfrak D^{\bf [n]}_{23}+ O(t^{-1+\ve}) t^{p_1}\mathfrak D^{\bf [n]}_{13}\notag \\
&\: + O(t^{-1+\ve-2p_1+2p_3}) t^{p_1}\mathfrak D^{\bf [n]}_{12} + O(t^{-2+p_1+(n+2)\ve+2p_3}|\log t|^2),\label{eq:D.transport.23}\\
\rd_t (t^{p_2} \mathfrak D^{\bf [n]}_{13}) =&\, O(t^{-1+\ve}) t^{p_2}\mathfrak D^{\bf [n]}_{13} + O(t^{-1}) t^{p_2} \mathfrak D^{\bf [n]}_{23} \notag\\
&\: + O(t^{-1+\ve-2p_2+2p_3}) t^{p_2} \mathfrak D^{\bf [n]}_{12} + O(t^{-2+p_2+(n+2)\ve+2p_3}|\log t|^2), \label{eq:D.transport.13}\\
\rd_t (t^{p_3} \mathfrak D^{\bf [n]}_{12}) =&\,  O(t^{-1+\ve}) t^{p_3}\mathfrak D^{\bf [n]}_{12} + O(t^{-1}) t^{p_3} \mathfrak D^{\bf [n]}_{13} \notag\\
&\: + O(t^{-1}) t^{p_3} \mathfrak D^{\bf [n]}_{23} + O(t^{-2+p_3+(n+2)\ve+2p_2}|\log t|^2). \label{eq:D.transport.12}
\end{align}

To use these equations, note that when $i$, $j$, $\ell$ are all distinct,
\begin{equation}\label{eq:D.initial}
\lim_{t\to 0^+} t^{p_i} \mathfrak D^{\bf [n]}_{j \ell} = 0.
\end{equation}
Indeed, using the estimates in  {Lemmas~\ref{lem:kn.gen.bounds} and \ref{lem:an.well.defined},} one checks that $\mathfrak D^{\bf [n]}_{23},\, \mathfrak D^{\bf [n]}_{13} = O(t^{2p_3-1})$ and $\mathfrak D^{\bf [n]}_{12} = O(t^{2p_2-1})$. This implies $t^{p_1}\mathfrak D^{\bf [n]}_{23} = O(t^{p_1+2p_3-1}) = O(t^{p_3-p_2})$, $t^{p_2} \mathfrak D^{\bf [n]}_{13} = O(t^{p_2 +2p_3-1}) = O(t^{p_3-p_1})$ and $t^{p_3} \mathfrak D^{\bf [n]}_{12} = O(t^{p_3 +2p_2-1}) = O(t^{p_2-p_1})$. We then obtain \eqref{eq:D.initial} using $p_1< p_2 < p_3$.

We now use equations \eqref{eq:D.transport.23}--\eqref{eq:D.transport.12} to estimate $\mathfrak D^{\bf [n]}_{ij}$. The key is to notice a reductive structure similar to that in the proof of Lemma~\ref{lem:an.well.defined}, except in this situation since the different components have different rates, we argue with a bootstrap argument. 

Make the bootstrap assumptions that
\begin{equation}\label{eq:BA.for.D}
\begin{split}
|\mathfrak D^{\bf [n]}_{23}|(t,x) \leq A&t^{-1+(n+2)\ve+2p_3} |\log t|^2,\quad |\mathfrak D^{\bf [n]}_{13}|(t,x) \leq At^{-1+(n+2)\ve+2p_3} |\log t|^2,\\
&\:\quad|\mathfrak D^{\bf [n]}_{12}|(t,x) \leq At^{-1+(n+2)\ve+2p_2} |\log t|^2,
\end{split}
\end{equation}
where $A$ is a large constant, such that denoting the implicit constant in the big-$O$ notation in \eqref{eq:D.transport.23}--\eqref{eq:D.transport.12} by $C$, we require $C\ll A$.

Plugging \eqref{eq:BA.for.D} into \eqref{eq:D.transport.23}, integrating, and using $p_2>p_1$,
we obtain 
\begin{equation}\label{eq:D.improved.23}
 |\mathfrak D^{\bf [n]}_{23}|(t,x) \leq Ct^{-1+(n+2)\ve+2p_3} |\log t|^2 + CA t^{-1+(n+3)\ve+2p_3} |\log t|^2 .
\end{equation}
Arguing similarly, first for $\mathfrak D^{\bf [n]}_{13}$ and then for $\mathfrak D^{\bf [n]}_{12}$, we also obtain
\begin{equation}\label{eq:D.improved.13}
 |\mathfrak D^{\bf [n]}_{13}|(t,x) \leq Ct^{-1+(n+2)\ve+2p_3} |\log t|^2 + CA t^{-1+(n+3)\ve+2p_3} |\log t|^2,
\end{equation}
\begin{equation}\label{eq:D.improved.12}
 |\mathfrak D^{\bf [n]}_{12}|(t,x) \leq Ct^{-1+(n+2)\ve+2p_2} |\log t|^2 + CA t^{-1+(n+3)\ve+2p_2} |\log t|^2.
\end{equation}

Choosing $t_n$ sufficiently small (so that $A t^\ve \leq 1$), it is easy to check that \eqref{eq:D.improved.23}--\eqref{eq:D.improved.12} improves the bootstrap assumptions in \eqref{eq:BA.for.D}. This gives the stated estimates for $\mathfrak D^{\bf [n]}_{ij}$ in the lemma when $|\alp| = 0$.

The estimates for the spatial derivatives are similar, except that we lose a factor of $|\log t|$ for each derivative we take (cf.~\eqref{eq:Ricci.improved}). \qedhere

\pfstep{Step~3: Estimating $\rd_t \mathfrak D^{\bf [n]}_{ij}$} Finally, we plug in the estimates for $\mathfrak D^{\bf [n]}_{ij}$ into \eqref{eq:D.transport.23}--\eqref{eq:D.transport.12} to obtain the desired estimates for $\rd_x^\alp \rd_t \mathfrak D^{\bf [n]}_{ij}$. \qedhere
\end{proof}

\begin{lemma}\label{lem:D.est.2}
For each $n \in \mathbb N$ and $\mathfrak D^{\bf [n]}_{ij}$ as in Lemma~\ref{lem:D.est}, define $(\wtDn)_i{}^j := (g^{\bf [n]})^{j\ell} \mathfrak D^{\bf [n]}_{i\ell}$. Then if $(n+1)\ve >2$, the following estimates hold for $(t,x) \in (0,t_n]\times \mathbb T^n$ for some $C_{\alp,n}>0$ (depending on $\alp$, $n$, in addition to $c_{ij}$ and $p_i$):
\begin{equation}\label{eq:D.est.2.main}
\sum_{r=0}^2 t^r |\rd_x^\alp \rd_t^r (\wtDn)_i{}^j|(t,x) \leq C_{n,\alp} \min\{ t^{-1+(n+2)\ve} |\log t|^{2+|\alp|}, t^{-1+(n+2)\ve-2p_j+2p_i} |\log t|^{2+|\alp|}\}.
\end{equation}
\end{lemma}
\begin{proof}
\pfstep{Step~1: Estimates for $\wtDn$ (when $r=0$)} By Lemma~\ref{lem:D.est},the estimate $\rd_x^\alp (g^{\bf [n]})^{j\ell}= O(t^{-2p_{\min\{j,\ell\}}}|\log t|^{|\alpha|})$ and the fact $t^{\ve}|\log t|^\ell \ls_\ell 1$, we immediately obtain 
\begin{equation}\label{eq:D.est.2.main.r=0}
 |\rd_x^\alp (\wtDn)_i{}^j|(t,x) \leq C_{n,\alp} \min\{ t^{-1+(n+2)\ve} |\log t|^{2+|\alp|}, t^{-1+(n+2)\ve-2p_j+2p_i} |\log t|^{2+|\alp|}\}.
\end{equation}

\pfstep{Step~2: Deriving evolution equations for $\wtDn$} Contracting \eqref{D.est.Ricci.diff.exact} with $(g^{\bf [n]})^{j b}$, using \eqref{eq:gn-1.transport} and the anti-symmetry of $\mathcal D^{\bf [n]}_{ij}$, we obtain 
\begin{equation}\label{eq:wtDn.prelim}
\rd_t (\wtDn)_{i}{}^b = (k^{\bf [n]})_a{}^a (\wtDn)_i{}^b + Ric(g^{\bf [n-1]})_i{ }^b  - (g^{\bf [n]})^{j b} Ric(g^{\bf [n-1]})_j{ }^\ell g^{\bf [n]}_{\ell i}.
\end{equation}

We notice now that since 
$$(g^{\bf [n-1]})^{j b} Ric(g^{\bf [n-1]})_j{ }^\ell g^{\bf [n-1]}_{\ell i} = Ric(g^{\bf [n-1]})_i{ }^b,$$
we have
\begin{equation}\label{eq:Ric.diff.with.n.n-1.1}
\begin{split}
&\: Ric(g^{\bf [n-1]})_i{ }^b  - (g^{\bf [n]})^{j b} Ric(g^{\bf [n-1]})_j{ }^\ell g^{\bf [n]}_{\ell i} \\
= &\: - (g^{\bf [n]})^{j b} Ric(g^{\bf [n-1]})_j{ }^\ell g^{\bf [n]}_{\ell i} + (g^{\bf [n-1]})^{j b} Ric(g^{\bf [n-1]})_j{ }^\ell g^{\bf [n-1]}_{\ell i} \\
= &\: - [ (g^{\bf [n]})^{j b} - (g^{\bf [n-1]})^{j b}] Ric(g^{\bf [n-1]})_j{ }^\ell g^{\bf [n]}_{\ell i} - (g^{\bf [n-1]})^{j b} Ric(g^{\bf [n-1]})_j{ }^\ell [ g^{\bf [n]}_{\ell i} - g^{\bf [n-1]}_{\ell i}] \\
= &\: O(t^{-2p_{\min\{j,b\}}})\times O( t^{n\ve}) \times O(\min\{ t^{-2+2\ve}|\log t|^{2+|\alp|},\, t^{-2+\ve-2p_\ell+2p_j}|\log t|^{2+|\alp|} \}) \times O(t^{2p_{\max\{\ell,i\}}}) \\
= &\: O(|\log t|^2\times \min\{ t^{-2+(n+2)\ve},\, t^{-2+(n+2)\ve-2p_b+2p_i} \}),
\end{split}
\end{equation}
where in estimating the terms we have used the form of the metric, computation of the inverse metric (see \eqref{invg0}), Lemmas~\ref{lem:an.well.defined} and \ref{lem:an-an-1}, and \eqref{eq:Ricci.improved}.

Differentiating \eqref{eq:Ric.diff.with.n.n-1.1} by $\rd_x^\alp \rd_t^r$, and arguing similarly, we also obtain the following higher derivative bounds for $r = 0,1$:
\begin{equation}\label{eq:Ric.diff.with.n.n-1.2}
\begin{split}
&\: \rd_x^\alp \{ t^r\rd_t^r [(Ric(g^{\bf [n-1]})_i{ }^b  - (g^{\bf [n]})^{j b} Ric(g^{\bf [n-1]})_j{ }^\ell g^{\bf [n]}_{\ell i})] \} \\
= &\: O(|\log t|^{2+|\alp|}\times \min\{ t^{-2+(n+2)\ve},\, t^{-2+(n+2)\ve-2p_b+2p_i} \}).
\end{split}
\end{equation}

Plugging the estimate \eqref{eq:Ric.diff.with.n.n-1.2} into \eqref{eq:wtDn.prelim}, using the estimates for $\kn$ (by Lemma~\ref{lem:kn.gen.bounds}) and $\rd_t \kn$ (by \eqref{eq:k.transport}, Lemma~\ref{lem:Ricci}, \eqref{eq:a-c.est} and Lemma~\ref{lem:kn.gen.bounds}), (and relabelling the indices,) we obtain
\begin{equation}\label{eq:wtDn}
\begin{split}
&\: \rd_x^\alp \rd_t  (\wtDn)_{i}{}^b \\
= &\: O(t^{-1}) (\sum_{|\bt|\leq |\alp|} \rd_x^\bt (\wtDn)_i{}^b) + O(|\log t|^{2+|\alp|} \times \min\{ t^{-2+(n+2)\ve}, t^{-2+(n+2)\ve-2p_b+2p_i}\}),
\end{split}
\end{equation}
and
\begin{equation}\label{eq:dtwtDn}
\begin{split}
 \rd_x^\alp \rd_t^2 (\wtDn)_{i}{}^b 
=&\: O(t^{-1})\sum_{|\bt|\leq |\alp|} \rd_x^\bt \rd_t (\wtDn)_i{}^b + O(t^{-2})\sum_{|\bt|\leq |\alp|} \rd_x^\bt (\wtDn)_i{}^b \\
&\: +  O(|\log t|^{2+|\alp|} \times \min\{ t^{-3+(n+2)\ve},\, t^{-3+(n+2)\ve-2p_b+2p_i}\}).
\end{split} 
\end{equation}

\pfstep{Step~3: Estimates for $\rd_t \wtDn$ and $\rd_t^2 \wtDn$ (when $r=1,\,2$)} Plugging \eqref{eq:D.est.2.main.r=0} into \eqref{eq:wtDn}, we obtain
\begin{equation}\label{eq:D.est.2.main.r=1}
 |\rd_x^\alp \rd_t(\wtDn)_i{}^j|(t,x) \leq C_{n,\alp} \min\{ t^{-2+(n+2)\ve} |\log t|^{2+|\alp|}, t^{-2+(n+2)\ve-2p_j+2p_i} |\log t|^{2+|\alp|}\}.
\end{equation}
Similarly, plugging in both \eqref{eq:D.est.2.main.r=1} and \eqref{eq:D.est.2.main.r=0} into \eqref{eq:dtwtDn}, we obtain
\begin{equation}\label{eq:D.est.2.main.r=2}
 |\rd_x^\alp \rd_t^2(\wtDn)_i{}^j|(t,x) \leq C_{n,\alp} \min\{ t^{-3+(n+2)\ve} |\log t|^{2+|\alp|}, t^{-3+(n+2)\ve-2p_j+2p_i} |\log t|^{2+|\alp|}\}.
\end{equation}

Combining \eqref{eq:D.est.2.main.r=0}, \eqref{eq:D.est.2.main.r=1} and \eqref{eq:D.est.2.main.r=2} yields \eqref{eq:D.est.2.main}. \qedhere

\end{proof}

The next lemma shows that even though $k^{\bf [n]}$ is \emph{not} the second fundamental form associated to $g^{\bf [n]}$, it is close to being the second fundamental form up to an error that vanishes sufficiently fast as $t\to 0^+$.
\begin{lemma}\label{lem:k.II}
When $(n+1)\ve >2$, the following estimates hold for $(t,x) \in (0,t_n]\times \mathbb T^n$ for some $C_{\alp,n}>0$ (depending on $\alp$, $n$, in addition to $c_{ij}$ and $p_i$):
\begin{equation}\label{eq:diff.k.II.1}
\sum_{r=0}^{2} t^r| \rd_x^{\alp} \rd_t^r (2(k^{\bf [n]})_i{ }^j +  (g^{\bf [n]})^{j\ell} \rd_t g^{\bf [n]}_{i\ell}) |(t,x) \leq C_{n,\alp} t^{-1+(n+2)\ve} |\log t|^{2+|\alp|}.
\end{equation}
\end{lemma}
\begin{proof}
By \eqref{eq:g.transport} and the definition of $\mathfrak D^{\bf [n]}_{ij}$ (in Lemma~\ref{lem:D.est}) and $\wtDn$ (in Lemma~\ref{lem:D.est.2}),
\begin{equation}\label{truekn}
\begin{split}
&\: 2(k^{\bf [n]})_i{ }^j + (g^{\bf [n]})^{j\ell} \rd_t g^{\bf [n]}_{i\ell}\\
= &\: 2(k^{\bf [n]})_i{ }^j - 2(k^{\bf [n]})_i{ }^j - (g^{\bf [n]})^{j\ell}\mathfrak D^{\bf [n]}_{i\ell} = - (g^{\bf [n]})^{j\ell}\mathfrak D^{\bf [n]}_{i\ell}=-(\wtDn)_i{}^j.
\end{split}
\end{equation}
The desired estimate is then an immediate consequence of Lemma~\ref{lem:D.est.2}. \qedhere
\end{proof}

\textbf{Lemma~\ref{lem:k.II} gives point (4) in Theorem~\ref{thm:parametrix}.}

\subsection{$(k^{[\bf n]})_i{}^j$ and $g^{[\bf n]}_{ij}$ satisfy evolution equations approximately}

In this subsection we prove point (5) in Theorem~\ref{thm:parametrix} (see Proposition~\ref{prop:approx.ev.1}), which then completes the proof of the theorem.
\begin{lemma}\label{lem:aux.dtRic}
For every $n\in \mathbb N$, the following estimates hold for $(t,x) \in (0,t_n]\times \mathbb T^n$, for some $C_{\alp,n}>0$ (depending on $\alp$, $n$, in addition to $c_{ij}$ and $p_i$):
\begin{equation*}
\begin{split}
&\: \left|\rd_x^\alp \rd_t \left( Ric(g^{\bf [n]})_i{ }^j- Ric(g^{\bf [n-1]})_i{ }^j\right)\right|(t,x) \leq C_{\alp,n} t^{-3+(n+2)\ve} |\log t|^{2+|\alp|}.
\end{split}
\end{equation*}
\end{lemma}
\begin{proof}
Going back to the proof of Lemma \ref{lem:Ricci} and using the form of the metrics $g^{\bf[n]}$ and $g^{[\bf n-1]}$, we notice that the each term in the difference of $Ric(g^{\bf [n]})_i{ }^j,Ric(g^{\bf [n-1]})_i{ }^j$ has the form:
\begin{align*}
&[\text{explicit powers of $t$ and $\log t$ with behavior $O(t^{-2+2\varepsilon}|\log t|^{2+|\alp|})$}]\\ 
\times&[\text{non-linear terms in $\partial_x^\alpha a^{[\bf n]},\partial_x^\alpha a^{[\bf n-1]}$ which are linear in the difference 
$\partial_x^\alpha(a^{[\bf n]}-a^{[\bf n-1]})$, $|\alpha|\leq2$}]
\end{align*}
The fact that $a^{[\bf n]},a^{[\bf n-1]}$ and their spatial derivatives are bounded, while $|\partial_x^\alpha(a^{[\bf n]}-a^{[\bf n-1]})|\lesssim t^{n\varepsilon}$ (see Lemma \ref{lem:an-an-1}), was then used in Lemma \ref{lem:kn-kn-1} to infer the bound \eqref{eq:Ric.diff.est.more.precise}. 
 
Now we verify that a time derivative acting on any of the previous type of terms, adds at worst a power of $t^{-1}$ in their behavior. For the factors which are explicit powers of $t$ this is evident. If $\partial_t$ hits either $a^{[\bf n]},a^{[\bf n-1]}$ factor or their difference $a^{[\bf n]}-a^{[\bf n-1]}$, we make use of \eqref{eq:a-c.est}, \eqref{eq:a.diff.est} and the conclusion follows.
\end{proof}

\begin{proposition}\label{prop:approx.ev.1}
For every $n\in \mathbb N$, the following estimates hold for $(t,x) \in (0,t_n]\times \mathbb T^n$, for some $C_{\alp,n}>0$ (depending on $\alp$, $n$, in addition to $c_{ij}$ and $p_i$):
\begin{align*}
\sum_{r=0}^{1} t^r\left|\rd_x^\alp \rd_t^r\left(\rd_t (k^{\bf [n]})_i{ }^j - Ric(g^{\bf [n]})_i{ }^j - (k^{\bf [n]})_\ell{ }^\ell(k^{\bf [n]})_i{ }^j\right)\right|(t,x) \leq&\, C_{\alp,n} t^{-2+(n+2)\varepsilon}|\log t|^{2+|\alp|}.
\end{align*}
\end{proposition}
\begin{proof}
Using the equation \eqref{eq:k.transport}, the estimate \eqref{eq:Ric.diff.est}, and Lemma \ref{lem:aux.dtRic}, we obtain
\begin{align*}
&\: \left|\rd_x^\alp\left(\rd_t (k^{\bf [n]})_i{ }^j - Ric(g^{\bf [n]})_i{ }^j - (k^{\bf [n]})_\ell{ }^\ell(k^{\bf [n]})_i{ }^j\right)\right|(t,x) \\
= &\: 
\left|\rd_x^\alp\left( Ric(g^{\bf [n]})_i{ }^j- Ric(g^{\bf [n-1]})_i{ }^j\right)\right|(t,x)\lesssim t^{-2+(n+2)\varepsilon}|\log t|^{2+|\alp|}
\end{align*}
and 
\begin{align*}
&\: \left|\rd_x^\alp\partial_t\left(\rd_t (k^{\bf [n]})_i{ }^j - Ric(g^{\bf [n]})_i{ }^j - (k^{\bf [n]})_\ell{ }^\ell(k^{\bf [n]})_i{ }^j\right)\right|(t,x) \\
= &\: 
\left|\rd_x^\alp\partial_t\left( Ric(g^{\bf [n]})_i{ }^j- Ric(g^{\bf [n-1]})_i{ }^j\right)\right|(t,x)\lesssim t^{-3+(n+2)\varepsilon}|\log t|^{2+|\alp|},
\end{align*}
as desired. \qedhere
\end{proof}

\textbf{Proposition~\ref{prop:approx.ev.1} implies point (5) of Theorem~\ref{thm:parametrix}. Together with the previous subsections, this completes the proof of Theorem~\ref{thm:parametrix}.}

\section{Approximate propagation of constraints}\label{sec:approx.constraints}

We continue to work under the assumptions of Theorem~\ref{mainthm} and take $g^{\bf [n]}$ and $k^{\bf [n]}$ as constructed in the beginning of Section~\ref{sec:parametrix} (so that for appropriately chosen $t_N$,  the estimates in Theorem~\ref{thm:parametrix} hold).

The goal of this section is to show that metrics ${^{(4)}g^{[{\bf n}]}}$ are also approximate solutions to the constraints, as $t\rightarrow0$, to an order that improves with the increase of $n$.  {To achieve this we argue by   \emph{propagation of constraints}, i.e.~we use the second Bianchi identity as propagation equations and use that} the constraints are asymptotically valid in the renormalized sense \eqref{asymHamconst}--\eqref{asymmomconst}.

It will be useful to setup some notations that we use in this section. For the remainder of this section, \textbf{$D$ will denote the Levi--Civita connection of the spacetime metric $^{(4)} g^{\bf [n]}$ and $\nab$ will denote the Levi--Civita connection of the metric $g^{\bf [n]}$ on the (spacelike) constant-$t$ hypersurfaces.} Moreover, \textbf{indices are lowered and raised with respect to the metric $g^{\bf [n]}$} (in particular $(g^{\bf [n]})^{ij} = ((g^{\bf [n]})^{-1})^{ij}$ in this section).

\begin{proposition}\label{prop:curv.id}
Let $^{(4)}g = -\ud t^2 + g$, where $g$ is a Riemannian metric. Define $\widetilde{k}_i{ }^j := - \f 12 (g^{-1})^{j\ell} \rd_t g_{i\ell}$ (the second fundamental form). Then the following identities hold:
\begin{align}
\label{Ricij}Ric(^{(4)}g)_i{ }^j = &\: -\rd_t \widetilde{k}_i{ }^j + Ric(g)_i{ }^j + \widetilde{k}_\ell{ }^\ell \widetilde{k}_i{ }^j,\\
\label{Ricti}Ric(^{(4)}g)_{ti} = &\: -(\mathrm{div}_g \widetilde{k})_i + \nab_i (\widetilde{k}_\ell{ }^\ell),\\
\label{Rictt}Ric(^{(4)}g)_{tt} = &\:\rd_t (\widetilde{k}_\ell{ }^\ell) - |\widetilde{k}|^2,\\
\label{Scalar}R(^{(4)}g) = &\:-2 \rd_t (\widetilde{k}_\ell{ }^\ell) + R(g) + |\widetilde{k}|^2 + (\widetilde{k}_\ell{ }^\ell)^2,\\
\label{Gij}G(^{(4)}g)_i{ }^j =&\: -\rd_t \widetilde{k}_i{ }^j + Ric(g)_i{ }^j + \widetilde{k}_\ell{ }^\ell \widetilde{k}_i{ }^j -\f 12 \delta_i{ }^j [-2 \rd_t ( \widetilde{k}_\ell{ }^\ell) + R(g) + |\widetilde{k}|^2 + (\widetilde{k}_\ell{ }^\ell)^2],\\
\label{Gti}G(^{(4)}g)_{ti} =&\: -(\mathrm{div} \widetilde{k})_i + \nab_i (\widetilde{k}_\ell{ }^\ell),\\
\label{Gtt}G(^{(4)}g)_{tt} = &\:\f 12[R(g)  - |\widetilde{k}|^2 + (\widetilde{k}_\ell{ }^\ell)^2],
\end{align}
where $G(^{(4)}g)_{\alpha\beta}$ is the Einstein tensor of ${^{(4)}g}$.  
\end{proposition}
\begin{proof}
The first three identities can be found in \cite[Chapter~6, (3.20)--(3.22)]{yCB09} (after substituting the lapse to be identically $1$). The remaining identities follow from the first three by simple algebraic manipulations.
\end{proof}
\begin{lemma}\label{lem:asymconst}
Given $n\in\mathbb{N}\cup\{0\}$, ${^{(4)}g^{\bf[n]}}$ and $\kn$ given by Theorem \ref{thm:parametrix} satisfy the estimates:
\begin{align}\label{asymconst}
\begin{split}
\big|\partial_x^\alpha [R(g^{\bf[n]})-|k^{\bf[n]}|^2+(\mathrm{tr}k^{\bf[n]})^2]\big|\leq C_{\alpha,n}t^{-2+\varepsilon},&\qquad\big|\partial_x^\alpha [\partial_t( {\mathrm{tr}k^{\bf[n]}})-|k^{\bf[n]}|^2]\big|\leq C_{\alpha,n}t^{-2+\varepsilon} \\
\big|\partial_x^\alpha[\nabla_j(k^{\bf[n]})_i{}^j-\partial_i(k^{\bf[n]})_\ell{}^\ell]\big|\leq &\,C_{\alpha,n}t^{-1+\varepsilon}
\end{split}
\end{align}
\end{lemma}
\begin{proof}
By point (3) in Theorem~\ref{thm:parametrix}, it follows that$$|\partial_x^\alpha R(g^{\bf[n]})|\leq C_{\alpha,n}t^{-2+\varepsilon}.$$
Writing also
\begin{align*}
|k^{\bf[n]}|^2-(\text{tr}k^{\bf[n]})^2=&\,(k^{\bf[n]}- t^{-1}\kappa)_i{}^j(k^{\bf[n]} - t^{-1}\kappa)_j{}^i-[(k^{\bf[n]} - t^{-1}\kappa)_\ell^\ell]^2
+\frac{1}{t^2}\sum_ip^2_i-\frac{1}{t^2}\bigg(\sum_ip_i\bigg)^2\\
&+(k^{\bf[n]} - t^{-1}\kappa)_i{}^j(t^{-1}\kappa)_j{}^i+(k^{\bf[n]} - t^{-1}\kappa)_j{}^i(t^{-1}\kappa)_i{}^j
-2(t^{-1}\kappa)_\ell{}^\ell(k^{\bf[n]} - t^{-1}\kappa)_\ell{}^\ell,
\end{align*}
we conclude the first estimate using condition $2.$ in Theorem \ref{mainthm} and the second inequality in $3.$, Theorem \ref{thm:parametrix}. 

For the second estimate, first note that after tracing the first inequality in $3.$ of Theorem~\ref{thm:parametrix}, we obtain
$$\left| \rd_x^\alp \left(\rd_t (\mathrm{tr}k^{\bf[n]}) - R(g^{\bf [n]}) -(\mathrm{tr}k^{\bf[n]})^2\right)  \right| \leq C_{\alp,n} t^{-2+(n+1)\ve}.$$
Combining this with the first estimate in \eqref{asymconst} that we have just established, we obtain the second estimate in \eqref{asymconst}.

%
%

We now turn to the third estimate in \eqref{asymconst}. For notational clarity, we focus on the case $|\alp| = 0$. All the higher derivative bounds are be derived analogously after noticing the crucial algebraic structure. We compute:
\begin{equation}\label{eq:approx.momentum.1}
\begin{split}
&\nabla_j(k^{\bf[n]})_i{}^j-\partial_i (k^{\bf[n]})_\ell{}^\ell \\
=&\,\partial_j(k^{\bf[n]} - t^{-1}\kappa)_i{}^j
-\partial_i(k^{\bf[n]} - t^{-1}\kappa)_\ell{}^\ell+\frac{\partial_j\kappa_i{}^j}{t}
-(\Gamma^{[\bf n]})_{ij}^\ell (k^{[\bf n]})_\ell{}^j
+(\Gamma^{[\bf n]})_{j\ell}^j (k^{[\bf n]})_i{}^\ell\\
=&\,\partial_j(k^{\bf[n]} - t^{-1}\kappa)_i{}^j
-\partial_i(k^{\bf[n]} - t^{-1}\kappa)_\ell{}^\ell+\frac{\partial_j\kappa_i{}^j}{t}\\
&-\frac{1}{2}(g^{[\bf n]})^{\ell b}(\partial_ig^{[\bf n]}_{jb}+\partial_jg^{[\bf n]}_{ib}-\partial_bg^{[\bf n]}_{ij})(k^{[\bf n]})_\ell{}^j
+\frac{1}{2}(g^{[\bf n]})^{j b}(\partial_jg^{[\bf n]}_{\ell b}+\partial_\ell g^{[\bf n]}_{jb}-\partial_bg^{[\bf n]}_{j\ell})(k^{[\bf n]})_i{}^\ell .
\end{split}
\end{equation}

Notice now that by \eqref{eq:diff.k.II.1}, we have:
\begin{equation*}
 (k^{\bf [n]})_i{ }^j 
=-\f 12  (g^{\bf [n]})^{j\ell} \rd_t g^{\bf [n]}_{i\ell} + O(t^{-1+(n+1)\ve}).
\end{equation*}
Therefore, 
\begin{equation}\label{eq:approx.momentum.2}
\begin{split}
&\: \frac{1}{2}(g^{[\bf n]})^{\ell b}(\partial_jg^{[\bf n]}_{ib}-\partial_bg^{[\bf n]}_{ij})(k^{[\bf n]})_\ell{}^j \\
= &\: -\frac{1}{4}(g^{[\bf n]})^{\ell b}\partial_jg^{[\bf n]}_{ib}  (g^{\bf [n]})^{jc} \rd_t g^{\bf [n]}_{c\ell} + \f 14(g^{[\bf n]})^{\ell b}\partial_bg^{[\bf n]}_{ij}(g^{\bf [n]})^{jc} \rd_t g^{\bf [n]}_{c\ell} + O(t^{-1+\ve}) = O(t^{-1+\ve}),
\end{split}
\end{equation}
where in order to show that this $ = O(t^{-1+\ve})$, we look at the second term, relabel the indices $b \leftrightarrow j$ and then swap $c \leftrightarrow \ell$ (using that $g^{\bf [n]}$ is symmetric), which then gives the negative of the first term.

For the term $-\frac{1}{2}(g^{[\bf n]})^{\ell b} \partial_ig^{[\bf n]}_{jb} (k^{[\bf n]})_\ell{}^j$, we first note that if $\ell>j$, then $(k^{[\bf n]})_\ell{}^j = O(t^{-1-2p_j+2p_\ell}) $ and $(g^{[\bf n]})^{\ell b} \partial_ig^{[\bf n]}_{jb} = O(|\log t|)$, so altogether we get an $O(t^{-1+\ve})$ contribution. If $\ell <j$, then $(g^{[\bf n]})^{\ell b} \partial_ig^{[\bf n]}_{jb} = O(t^\ve)$, which together with $(k^{[\bf n]})_\ell{}^j = O(t^{-1})$, we get a combined contribution of $O(t^{-1+\ve})$. We therefore only get the contribution when $j = \ell$, i.e.
\begin{equation}\label{eq:approx.momentum.3}
\begin{split}
&\: -\frac{1}{2}(g^{[\bf n]})^{\ell b} \partial_ig^{[\bf n]}_{jb} (k^{[\bf n]})_\ell{}^j = - \sum_{j=1}^3 \f 12 (g^{[\bf n]})^{j b} \partial_ig^{[\bf n]}_{jb} (k^{[\bf n]})_j{}^j + O(t^{-1+\ve}).
\end{split}
\end{equation}

Combining \eqref{eq:approx.momentum.2} and \eqref{eq:approx.momentum.3}, we have
\begin{equation}\label{eq:approx.momentum.4}
-\frac{1}{2}(g^{[\bf n]})^{\ell b}(\partial_ig^{[\bf n]}_{jb}+\partial_jg^{[\bf n]}_{ib}-\partial_bg^{[\bf n]}_{ij})(k^{[\bf n]})_\ell{}^j = - \sum_{j=1}^3 \f 12 (g^{[\bf n]})^{j b} \partial_ig^{[\bf n]}_{jb} (k^{[\bf n]})_j{}^j + O(t^{-1+\ve}).
\end{equation}

Plugging \eqref{eq:approx.momentum.4} into \eqref{eq:approx.momentum.1}, using the estimate \eqref{eq:main.parametrix.k.bd}, and noting that by symmetry $(g^{[\bf n]})^{j b}(\partial_jg^{[\bf n]}_{\ell b}-\partial_bg^{[\bf n]}_{j\ell}) =0$, we obtain
\begin{equation}\label{eq:approx.momentum.5}
\nabla_j(k^{\bf[n]})_i{}^j-\partial_i\text{tr}k^{\bf[n]} =\frac{\partial_j\kappa_i{}^j}{t} - \sum_{j=1}^3 \f 12 (g^{[\bf n]})^{j b} \partial_ig^{[\bf n]}_{jb} (k^{[\bf n]})_j{}^j + \f 12(g^{[\bf n]})^{j b}\partial_\ell g^{[\bf n]}_{jb} (k^{[\bf n]})_i{}^\ell + O(t^{-1+\ve}).
\end{equation}

Finally, notice that the second and third terms in \eqref{eq:approx.momentum.5}, when $j\neq b$, contribute only $O(t^{-1+\ve})$. We have thus obtained 
\begin{align*}
\nabla_j(k^{\bf[n]})_i{}^j-\partial_i\text{tr}k^{\bf[n]}
=&\,\frac{\partial_j\kappa_i{}^j}{t}
-\sum_{\ell=1}^3\frac{1}{2}(g^{[\bf n]})^{\ell \ell}\partial_ig^{[\bf n]}_{\ell\ell}(k^{[\bf n]})_\ell{}^\ell
+  {\sum_{j,\ell=1}^3} \frac{1}{2}(g^{[\bf n]})^{j j}\partial_\ell g^{[\bf n]}_{jj}(k^{[\bf n]})_i{}^\ell
+O(t^{-1+\varepsilon})\\
=&\,\frac{\partial_j\kappa_i{}^j}{t}+\sum_{\ell=1}^3\bigg(\frac{1}{2}\frac{\partial_ic_{\ell\ell}}{c_{\ell\ell}}\frac{p_\ell}{t}+\frac{p_\ell\partial_ip_\ell}{t}\log t\bigg)+O(t^{-1+\varepsilon})\\
&+\sum_{j {,\ell}=1}^3\bigg(\frac{\kappa_i{}^\ell\partial_\ell p_j}{t}\log t+\mathbbm{1}_{\{\ell>i\}}\frac{\partial_\ell c_{jj}}{2c_{jj}}\frac{\kappa_i{}^\ell}{t}-\frac{1}{2}\frac{\partial_i c_{jj}}{c_{jj}}\frac{p_i}{t}\bigg)\\
=&\,\frac{1}{2t}\sum_{\ell=1}^3\bigg(\frac{\partial_ic_{\ell\ell}}{c_{\ell\ell}}(p_\ell-p_i)
+2\partial_\ell\kappa_i{}^\ell
+\mathbbm{1}_{\{\ell>i\}}\frac{\partial_\ell(c_{11}c_{22}c_{33})}{c_{11}c_{22}c_{33}}\kappa_i{}^\ell\bigg)+O(t^{-1+\varepsilon}),
\end{align*}
where in the last equality we use condition (2) in Theorem \ref{mainthm}.
The desired estimate now follows by employing condition (4) in Theorem \ref{mainthm}.
\end{proof}
%
%
Combining Proposition~\ref{prop:curv.id}, Lemma \ref{lem:asymconst} and Theorem \ref{thm:parametrix}, we deduce the following bounds for the relevant curvature components of ${^{(4)}g^{[{\bf n}]}}$.
\begin{proposition}\label{prop:approx.G}
Given $n\in \mathbb N$ such that $(n+1)\ve >2$ and $^{(4)}g^{\bf [n]}$ as in Theorem~\ref{thm:parametrix}, the following estimates hold:
\begin{align}
\label{est:curv.n} |\partial_x^\alpha {Ric_i}^j({^{(4)}g^{[{\bf n}]}})|,|\partial_x^\alpha Ric_{tt}({^{(4)}g^{[{\bf n}]}})|,|\partial_x^\alpha R({^{(4)}g^{[{\bf n}]}})|,|\partial_x^\alpha {G_i}^j({^{(4)}g^{[{\bf n}]}})|,|\partial_x^\alpha G_{tt}({^{(4)}g^{[{\bf n}]}})| \leq C_{\alpha,n}t^{-2+(n+1)\varepsilon},\\
\label{est:curv.n.dt} |\partial_x^\alpha\partial_t {Ric_i}^j({^{(4)}g^{[{\bf n}]}})|\leq C_{\alpha,n}t^{-3+(n+1)\varepsilon}, \\
\label{est:curv.n.better}|\partial_x^\alpha G_{ti}({^{(4)}g^{[{\bf n}]}})|,|\partial_x^\alpha Ric_{ti}({^{(4)}g^{[{\bf n}]}})|\leq C_{ {\alp},n}t^{- {1}+(n+1)\varepsilon},
\end{align}
for all $(t,x)\in(0,t_N ] \times \mathbb{T}^3$.
\end{proposition}
\begin{proof}
In this proof, the implicit constants in $\ls$ depend on $\alp$, $n$, $c_{ij}$ and $p_i$.

\pfstep{Step~0: Estimates for $Ric_i{ }^j({^{(4)}g^{[{\bf n}]}})$} According to \eqref{Ricij} in Proposition~\ref{prop:curv.id} and the estimates in \eqref{eq:main.parametrix.k.bd}, Lemma~\ref{lem:k.II} and Proposition~\ref{prop:approx.ev.1}, it follows that
\begin{equation}\label{eq:dxalp.Ric.g4.ij.precise}
|\rd_x^\alp [ Ric(^{(4)}g^{\bf [n]})_i{ }^j]|(t,x) \ls t^{-2+(n+2)\ve}|\log t|^{2+|\alp| },
\end{equation}
which clearly implies in particular the needed estimate for $Ric(^{(4)}g^{\bf [n]})_i{ }^j$ in \eqref{est:curv.n}.

Also, using \eqref{eq:main.parametrix.k.bd}, Lemma~\ref{lem:k.II} and Proposition~\ref{prop:approx.ev.1}, we also obtain the estimate \eqref{est:curv.n.dt} for $\partial_x^\alpha\partial_t {Ric_i}^j({^{(4)}g^{[{\bf n}]}})$.

It suffices then to show that the estimates for $\partial_x^\alpha Ric_{ti}({^{(4)}g^{[{\bf n}]}}),\partial_x^\alpha Ric_{tt}({^{(4)}g^{[{\bf n}]}})$ hold true, since all the remaining terms in \eqref{est:curv.n} are algebraic combinations of the previous three. 

\pfstep{Step~1: Deriving the ODEs} By virtue of the contracted second Bianchi identity we have:
\begin{equation}\label{eq:Bianchi.ODE.1}
\begin{split}
\partial_tRic_{ti}({^{(4)}g^{[{\bf n}]}})=&\,D_tRic_{ti}({^{(4)}g^{[{\bf n}]}}) + \f 12 (g^{[\bf n]})^{j\ell} \rd_t g^{[\bf n]}_{i\ell} Ric_{tj}({^{(4)}g^{[{\bf n}]}})\\
=&-\frac{1}{2}\partial_iR({^{(4)}g^{[{\bf n}]}})+D_j{Ric_i}^j({^{(4)}g^{[{\bf n}]}}) + \f 12 (g^{[\bf n]})^{j\ell} \rd_t g^{[\bf n]}_{i\ell} Ric_{tj}({^{(4)}g^{[{\bf n}]}})\\
=&\,\frac{1}{2}\partial_iRic_{tt}({^{(4)}g^{[{\bf n}]}}) \underbrace{-  \f 12 (g^{[\bf n]})^{j\ell} \rd_t g^{[\bf n]}_{j\ell} Ric_{ti}({^{(4)}g^{[{\bf n}]}})}_{=:\mathrm{I}}\\
&\underbrace{-\frac{1}{2}\partial_i{Ric_j}^j({^{(4)}g^{[{\bf n}]}})+\nabla_j{Ric_i}^j({^{(4)}g^{[{\bf n}]}})}_{=:\mathrm{II}},\\
\end{split}
\end{equation}
\begin{equation}\label{eq:Bianchi.ODE.2}
\begin{split}
\partial_tRic_{tt}({^{(4)}g^{[{\bf n}]}})=&\,D_tRic_{tt}({^{(4)}g^{[{\bf n}]}})=-\frac{1}{2}\partial_t R({^{(4)}g^{[{\bf n}]}})+D^jRic_{tj}({^{(4)}g^{[{\bf n}]}})\\
=&\,\underbrace{\frac{1}{2}\partial_tRic_{tt}({^{(4)}g^{[{\bf n}]}})}_{=:\mathrm{III}}  \underbrace{ - \f 12 (g^{[\bf n]})^{j\ell} \rd_t g^{[\bf n]}_{j\ell} Ric_{tt}({^{(4)}g^{[{\bf n}]}})}_{=:\mathrm{IV}}+\underbrace{\nabla^j(Ric_t)_j({^{(4)}g^{[{\bf n}]}})}_{=:\mathrm{V}}\\
&\underbrace{-\frac{1}{2}\partial_t{Ric_j}^j({^{(4)}g^{[{\bf n}]}})}_{=:\mathrm{VI}} \underbrace{-\f 12 (g^{[\bf n]})^{j\ell} \rd_t g^{[\bf n]}_{i\ell} {Ric_j}^i({^{(4)}g^{[{\bf n}]}})}_{=:\mathrm{VII}}.
\end{split}
\end{equation}
where $D$ denotes the Levi--Civita connection of ${^{(4)}g^{[{\bf n}]}}$, and $\nabla^j(Ric_t)_j({^{(4)}g^{[{\bf n}]}})$ means that we take $(Ric_t)_j$ as a tensor field tangent to the constant-$t$ hypersurfaces and then differentiate with the connection $\nabla$, i.e.
$$\nab^j (Ric_t)_j({^{(4)}g^{[{\bf n}]}}) = (g^{\bf [n]})^{ij} \rd_i (Ric_t)_j({^{(4)}g^{[{\bf n}]}}) - (g^{\bf [n]})^{ij} (\Gamma^{\bf [n]})_{ij}^\ell (Ric_t)_\ell({^{(4)}g^{[{\bf n}]}}).$$

We now estimate the terms in \eqref{eq:Bianchi.ODE.1} and \eqref{eq:Bianchi.ODE.2}. For term $\mathrm{I}$ in \eqref{eq:Bianchi.ODE.1}, we use \eqref{eq:main.parametrix.k.bd} and \eqref{eq:main.parametrix.k.and.II.1} in Theorem~\ref{thm:parametrix} to obtain
\begin{equation}\label{eq:prop.constraint.ODE.I}
\mathrm{I} = [-\f 1t + O(t^{-1+\ve})] Ric_{ti}({^{(4)}g^{[{\bf n}]}}).
\end{equation}

The first term in $\mathrm{II}$ can be directly estimated by \eqref{eq:dxalp.Ric.g4.ij.precise}. To handle the second term in $\mathrm{II}$, we compute using the form of the metric and \eqref{eq:main.parametrix.a.bd} to obtain
\begin{align*}
\nabla_j{Ric_i}^j({^{(4)}g^{[{\bf n}]}})=&\,\partial_j{Ric_i}^j({^{(4)}g^{[{\bf n}]}})-(\Gamma^{\bf [n]})_{ij}^\ell {Ric_\ell}^j({^{(4)}g^{[{\bf n}]}})+(\Gamma^{[\bf n]})_{j\ell}^j{Ric_i}^\ell({^{(4)}g^{[{\bf n}]}})\\
(\Gamma^{[\bf n]})_{j\ell}^j=&\,\frac{1}{2}(g^{[\bf n]})^{jb}(\partial_j g^{\bf [n]}_{b\ell}+\partial_\ell g_{bj}^{\bf [n]}-\partial_bg^{[\bf n]}_{j\ell})=O(|\log t|)\\
(\Gamma^{\bf [n]})_{ij}^\ell {Ric_\ell}^j({^{(4)}g^{[{\bf n}]}})=&\,\frac{1}{2}(g^{[\bf n]})^{\ell b}(\partial_i g^{\bf [n]}_{bj} + \partial_j g_{bi}^{\bf [n]}-\partial_bg^{[\bf n]}_{ij}){Ric_\ell}^j({^{(4)}g^{[{\bf n}]}})\\
=&\,[O(|\log t|)-(g^{[\bf n]})^{\ell b}\partial_bg^{[\bf n]}_{ij}]{Ric_\ell}^j({^{(4)}g^{[{\bf n}]}})\\
(g^{[\bf n]})^{\ell b}\partial_bg^{[\bf n]}_{ij}{Ric_\ell}^j({^{(4)}g^{[{\bf n}]}})=&\,O_b(|\log t|)(g^{[\bf n]})^{\ell b}g^{[\bf n]}_{ij}{Ric_\ell}^j({^{(4)}g^{[{\bf n}]}})=O_b(|\log t|){Ric_i}^b({^{(4)}g^{[{\bf n}]}})
\end{align*} 
Combining all the above and using \eqref{eq:dxalp.Ric.g4.ij.precise}, we obtain
\begin{equation}\label{eq:prop.constraint.ODE.II}
\mathrm{II} = O( t^{-2+(n+1)\ve}).
\end{equation}

Plugging \eqref{eq:prop.constraint.ODE.I} and \eqref{eq:prop.constraint.ODE.II} into \eqref{eq:Bianchi.ODE.1} yields
\begin{equation}\label{RictiODE.1.1}
\partial_tRic_{ti}({^{(4)}g^{[{\bf n}]}})+[\frac{1}{t}+O(t^{-1+\varepsilon})]Ric_{ti}({^{(4)}g^{[{\bf n}]}})=
\frac{1}{2}\partial_iRic_{tt}({^{(4)}g^{[{\bf n}]}})+O(t^{-2+(n+1)\varepsilon}).
\end{equation}

Equation \eqref{eq:Bianchi.ODE.2} can be treated similarly. We subtract $\mathrm{III}$ to the LHS, use \eqref{eq:main.parametrix.k.bd} and \eqref{eq:main.parametrix.k.and.II.1} to control the coefficients in $\mathrm{IV}$ and $\mathrm{VII}$, keep the term $\mathrm{IV}$, and use \eqref{est:curv.n.dt} for term $\mathrm{VI}$ so that we obtain
\begin{equation}\label{RictiODE.1.2}
\partial_tRic_{tt}({^{(4)}g^{[{\bf n}]}})+[\frac{2}{t}+O(t^{-1+\varepsilon})]Ric_{tt}({^{(4)}g^{[{\bf n}]}})=2 {\nab^j(Ric_t)_j}({^{(4)}g^{[{\bf n}]}})+O(t^{-3+(n+1)\varepsilon}).
\end{equation}

In a similar way, we obtain the equations for higher derivatives analogous to \eqref{RictiODE.1.1} and \eqref{RictiODE.1.2}. After putting in an integrating factor, the equations read
\begin{align}\label{RictiODE.2}
\begin{split} 
\partial_t( t \rd_x^\alp Ric_{ti}({^{(4)}g^{[{\bf n}]}}))= &\: \frac{t}{2}\partial_i \rd_x^\alp Ric_{tt}({^{(4)}g^{[{\bf n}]}})+ O(t^{\varepsilon}) \sum_{|\bt|\leq |\alp|} \rd_x^\bt Ric_{ti}({^{(4)}g^{[{\bf n}]}}) + O(t^{-1+(n+1)\varepsilon}),\\
\partial_t( t^2 \rd_x^\alp Ric_{tt}({^{(4)}g^{[{\bf n}]}})) =&\: 2t^2 \rd_x^\alp \nab^j(Ric_t)_j({^{(4)}g^{[{\bf n}]}}) + O(t^{1+\varepsilon}) \sum_{|\bt|\leq |\alp|} \rd_x^\bt Ric_{tt}({^{(4)}g^{[{\bf n}]}})+O(t^{-1+(n+1)\varepsilon}).
\end{split} 
\end{align}
\pfstep{ {Step~2: Solving the ODEs}}  {We will view the two equations in \eqref{RictiODE.2} as ODEs  {in $t$}. In particular we will not be concerned with the loss of derivatives since we have bounds for all order of derivatives of the approximate solutions.}

Note that Lemmas~\ref{lem:k.II}, \ref{lem:asymconst} and the identities \eqref{Ricti},  {\eqref{Rictt}} imply the estimates:
\begin{align}\label{Ricticond}
|\partial_x^\alpha Ric_{ti}({^{(4)}g^{[{\bf n}]}})|\lesssim t^{-1+\varepsilon},\qquad |\partial_x^\alpha Ric_{tt}({^{(4)}g^{[{\bf n}]}})|\lesssim t^{-2+\varepsilon}.
\end{align}

In particular, this means that the initial data (at $\{t=0\}$) for $t\rd_x^\alp Ric_{ti}({^{(4)}g^{[{\bf n}]}})$ and $t^2 \rd_x^\alp Ric_{tt}({^{(4)}g^{[{\bf n}]}})$ both vanish.
 {Now since $|\partial_x^\alpha Ric_{ti}({^{(4)}g^{[{\bf n}]}})|\lesssim t^{-1+\varepsilon}$ (for all $\alp$), it follows that $|\rd_x^\alp \nab^i (Ric_t)_i({^{(4)}g^{[{\bf n}]}})| \ls t^{-3 + 2\ve}$ (for this we simply use that $|\rd_x^\alp (g^{\bf [n]})^{i\ell }|,\,|\rd_x^\alp [(g^{\bf [n]})^{i \ell}(\Gamma^{\bf [n]})_{i \ell}^j] | \ls t^{-2+\ve}$).
Hence, integrating the second equation in \eqref{RictiODE.2} and using Gr\"onwall's inequality, we obtain
 }
\begin{align}\label{Rictt.imp.1}
t^2| {\rd_x^\alp}Ric_{tt}({^{(4)}g^{[{\bf n}]}})|\lesssim t^{2\varepsilon}+t^{(n+1)\varepsilon}\qquad\Longrightarrow\qquad| {\rd_x^\alp}Ric_{tt}({^{(4)}g^{[{\bf n}]}})|\lesssim t^{-2+2\varepsilon} {.}
\end{align}
 {Plugging this estimate into the first equation in \eqref{RictiODE.2}, we then obtain using Gr\"onwall's inequality}
\begin{align}\label{Ricti.imp.1}
t| {\rd_x^\alp}Ric_{ti}({^{(4)}g^{[{\bf n}]}})|\lesssim t^{2\varepsilon}+t^{(n+1)\varepsilon}\qquad\Longrightarrow\qquad| {\rd_x^\alp}Ric_{ti}({^{(4)}g^{[{\bf n}]}})|\lesssim t^{-1+2\varepsilon}
\end{align}
Notice that  {\eqref{Rictt.imp.1} and \eqref{Ricti.imp.1} improves over} \eqref{Ricticond}.  {We now repeat the above argument, but plugging in these improve estimates to obtain (assuming $n\geq 2$)}
$$  {|\rd_x^\alp Ric_{tt}({^{(4)}g^{[{\bf n}]}})|\ls t^{-2+3\varepsilon}, \qquad |\rd_x^\alp Ric_{ti}({^{(4)}g^{[{\bf n}]}})|\ls t^{-1+3\varepsilon}} .$$
 {Iterating this argument then gives the desired estimates. (The rate for $\rd_x^\alp Ric_{ti}$ is limited by the last term on the RHS of the first equation in \eqref{RictiODE.2}.)} This completes the proof of the proposition. \qedhere
\end{proof}

\section{Construction of an actual solution}\label{sec:actual}

We continue to work under the assumptions of Theorem~\ref{mainthm} and take $g^{\bf [n]}$ and $k^{\bf [n]}$ as constructed in the beginning of Section~\ref{sec:parametrix} (so that for appropriately chosen $t_N$, the estimates in Theorem~\ref{thm:parametrix} and Proposition~\ref{prop:approx.G} hold).

The main result of this section will be to prove existence of a solution to a system of reduced equations (to be introduced below in \eqref{eq:hyp.sys} of Section~\ref{sec:derivation.of.equations}). See Theorem~\ref{thm:main.reduced} for the precise statement of the main result, and see the rest of Section~\ref{subsec:locexist} for a discussion of the proof of Theorem~\ref{thm:main.reduced} and an outline of the later parts of the section.


\subsection{ {Deriving the reduced equations}}\label{sec:derivation.of.equations} 

 {As already described in Section~\ref{sec:EEinsyncoord} in the introduction, we will control $k_i{ }^j$ using a second-order wave-like equation. In this subsection, we derive the equation that we will use.}

By \eqref{Ricij} in Proposition~\ref{prop:curv.id}, if a metric takes the form \eqref{metricansatz}, and $k$ is the second fundamental form, then
\begin{equation}\label{eq:R00}
{Ric({^{(4)}g})_i}^j=-\rd_t k_i{ }^j + Ric(g)_i{ }^j + k_\ell{ }^\ell k_i{ }^j.
\end{equation}
Taking a $\rd_t$ derivative of \eqref{eq:R00}, we obtain
$$\rd_t {Ric({^{(4)}g})_i}^j = -\rd_t^2 k_i{ }^j + \rd_t Ric(g)_i{ }^j + \rd_t[k_\ell{ }^\ell k_i{ }^j].$$
To compute $\rd_t Ric(g)_i{ }^j$, we use the \emph{variation of Ricci formula} (see for example equation (2.31) in \cite{CLN}) and the fact $\rd_t g_{ij} = -2 k_{ij}$:
$$\rd_t Ric {(g)}_{ij} = \Delta_L k_{ij} + \nabla^2_{ij} k_\ell{ }^\ell - \nabla_i (\mathrm{div} k)_j -\nabla_j (\mathrm{div} k)_i,$$
where $\Delta_L$ is the Lichnerowicz Laplacian (on symmetric $2$-tensors) given by 
$$\Delta_L v_{ij} := \Delta_g v_{ij} + 2 Riem {(g)}^m{ }_{ij}{ }^\ell v_{m\ell} - Ric {(g)}_i{ }^\ell v_{j\ell} - Ric {(g)}_j{ }^\ell v_{i\ell}.$$
Using again $\rd_t g_{ij} = -2 k_{ij}$, it follows that
\begin{equation}\label{eq:variation.of.Ricci}
\rd_t Ric {(g)}_i{ }^j = \Delta_g k_i{ }^j + 2 R {iem(g)}^m{ }_{i}{ }^{j}{ }_{\ell} k_m{ }^{\ell} + Ric {(g)}_i{ }^\ell k_{\ell}{ }^j - Ric {(g)}_\ell{ }^j k_{i}{ }^{\ell} + \nabla_i\nabla^j k_\ell{ }^\ell - \nabla_i (\mathrm{div} k)^j -\nabla^j (\mathrm{div} k)_i.
\end{equation}

We will further analyze two groups of terms on the RHS of \eqref{eq:variation.of.Ricci}:
\begin{enumerate}
\item Denoting $\mathcal G_i:= G_{ti} {(g^{(4)})}$ and considering it as a tensor on $\{ t = \mbox{constant}\}$, we have
\begin{equation}\label{eq:RHS.of.Ricci.variation}
\nabla_i\nabla^j k_\ell{ }^\ell - \nabla_i (\mathrm{div} k)^j -\nabla^j (\mathrm{div} k)_i = (g^{-1})^{j\ell}\nabla_i \mathcal G_\ell + \nabla^j \mathcal G_i - \nabla_i\nabla^j k_\ell{ }^\ell.
\end{equation}

\item In three dimensions, the Riemann curvature tensor can be expressed in terms of the Ricci curvature tensor (see (1.62) in \cite{CLN}):
\begin{equation}\label{Riem=Ric}
\begin{split}
 {Riem(g)}^m{ }_i{ }^j{ }_\ell =&\:  -Ric {(g)}^{mj}g_{i\ell} + Ric {(g)}_{\ell}{}^m \de_i^j - Ric {(g)}_{i\ell} (g^{-1})^{mj} \\
&\: + Ric {(g)}_i{}^j\de_{\ell}^m - \f 12 R {(g)} (\de^m_\ell \de_i^j - (g^{-1})^{mj} g_{i\ell}),
\end{split}
\end{equation}
 {where $R(g)$ denotes the scalar curvature of $g$.}
Therefore, the terms
$$2 R {iem(g)}^m{ }_{i}{ }^{j}{ }_{\ell} k_m{ }^{\ell} + Ric {(g)}_i{ }^\ell k_{\ell}{ }^j - Ric {(g)}_\ell{ }^j k_{i}{ }^{\ell}$$
can be written as some linear combinations of contractions of $Ric {(g)}$ and $k$.  {Using again \eqref{eq:R00}}, we can replace $Ric(g)_i{ }^j$ by ${Ric({^{(4)}g})_i}^j +\rd_t k_i{ }^j - k_\ell{ }^\ell k_i{ }^j$.
\end{enumerate}

It therefore follows that the second fundamental form $k$ verifies the following equation:
\begin{equation}\label{eq:k.wave.prelim}
\begin{split}
&\: - \rd_t^2 k_i{ }^j + \Delta_g k_i{ }^j -\nab_i\nab^j k_\ell{ }^\ell + (k\star k\star k)_i{ }^j + (\rd_t k \star k)_i{ }^j \\
=&\:  {-\partial_t Ric_i{}^j({^{(4)}g})+ \nabla_i \mathcal{G}^j+\nabla^j\mathcal{G}_i
-3k_i{}^mRic_m{}^j({^{(4)}g}) + 2\de_i^jk_m{}^\ell Ric_{\ell}{}^m({^{(4)}g}) -k_\ell{}^j Ric_i{}^\ell({^{(4)}g})}\\
&\:  {+ 2k_\ell{}^\ell Ric_i{}^j({^{(4)}g})
-(k_\ell{}^\ell \de_i^j - k_i{}^j) Ric_m{}^m({^{(4)}g}), }
\end{split}
\end{equation}
where
\begin{align}\label{kstark}
\begin{split}
(k\star k\star k)_i{ }^j:=&-2k_a{}^a\big[-g^{ma}k_a{}^jg_{il}+k_\ell{}^m\delta_i{}^j-g_{a\ell}k_i{}^ag^{mj}+k_i{}^j\delta_l{}^m-\frac{1}{2}k_a{}^a(\delta_\ell{}^m\delta_i{}^j-g^{mj}g_{il})\big]k_m{}^\ell\\
=&\,4k_\ell{}^\ell k_a{}^jk_i{}^a-2k_a{}^a(k_\ell{}^mk_m{}^\ell)\delta_i{}^j+(k_a{}^a)^3\delta_i{}^j-3(k_\ell{}^\ell)^2k_i{}^j\\
(\rd_t k \star k)_i{ }^j:=&\,\partial_t(k_\ell{}^\ell k_i{}^j)
-2\partial_tk_a{}^jk_i{}^a+2\partial_tk_\ell{}^m\delta_i{}^jk_m{}^\ell 
-2\partial_tk_i{}^ak_a{}^j+2\partial_tk_i{}^jk_\ell{}^\ell \\
&-\partial_tk_a{}^a\delta_i{}^jk_\ell{}^\ell+\partial_tk_a{}^ak_i{}^j
+\partial_tk_i{}^\ell k_\ell{}^j-\partial_tk_\ell{}^jk_i{}^\ell\\
=&-3\partial_tk_a{}^jk_i{}^a+\partial_t(k_\ell{}^m k_m{}^\ell)\delta_i{}^j 
-\partial_tk_i{}^ak_a{}^j+2\partial_tk_\ell{}^\ell k_i{}^j+3\partial_tk_i{}^jk_\ell{}^\ell 
-\frac{1}{2}\partial_t(k_a{}^a)^2\delta_i{}^j
\end{split}
\end{align}
We note that the terms $k\star k\star k$ and $\rd_t k \star k$ satisfy
\begin{equation}\label{eq:trace.of.nonlinear}
(k\star k\star k)_i{ }^i + (\rd_t k \star k)_i{ }^i = \rd_t |k|^2 - 2 k_\ell{ }^\ell |k|^2 + 2k_i{ }^i\rd_t k_\ell{ }^\ell.
\end{equation}

In particular, if $^{(4)}g$ solves the Einstein vacuum equations, then
\begin{equation}\label{eq:k.wave}
\rd_t^2 k_i{ }^j = \Delta_g k_i{ }^j -\nab_i\nab^j k_\ell{ }^\ell + (k\star k\star k)_i{ }^j + (\rd_t k \star k)_i{ }^j.
\end{equation}
The equation \eqref{eq:k.wave} is almost a wave equation for $k$, except that there is a top order $\nab_i\nab^j k_\ell{ }^\ell$ term on the RHS. To proceed we think of $h = k_\ell{ }^\ell$ as an \emph{independent} variable. If the Einstein vacuum equations were satisfied, then \eqref{Rictt} in Proposition~\ref{prop:curv.id} imposes that $\rd_t h = |k|^2$. It is therefore reasonable to look for a solution to the Einstein vacuum equations by solving the following coupled system of equations:
\begin{equation}\label{eq:hyp.sys}
\begin{split}
\rd_t h =&\:  |k|^2,\\
\rd_t^2 k_i{ }^j =&\: \Delta_g k_i{ }^j -\nab_i\nab^j h + (k\star k\star k)_i{ }^j + (\rd_t k \star k)_i{ }^j,\\
\rd_t g_{ij} = &\:  {- k_i{ }^{\ell} g_{j\ell} - k_j{ }^{\ell} g_{i\ell}}.\\
\end{split}
\end{equation}
Remark that given a solution to \eqref{eq:hyp.sys}, it follows that $g^{-1}$ satisfies
\begin{equation}\label{eq:g-1.transport}
\rd_t (g^{-1})^{ij} =  {k_\ell{ }^j (g^{-1})^{i\ell} + k_\ell{ }^i (g^{-1})^{j\ell}}.
\end{equation}

Our strategy will be to solve the system \eqref{eq:hyp.sys} and then a posteriori justify that it is indeed a solution to the Einstein vacuum equations (in Section \ref{subsec:vanEVE}).

\subsection{Notations}
Before we proceed, we introduce some notations.

In the following we will consider (at least) two spacetime metrics $^{(4)} g = -\ud t^2 + g_{ij} \ud x^i \ud x^j$ and $^{(4)} g^{\bf [n]} = -\ud t^2 + g^{\bf [n]}_{ij} \ud x^i \ud x^j$
on the domain $I_t \times \mathbb T^3$ (where $I_t\subset \mathbb R$ is an interval, possibly open, closed or half-open). 

We make the following definitions assuming we are given such $I_t$, $^{(4)} g $ and $^{(4)} g^{\bf [n]}$.

\begin{definition}[Constant-$t$ hypersurfaces]
Given $t\in I_t$ define
$$\Sigma_t:= \{(\tau, x): \tau = t,\, x\in \mathbb T^3\}.$$
\end{definition}

\begin{definition}[Connections]\label{def:connections}
\begin{enumerate}
\item Denote by $\nab$ the Levi--Civita connection of $g$, and by $\nab^{\bf [n]}$ the Levi--Civita connection of $g^{\bf[n]}$.
\item Denote $\nab^{(d)} := \nab - \nab^{\bf [n]}$. Remark that $\nab^{(d)}$ is a $(1,2)$-tensor.
\item Let $r \in \mathbb N$ and $\mathcal T$ be an $(m,l)$-tensor. Define $\nab^{(r)} T$ to be the $(m,l+r)$-tensor given by
$$(\nab^{(r)} T)^{j_i\ldots j_m}_{a_1\ldots a_r i_1\ldots i_l} = \nab_{a_1}\cdots \nab_{a_r} T^{j_i\ldots j_m}_{i_1\ldots i_l}.$$
\end{enumerate}
\end{definition}

\begin{definition}[Norms]
\begin{enumerate}
\item Given two rank $(m,l)$ tensors $\mathcal T^{(1)}$ and $\mathcal T^{(2)}$, define the inner product
$$\langle \mathcal T^{(1)},\,\mathcal T^{(2)}\rangle_g := (g^{-1})^{i_1b_1}\ldots (g^{-1})^{i_lb_l}g_{j_1c_1}\ldots g_{j_mc_m}(\mathcal{T}^{(1)})^{j_i\ldots j_m}_{i_1\ldots i_l}(\mathcal{T}^{(2)})^{c_i\ldots c_m}_{b_1\ldots b_l}.$$
\item Given a rank $(m,l)$ tensor $\mathcal T$, define 
$$|\mathcal T|_g^2 := \langle \mathcal T,\,\mathcal T\rangle_g = (g^{-1})^{i_1b_1}\ldots (g^{-1})^{i_lb_l}g_{j_1c_1}\ldots g_{j_mc_m}\mathcal{T}^{j_i\ldots j_m}_{i_1\ldots i_l}\mathcal{T}^{c_i\ldots c_m}_{b_1\ldots b_l}.$$
\item Given a tensor $\mathcal T$ and $p\in [1,+\infty)$, define
$$\|\mathcal T\|_{L^p(\Sigma_t,g)} := (\int_{\Sigma_t} |\mathcal T|_g^p \, \mathrm{vol}_{\Sigma_t})^{\f 1p},$$
where $\mathrm{vol}_{\Sigma_t} = \sqrt{\det g}\,\ud x$ is the volume form induced by the metric $g$.

For $p=+\infty$, define
$$\|\mathcal T\|_{L^\i(\Sigma_t,g)} := \mathrm{ess\, sup}_{x\in \mathbb T^3}  |\mathcal T|_g(t,x).$$
\item For $r \in \mathbb N \cup \{0\}$ and $p \in [1,+\infty]$, define the geometric Sobolev space 
$$\|\mathcal T\|_{W^{r,p}(\Sigma_t, g)}:= \sum_{r' = 0}^{r} \|\nab^{(r')} \mathcal T\|_{L^p(\Sigma_t,g)}.$$
\item For $r \in \mathbb N \cup \{0\}$ and $p \in [1,+\infty]$, define the \emph{homogeneous} geometric Sobolev space 
$$\|\mathcal T\|_{\dot W^{r,p}(\Sigma_t, g)}:= \|\nab^{(r)} \mathcal T\|_{L^p(\Sigma_t,g)}.$$
\item For $r \in \mathbb N \cup \{0\}$, define
$$H^r(\Sigma_t,g) := W^{r,2}(\Sigma_t, g),\quad \dot H^r(\Sigma_t,g) := \dot W^{r,2}(\Sigma_t, g).$$
\item Define the norm $t^\alp L^2(\Sigma_t,g)$ (for $\alp \in \mathbb R\setminus \{0\}$) by 
$$\| \mathcal T'\|_{t^\alp L^2(\Sigma_t,g)} := t^{-\alp} \|\mathcal T \|_{L^2(\Sigma_t,g)}.$$
\item Given any two Banach spaces $X$ and $Y$, the vector spaces $X+Y=\{x+y: x\in X,\,y\in Y\}$ and $X\cap Y$ are endowed with Banach space structures with norms
$$\|v\|_{X+Y} := \inf_{v=x+y,\,(x,y)\in X\times Y} (\|x\|_X + \|y\|_Y),\quad \|v\|_{X\cap Y} := \|v\|_X + \|v\|_Y.$$
\item Finally, define $L^p(\Sigma_t,g^{\bf [n]})$, $W^{r,p}(\Sigma_t,g^{\bf [n]})$ and $\dot W^{r,p}(\Sigma_t,g^{\bf [n]})$ etc.~as above but with $g$ replaced by $g^{\bf [n]}$ (and $\nab$ replaced by $\nab^{\bf [n]}$).
\end{enumerate}
\end{definition}

\subsection{Existence of solutions to \eqref{eq:hyp.sys} and the main steps of the proof}\label{subsec:locexist}

 Our first step of the proof of Theorem \ref{mainthm} is to build a solution to \eqref{eq:hyp.sys}.  The following is the main existence result for \eqref{eq:hyp.sys}, whose proof will occupy the remainder of the section.

\begin{theorem}\label{thm:main.reduced}
For every $s ,\, N_0 \in \mathbb N$ obeying $s\geq 5$, there exists $n_{N_0,s} \in \mathbb N$ sufficiently large such that for any $n\geq n_{N_0,s}$, there exist $T_{N_0,s,n} > 0$ sufficiently small and a solution $(g,h,k)$ to \eqref{eq:hyp.sys} in $(0,T_{N_0,s,n}]\times \mathbb T^3$ which satisfy the following estimates:
\begin{equation}\label{eq:main.reduced.est.in.thm}
\begin{split}
\sum_{r=0}^s t^{2r} \| k^{(d)} \|_{H^r(\Sigma_t,g)}^2 &\: + \sum_{r=0}^{s-1} t^{2(r+1)}\|\rd_t k^{(d)}\|^2_{H^r(\Sigma_t,g)} + \sum_{r=0}^{s+1} t^{2r} \| h^{(d)} \|_{H^r(\Sigma_t,g)}^2 \\
&\: + \sum_{r=0}^{s+1} t^{2(r-1)} (\| g^{(d)} \|_{H^r(\Sigma_t,g)}  + \| (g^{-1})^{(d)} \|_{H^r(\Sigma_t,g)}) \leq t^{2N_0+ 2s},
\end{split}
\end{equation}
where $k^{(d)}=k-k^{[\bf n]},h^{(d)}=h-h^{[\bf n]},g^{(d)}=g-g^{\bf[n]},(g^{(d)})^{-1}=g^{-1}-(g^{\bf[n]})^{-1}$.
Moreover, $k_{ij} = g_{\ell j}k_i{ }^\ell$ is symmetric in $i$ and $j$.
\end{theorem}

We will prove Theorem~\ref{thm:main.reduced} with the following steps (see the conclusion of the proof in Section~\ref{sec:conclusion.of.actual.wave}):
\begin{enumerate}
\item For $T_{\mathrm{aux}}>0$ (with $T_{\mathrm{aux}} \ll T_{N_0,s,n}$), we construct \emph{local} solutions to \eqref{eq:hyp.sys} in $[T_{\mathrm{aux}}, T_{\mathrm{aux}} + \de) \times \mathbb T^3$ (with $\de$ potentially depending on $T_{\mathrm{aux}}$) (\textbf{Lemma~\ref{lem:locexist}}).
\item For $s$, $N$, $n$ and $T_{N_0,s,n}$ as in  Theorem~\ref{thm:main.reduced}, we prove \emph{uniform} estimates to show that the solution can be extended to $[T_{\mathrm{aux}}, T_{N_0,s,n}]$ . This is carried out in a bootstrap argument and is the main step (\textbf{Theorem~\ref{thm:bootstrap}, Corollary~\ref{cor:bootstrap}}).
\item Using a \emph{compactness} argument, we take a sequence of auxiliary times $(T_{\mathrm{aux}})_i \to 0^+$ and extract a subsequence of solutions converging to a limiting solution to \eqref{eq:hyp.sys} on $(0,T_{N_0,s,n}]\times \mathbb T^3$ (\textbf{Proposition~\ref{prop:AA.main}}).
\end{enumerate}

We will further elucidate these steps in the subsubsections below. Most of the proofs will then be given in later subsections.

\subsubsection{Step~1: A local solution}
We begin with the following local existence result for \eqref{eq:hyp.sys}:
\begin{lemma}[Local existence]\label{lem:locexist}
For every  {$T_{\mathrm{aux}}>0$ sufficiently small} and $n\in \mathbb N$, there exist a $\delta>0$ (depending a priori both on $T_{\mathrm{aux}}$ and $n$) and a unique smooth solution $(g^{\mathrm{aux}}, k^{\mathrm{aux}}, h^{\mathrm{aux}})$ to \eqref{eq:hyp.sys} in $[T_{\mathrm{aux}},T_{\mathrm{aux}}+\delta]\times \mathbb T^3$, such that at $t= T_{\mathrm{aux}}$, $(g^{\mathrm{aux}}, k^{\mathrm{aux}}, h^{\mathrm{aux}})$ attains the following prescribed values:
\begin{align*}
\begin{split}
g_{ij}^{\mathrm{aux}} \restriction_{t = T_{\mathrm{aux}}} = g_{ij}^{\bf [n]}\restriction_{t = T_{\mathrm{aux}}},\quad  h^{\mathrm{aux}} \restriction_{t = T_{\mathrm{aux}}} = (k^{\bf [n]})_{i}{ }^i\restriction_{t = T_{\mathrm{aux}}},\\
(k^{\mathrm{aux}})_{i}{ }^j\restriction_{t = T_{\mathrm{aux}}} = (k^{\bf [n]})_{i}{ }^j\restriction_{t = T_{\mathrm{aux}}},\quad (\partial_t k^{\mathrm{aux}})_{i}{ }^j \restriction_{t = T_{\mathrm{aux}}} = (\partial_t k^{\bf [n]})_{i}{ }^j \restriction_{t = T_{\mathrm{aux}}}.
\end{split}
\end{align*}
Moreover, $g_{ij}^{\mathrm{aux}}=g_{ji}^{\mathrm{aux}}$.
\end{lemma}

Such a local existence result is almost standard. The only issue is that the second equation of the system \eqref{eq:hyp.sys} contains the term $\nabla_i\nabla^j h$ on the RHS, which seems to ``have one derivative too many''. This issue can be treated by deriving elliptic estimates for $h$, by commuting $\partial_th=|k|^2$ with $\Delta_g$ and using the wave equation for $k$, see discussions in Section~\ref{sec:EEinsyncoord} and Lemma~\ref{lem:remove.modified.2}.  {We will use this result but will omit its straightforward proof.}

Once existence is obtained, since $g^{\mathrm{aux}}_{ij}$ is symmetric at $t = T_{\mathrm{aux}}$ and $\rd_t (g_{ij}^{\mathrm{aux}} - g_{ji}^{\mathrm{aux}}) = 0$, it immediately follows that $g_{ij}^{\mathrm{aux}}=g_{ji}^{\mathrm{aux}}$.



\subsubsection{Step~2: The main bootstrap argument}\label{sec:step.2.bootstrap}

Our next step is to prove a uniform time of existence independent of $T_{\mathrm{aux}}$. To state the result, let us define, for $(g^{\mathrm{aux}}, k^{\mathrm{aux}}, h^{\mathrm{aux}})$ as in Lemma~\ref{lem:locexist},
\begin{align}
g^{(d)}_{ij}:= g^{\mathrm{aux}}_{ij} - g^{\bf [n]}_{ij},\quad ((g^{-1})^{(d)})^{ij}:= ((g^{\mathrm{aux}})^{-1})^{ij} - ((g^{\bf [n]})^{-1})^{ij},\label{adkdhd.1}\\
 (k^{(d)})_i{ }^j:= (k^{\mathrm{aux}})_i{ }^j - (k^{\bf [n]})_i{ }^j, \quad h^{(d)}:= h^{\mathrm{aux}}- h^{\bf [n]}.\label{adkdhd.2}
\end{align}
We stipulate that the metric $g^{\mathrm{aux}}_{ij}$ takes the form \eqref{metricansatz} and define $a^{\mathrm{aux}}_{ij}$ according to \eqref{metricansatz}.

Introduce the following \textbf{bootstrap assumptions}:
\begin{align}
\max_{i,j} |a_{ij}^{\mathrm{aux}} - c_{ij}|(t,x) \leq t^{\f {\ve}{2}}, \label{eq:BA1}\\
\| g^{(d)}\|_{W^{s-1,\infty}(\Sigma_t,g^{\mathrm{aux}})} + \| (g^{-1})^{(d)}\|_{W^{s-1,\infty}(\Sigma_t,g^{\mathrm{aux}})} \leq 1, \label{eq:BA2} \\ 
\| g^{(d)}\|_{H^{s+1}(\Sigma_t,g^{\mathrm{aux}})} + \| (g^{-1})^{(d)}\|_{H^{s+1}(\Sigma_t,g^{\mathrm{aux}})} \leq t^{\f 52}, \label{eq:BA3} \\
\| h^{(d)}\|_{H^{s+1}(\Sigma_t,g^{\mathrm{aux}})} + \| k^{(d)}\|_{H^{s}(\Sigma_t,g^{\mathrm{aux}})} + \| \rd_t k^{(d)}\|_{H^{s-1}(\Sigma_t,g^{\mathrm{aux}})} \leq t^{\f 52}. \label{eq:BA4}
\end{align}

The following is the main bootstrap theorem, whose proof constitutes most of this section (in Sections~\ref{sec:bootstrap.outline}--\ref{sec:conclusion.of.bootstrap.thm}):
%
%
\begin{theorem}[Bootstrap theorem]\label{thm:bootstrap}
For every $s ,\, N_0\in \mathbb N$ such that 
$s\geq 5$, there exists $n_{N_0,s} \in \mathbb N$ sufficiently large such that for every $n\geq n_{N_0,s}$, the following holds for some $T_{N_0,s,n} > 0$ sufficiently small.

Suppose $(g^{\mathrm{aux}}, k^{\mathrm{aux}}, h^{\mathrm{aux}})$ is the solution to \eqref{eq:hyp.sys} on a time interval $[T_{\mathrm{aux}}, T_{\mathrm{Boot}})$ (for some $T_{\mathrm{Boot}} \in (T_{\mathrm{aux}}, T_{N_0,s,n}]$), with initial data at $t = T_{\mathrm{aux}}$ given as in Lemma~\ref{lem:locexist}. Assume moreover that the bootstrap assumptions \eqref{eq:BA1}--\eqref{eq:BA4} all hold on $[T_{\mathrm{aux}}, T_{\mathrm{Boot}}) \times \mathbb T^3$.

Then in fact the following estimates hold:
%
%
\begin{align}\label{gdkdhdest}
\begin{split}
 \sum_{r=0}^{s} t^{2r} \|k^{(d)}\|^2_{H^r(\Sigma_t,g^{\mathrm{aux}})} + \sum_{r=0}^{s-1} t^{2r+2} \|\partial_tk^{(d)}\|^2_{H^r(\Sigma_t,g^{\mathrm{aux}})} + \sum_{r=0}^{s+1} t^{2r} \|h^{(d)}\|^2_{H^r(\Sigma_t,g^{\mathrm{aux}})}\\
  + \sum_{r=0}^{s+1} t^{2r-2} (\| g^{(d)} \|^2_{H^r(\Sigma_t,g^{\mathrm{aux}})} + \| (g^{-1})^{(d)} \|^2_{H^r(\Sigma_t,g^{\mathrm{aux}})}) \leq C t^{2N_0+2s}
\end{split}
\end{align}
on $[T_{\mathrm{aux}}, T_{\mathrm{Boot}}) \times \mathbb T^3$, where $C>0$ may depend on $s$, $N_0$ and the data, but is independent of $T_{\mathrm{aux}}$. 

Moreover, taking $T_{N_0,s,n}$ smaller if necessary, \eqref{gdkdhdest} improves over the bootstrap assumptions \eqref{eq:BA1}--\eqref{eq:BA4}. 
\end{theorem}
As is standard, the bootstrap theorem implies immediately, using a continuity argument, that the solution can be extended up to time $T_{N_0,s,n}$:
\begin{corollary}\label{cor:bootstrap}
Let $s$, $N_0$, $n$ and $T_{N_0,s,n}$ be as in Theorem~\ref{thm:bootstrap}. Then the local solution given in Lemma~\ref{lem:locexist} can in fact be extended to all of $[T_{\mathrm{aux}}, T_{N_0,s,n}) \times \mathbb T^3$. Moreover, the estimates \eqref{gdkdhdest} hold.
\end{corollary}

\subsubsection{Step~3: Conclusion of the argument}\label{sec:conclusion.of.actual.wave}

\begin{proposition}\label{prop:AA.main}
Let $s$, $N_0$, $n$ and $T_{N_0,s,n}$ be as in Theorem~\ref{thm:bootstrap}.

Them there exists a decreasing sequence of auxiliary times $\{ T_{\mathrm{aux},I} \}_{I=1}^{+\infty} \subset (0,T_{N_0,s,n})$, $\lim_{I\to +\infty} T_{\mathrm{aux},I} = 0$ such that the following holds:
\begin{enumerate}
\item The corresponding solutions $\{ (g^{\mathrm{aux}}_I, k^{\mathrm{aux}}_I, h^{\mathrm{aux}}_I \}_{I=1}^{+\infty}$ given by Lemma~\ref{lem:locexist} converge \emph{locally} in $C^3\times C^2\times C^2$ (as $I\to +\infty$) to a limit $(g, k ,h)$.
\item The limit, which we denote by $(g, k ,h)$, solves \eqref{eq:hyp.sys} in $(0,T_{N_0,s,n}]\times \mathbb T^3$.
\item Denoting $g^{(d)} = g - g^{\bf [n]}$, $(g^{-1})^{(d)} = g^{-1} - (g^{\bf [n]})^{-1}$, $k^{(d)} = k - k^{\bf [n]}$ and $h^{(d)} = h - h^{\bf [n]}$, the estimate \eqref{gdkdhdest} holds.
\item The limit $(g,k)$ satisfies $k_{ij} = -\f 12 \rd_t g_{ij}$.
\end{enumerate}

\end{proposition}

The proof of Proposition~\ref{prop:AA.main} will be given in Section~\ref{sec:pf.AA.main}.

\begin{proof}[Proof of Theorem~\ref{thm:main.reduced}]
The limiting solution given by Proposition~\ref{prop:AA.main} satisfies all the conclusions of Theorem~\ref{thm:main.reduced}. This thus concludes the proof of Theorem~\ref{thm:main.reduced}. \qedhere
\end{proof}

\subsection{Definition of the energies and an outline of the proof of Theorem~\ref{thm:bootstrap}}\label{sec:bootstrap.outline}
 {From now on until the end of Section~\ref{sec:conclusion.of.bootstrap.thm}, we focus on the proof of Theorem~\ref{thm:bootstrap}. \textbf{To lighten our notations, in these sections we write $g = g^{\mathrm{aux}}$, $a = a^{\mathrm{aux}}$, $h=h^{\mathrm{aux}}$ and $k = k^{\mathrm{aux}}$.}}

The crux of our proof of Theorem~\ref{thm:bootstrap} is to bound an appropriate energy, which we define now.

Define the energy
\begin{equation}\label{eq:energy.def}
\begin{split}
\mathcal E_s(t):= &\: \sum_{r=0}^{s-1} t^{2r+2} \|\rd_t k^{(d)}\|^2_{\dot H^r(\Sigma_t,g)} +  \sum_{r=0}^{s} t^{2r} \|k^{(d)}\|^2_{\dot H^r(\Sigma_t,g)} \\
&\: + \sum_{r=0}^{s+1} t^{2r} \|h^{(d)}\|^2_{\dot H^r(\Sigma_t,g)} + \sum_{r=0}^{s+1} t^{2r-2} (\| g^{(d)} \|^2_{\dot H^r(\Sigma_t,g)} + \| (g^{-1})^{(d)} \|^2_{\dot H^r(\Sigma_t,g)}).
\end{split}
\end{equation}

Define also the modified energy
\begin{equation}\label{eq:menergy.def}
\begin{split}
\widetilde{\mathcal E_s}(t):=  &\: \sum_{r=0}^{s-1} t^{2r+2} \|\rd_t \nab^{(r)} k^{(d)}\|^2_{L^2(\Sigma_t,g)} +  \sum_{r=0}^{s} t^{2r} \|k^{(d)}\|^2_{\dot H^r(\Sigma_t,g)} \\
&\: + \sum_{r=0}^{s} t^{2r} \|h^{(d)}\|^2_{\dot H^r(\Sigma_t,g)} + \sum_{r=0}^{s} t^{2r-2} (\| g^{(d)} \|^2_{\dot H^r(\Sigma_t,g)} + \| (g^{-1})^{(d)} \|^2_{\dot H^r(\Sigma_t,g)}) \\
&\: + t^{2(s+1)} \| \widetilde{ \nab^{(s+1)}_{ren} h^{(d)}} \|_{L^2(\Sigma_t,g)}^2 +t^{2s} \| \widetilde{ \nab^{(s+1)}_{ren} g^{(d)}} \|_{L^2(\Sigma_t,g)}^2 + t^{2s}\| \widetilde{ \nab^{(s+1)}_{ren} (g^{-1})^{(d)}} \|_{L^2(\Sigma_t,g)}^2,
\end{split}
\end{equation}
where $\widetilde{ \nab_{ren}^{(s+1)} h^{(d)}}$, $\widetilde{ \nab^{(s+1)}_{ren} g^{(d)}}$ and $\widetilde{ \nab^{(s+1)}_{ren} (g^{-1})^{(d)}}$ are the \emph{renormalized top-order quantities} defined by
\begin{equation}\label{def:hd.ren}
(\widetilde{ \nab^{(s+1)}_{ren} h^{(d)}})_{i_1\cdots i_{s-1}} := \Delta_g \nab^{(s-1)}_{i_1\cdots i_{s-1}} h^{(d)} -2(\kn + \kd)_i{ }^j  \rd_t \nab^{(s-1)}_{i_1\cdots i_{s-1}} (k^{(d)})_j{ }^i,
\end{equation}
\begin{equation}\label{def:gd.ren}
\begin{split}
&\: (\widetilde{ \nab^{(s+1)}_{ren} g^{(d)}})_{i_1\cdots i_{s-2} a i j} \\
:= &\:\Delta_g\nabla^{(s-1)}_{i_1\cdots i_{s-2} a} g_{ij}^{(d)}+2g_{\ell(j}\partial_t \nabla^{(s-1)}_{i_1\cdots i_{s-2} a} (k^{(d)})_{i)}{}^\ell \\
&\: + g^{(d)}_{bj}\partial_t\nabla^{(s-2)}_{i_1\cdots i_{s-2}}((g^{-1})^{be} g_{m(i|}\nabla_e(k^{(d)})_{a)}{}^m - \nabla_{(a} k_{i)}{ }^b - (g^{-1})^{be} g_{ {d}(a} \nab_{i)} k_e{}^d)\\
&\: + g^{(d)}_{ib}\partial_t\nabla^{(s-2)}_{i_1\cdots i_{s-2}}((g^{-1})^{be} g_{m(j|}\nabla_e(k^{(d)})_{a)}{}^m - \nabla_{(a} k_{j)}{ }^b - (g^{-1})^{be} g_{ {d}(a} \nab_{j)} k_e{}^d),
\end{split}
\end{equation}
and
\begin{equation}\label{def:g-1d.ren}
\begin{split}
&\: (\widetilde{ \nab^{(s+1)}_{ren} (g^{-1})^{(d)}})^{ij}_{i_1\cdots i_{s-2} a} \\
:= &\:\Delta_g\nabla^{(s-1)}_{i_1\cdots i_{s-2} a} ((g^{-1})^{(d)})^{ij} - 2 (g^{-1})^{\ell(j} \partial_t \nabla^{(s-1)}_{i_1\cdots i_{s-2} a} (k^{(d)})_{\ell}{}^{i)} \\
&\:+ ((g^{-1})^{(d)})^{bj} \partial_t\nabla^{(s-2)}_{i_1\cdots i_{s-2}} ((g^{-1})^{ie} g_{m(b|}\nab_e k_{a)}{ }^m - \nabla_{(a} k_{b)}{ }^i - (g^{-1})^{ie} g_{ {d}(a} \nab_{b)} k_e{}^d) \\
&\: + ((g^{-1})^{(d)})^{ib} \partial_t\nabla^{(s-2)}_{i_1\cdots i_{s-2}} ((g^{-1})^{je} g_{m(b|}\nab_e k_{a)}{ }^m - \nabla_{(a} k_{b)}{ }^j - (g^{-1})^{je} g_{ {d}(a} \nab_{b)} k_e{}^d).
\end{split}
\end{equation}

We remark explicitly that the modified energy and the energy differ by the following:
\begin{itemize}
\item The energy controls the $\nab^{(r)}$ derivative of $\rd_t k^{(d)}$ while the modified energy controls the $\rd_t$ derivative of $\nab^{(r)} k^{(d)}$.
\item The modified energy only controls $h^{(d)}$, $g^{(d)}$ and $(g^{-1})^{(d)}$ up to $s$ derivatives; at the top order it only controls the renormalized top-order quantities.
\end{itemize}

Since the proof will take several subsections, we give an outline of the strategy for proving Theorem~\ref{thm:bootstrap}.
\begin{itemize}
\item In \textbf{Section~\ref{sec:bootstrap.prelim}}, we begin with some preliminary estimates. 
\item In \textbf{Section~\ref{subsec:mainest.k}}, we carry out the energy estimate for $\kd$ using the wave equation it satisfies.
\item In \textbf{Section~\ref{subsec:mainest.other}}, we carry out the energy estimates for $h^{(d)}$, $g^{(d)}$ and $(g^{-1})^{(d)}$ using the transport equations they satisfy. Combining the results in Sections~\ref{subsec:mainest.k} and \ref{subsec:mainest.other}, we will obtain an estimate of the modified energy $\widetilde{\mathcal E_s}$ by the energy $\mathcal E_s$.
\item In \textbf{Section~\ref{sec:conclusion.of.bootstrap.thm}}, we complete the proof of Theorem~\ref{thm:bootstrap}. The main ingredient is to control $\mathcal E_s$ and $\widetilde{\mathcal E_s}$ using energy estimates, and the close everything using the Fuchsian ideas as illustrated in Section~\ref{sec:ideas.Fuchsian}.
\end{itemize}

\subsubsection{Remarks on the dependence of constants (and related conventions)}\label{sec:rmks.constants}

Before we proceed, we make some important remarks regarding the dependence of constants throughout the proof of Theorem~\ref{thm:bootstrap}.

\textbf{From now on fix $s\in \mathbb N$ with $s \geq 5$} as in Theorem~\ref{thm:bootstrap}.

We will use $C_0$ and $C_n$ as \underline{general positive constants}. They may change from line to line. \textbf{Both $C_0$ and $C_n$ may depend on the data $c_{ij}$, $p_i$ and also $s$, but importantly \emph{$C_n$ may depend on $n$ while $C_0$ is \underline{not} allowed to depend on $n$}.}

\textbf{We always assume without loss of generality that $T_{N_0,s,n} \leq 1$.}

\subsubsection{Remarks regarding $k$}

Another important remark regarding the proof of the bootstrap argument is that (despite the notation) we do \underline{not} know that $k$ is the second fundamental form of the constant-$t$ hypersurfaces. (In particular, we do not know that $g_{\ell [i} k_{j]}{ }^\ell = 0$.) In fact, it is only after extracting a limit in Proposition~\ref{prop:AA.main} that we know that the \underline{limiting} $k$ is an honest second fundamental form. 

\subsection{Preliminary estimates for the bootstrap argument}\label{sec:bootstrap.prelim}

In this subsection we work under the assumptions of Theorem~\ref{thm:bootstrap}. In particular, we assume the validity of the bootstrap assumptions \eqref{eq:BA1}--\eqref{eq:BA4}.

\subsubsection{Sobolev embedding and basic comparisons of norms}

\begin{lemma}\label{lem:stupid.comparison}
The following pointwise estimate holds for all scalar functions $f$ on $(0,T_{\mathrm{Boot}})$:
$$C_0^{-1} t|\nab f|_g \leq \sum_{i = 1}^3 |\rd_i f| \leq C_0 t^{-1} |\nab f|_g.$$
\end{lemma}
\begin{proof}
By definition, $|\nab f|_g^2 = (g^{-1})^{ij} \rd_i f \rd_j f $. To get the desired estimates, we just use a very wasteful estimate that $C_0^{-1} t^2\leq \min_{i,j} | (g^{-1})^{ij}|\leq \max_{i,j} | (g^{-1})^{ij}|\leq C_0 t^{-2}$ (which follows directly from \eqref{eq:BA1} and computations as in \eqref{invg0}). \qedhere
\end{proof}

\begin{lemma}[Sobolev embedding]\label{lem:Sobolev}
The following holds for every  $(m,l)$ $\Sigma$-tangent tensor $\mathcal T$:
\begin{equation}\label{eq:Sobolev.1}
\|\mathcal T\|_{L^\i(\Sigma_t,g)} \leq C_0 t^{-\f 54} \|\mathcal{T}\|_{W^{1,4}(\Sigma_t,g)},\quad \|\mathcal T\|_{L^4(\Sigma_t,g)} \leq C_0 t^{-\f 54} \|\mathcal{T}\|_{W^{1,2}(\Sigma_t,g)}.
\end{equation}

In particular, these inequalities imply
\begin{equation}\label{eq:Sobolev.2}
\|\mathcal T\|_{L^\i(\Sigma_t,g)} \leq C_0 t^{-\f 52} \|\mathcal{T}\|_{H^2(\Sigma_t,g)},
\end{equation}
and \begin{equation}\label{eq:Sobolev.used}
\begin{split}
\|\nab^{(r)} \mathcal T\|_{(L^2 \cap t^{-s-\f 52+\ve} L^\i)(\Sigma_t,g)} 
\leq &\: C_0 \sum_{r' = r}^{r+ 2} t^{r'-r} \|\mathcal T\|_{\dot{H}^{r'}(\Sigma_t,g)},
\end{split}
\end{equation}
\end{lemma}
\begin{proof}
We first prove the inequalities \eqref{eq:Sobolev.1} for scalar functions $f$. Using the Sobolev embedding for $\mathbb T^3$ in \emph{coordinates}, it follows that
\begin{equation}\label{eq:Sobolev.scalar.1}
\|f\|_{L^\i(\Sigma_t,g)} \leq C_0 \sum_{|\alp|\leq 1} (\int_{\Sigma_t} |\rd_x^\alp f|^4 \, \ud x)^{\f 14} \leq C_0  (\int_{\Sigma_t} (|f|^4  + t^{-4}|\nab f|_g^4) \, \ud x)^{\f 14} \leq C_0 t^{-\f 54} \| f \|_{W^{1,4}(\Sigma_t,g)},
\end{equation}
where in the penultimate inequality we used Lemma~\ref{lem:stupid.comparison} and in the last inequality we have used $C_0^{-1} t^{-1} \mathrm{vol}_{\Sigma} \leq  \ud x \leq C_0 t^{-1} \mathrm{vol}_{\Sigma} $ (which follows from the bootstrap assumption \eqref{eq:BA1}).

For the second inequality in \eqref{eq:Sobolev.1} for a scalar function $f$, we proceed similarly to obtain
\begin{equation}\label{eq:Sobolev.scalar.2}
\|f\|_{L^4(\Sigma_t,g)} \leq C_0 t^{\f 14}\sum_{|\alp|\leq 1} (\int_{\Sigma_t} |\rd_x^\alp f|^2 \, \ud x)^{\f 12} \leq C_0  t^{\f 14} (\int_{\Sigma_t} (|f|^2+ t^{-2}|\nab f|^2) \, \ud x)^{\f 12} \leq C_0 t^{-\f 54} \| f \|_{H^1(\Sigma_t)}.
\end{equation}

Now given a general $(m,l)$ tensor $\mathcal T$, using \eqref{eq:Sobolev.scalar.1} and \eqref{eq:Sobolev.scalar.2} with $f_\alp = \sqrt{|\mathcal T|^2_g + \alp^2}$ ($\alp >0$) and taking $\alp\to 0$, we obtain the desired inequalities in \eqref{eq:Sobolev.1}. 

Next, it is easy to see that \eqref{eq:Sobolev.1} implies \eqref{eq:Sobolev.2}. 

Finally, by \eqref{eq:Sobolev.2}, and the fact $s-\ve >2$,
\begin{equation*}
\begin{split}
\|\nab^{(r)} \mathcal T\|_{(L^2 \cap t^{-s-\f 52+\ve} L^\i)(\Sigma_t,g)} \leq &\: C_0 (\| \nab^{(r)}\mathcal T\|_{L^2(\Sigma_t,g)} + t^{s+\f 52-\ve} t^{-\f 52} \|\nab^{(r)} \mathcal T\|_{H^2(\Sigma_t,g)} )\\
\leq &\: C_0 \sum_{r' = r}^{r+ 2} t^{r'-r} \|\mathcal T\|_{\dot{H}^{r'}(\Sigma_t,g)},
\end{split}
\end{equation*}
which is \eqref{eq:Sobolev.used}.
\qedhere
\end{proof}

We will also need to compare norms with respect to $g$ and with respect to the trivial metric $\sum_{i=1}^3 (dx^i)^2$.

\begin{lemma}\label{lem:coord.v.geo.norm}
Given any rank $(l, m)$ $\Sigma$-tangent tensor $\mathcal T$,
\begin{equation}\label{eq:coord.v.geo.norm.1}
\begin{split}
&\: |  \underbrace{\nab^{\bf [n]} \cdots \nab^{\bf [n]}}_{\mbox{$k$ times}} \mathcal T  |_{g^{\bf [n]}}^2 \\
\leq &\: C_n t^{-k(2-\varepsilon)} \sum_{r=0}^k \sum_{i_1,\ldots,i_r} \sum_{b_1,\ldots,b_l}\sum_{j_1,\ldots,j_m}  t^{- 2 p_{b_1}} \cdots t^{-2 p_{b_l}}\cdot  t^{2 p_{j_1}} \cdots t^{2 p_{j_m}} |\rd_{i_1}\cdots \rd_{i_r}\mathcal{T}_{b_1,\ldots b_l}^{j_1\ldots j_m}|^2.
\end{split}
\end{equation}
\end{lemma}
\begin{proof}
Using the form of the metric \eqref{metricansatz}, the bound \eqref{eq:main.parametrix.a.bd} on $a^{[\bf n]}$, and the fact $p_1 < p_2 < p_3$,
\begin{align}\label{eq:T.in.geo.norm}
|\mathcal T  |_{g^{\bf [n]}}^2 \notag
= &\: ((g^{[\bf n]})^{-1})^{i_1b_1}\ldots ((g^{[\bf n]})^{-1})^{i_lb_l}g^{[\bf n]}_{j_1c_1}\ldots g^{[\bf n]}_{j_mc_m}\mathcal{T}^{j_i\ldots j_m}_{b_1\ldots b_l}\mathcal{T}^{c_i\ldots c_m}_{i_1\ldots i_l} \\
\notag\leq &\: C_0  \sum_{b_1,\ldots,b_l}\sum_{j_1,\ldots,j_m} \sum_{i_1,\ldots,i_l}  \sum_{c_1,\ldots,c_m} t^{-2p_{\min\{i_1, b_1 \}}}\cdots  t^{-2p_{\min\{i_l, b_l \}}} \cdot \\
\notag&\cdot t^{2p_{\max\{j_1, c_1\}}} \cdots t^{2p_{\max\{j_m, c_m\}}} \mathcal{T}^{j_i\ldots j_m}_{b_1\ldots b_l}\mathcal{T}^{c_i\ldots c_m}_{i_1\ldots i_l} \\
\leq &\:  C_0 (\sum_{b_1,\ldots,b_l}\sum_{j_1,\ldots,j_m}  t^{-p_{\min\{i_1, b_1 \}}}\cdots  t^{-p_{\min\{i_l, b_l \}}} \cdot t^{p_{\max\{j_1, c_1\}}} \cdots t^{p_{\max\{j_m, c_m\}}}  |\mathcal{T}^{j_i\ldots j_m}_{b_1\ldots b_l}|) \\
\notag&\: \times (\sum_{i_1,\ldots,i_l}  \sum_{c_1,\ldots,c_m} t^{-p_{\min\{i_1, b_1 \}}}\cdots  t^{-p_{\min\{i_l, b_l \}}} \cdot t^{p_{\max\{j_1, c_1\}}} \cdots t^{p_{\max\{j_m, c_m\}}}  |\mathcal{T}^{c_i\ldots c_m}_{i_1\ldots i_l}|) \\
\notag\leq &\: C_0 \sum_{b_1,\ldots,b_l}\sum_{j_1,\ldots,j_m}  t^{-2p_{b_1}} \cdots  t^{- 2p_{b_l} } \cdot t^{2p_{j_1} } \cdots t^{2p_{j_m}}  |\mathcal{T}^{j_i\ldots j_m}_{b_1\ldots b_l}|^2.
\end{align}
This proves \eqref{eq:coord.v.geo.norm.1} when there are no derivatives (i.e.~$k=0$).

Define the flat connection $\nab^{(flat)}$ to be Levi--Civita connection associated to $\sum_{i=1}^3 (dx^i)^2$, i.e.
$$\nab^{(flat)}_{i_1} \cdots \nab^{(flat)}_{i_r} \mathcal{T}^{j_i\ldots j_m}_{b_1\ldots b_l} = \rd_{i_1} \cdots \rd_{i_r} \mathcal{T}^{j_i\ldots j_m}_{b_1\ldots b_l}.$$ 
Then, since $p_j <1 -\ve < 1-\f \ve 2$, \eqref{eq:T.in.geo.norm} gives
\begin{equation}\label{eq:T.in.geo.norm.w.derivatives}
\begin{split}
&\: |\nab^{(flat)}_{i_1} \cdots \nab^{(flat)}_{i_r} \mathcal T |_{g^{\bf [n]}}^2 \\
\leq &\: C_0 \sum_{i_1,\dots,i_r} \sum_{b_1,\ldots,b_l}\sum_{j_1,\ldots,j_m} t^{-2p_{i_1}} \cdots t^{-2p_{i_r}} \cdot t^{-2p_{b_1}} \cdots  t^{- 2p_{b_l} } \cdot t^{2p_{j_1} } \cdots t^{2p_{j_m}}  |\nab^{(flat)}_{i_1} \cdots \nab^{(flat)}_{i_r} \mathcal{T}^{j_i\ldots j_m}_{b_1\ldots b_l}|^2 \\
\leq &\: C_0 t^{-(2-\ve)r} \sum_{i_1,\dots,i_r} \sum_{b_1,\ldots,b_l}\sum_{j_1,\ldots,j_m} t^{-2p_{b_1}} \cdots  t^{- 2p_{b_l} } \cdot t^{2p_{j_1} } \cdots t^{2p_{j_m}}  |\rd_{i_1} \cdots \rd_{i_r} \mathcal{T}^{j_i\ldots j_m}_{b_1\ldots b_l}|^2.
\end{split}
\end{equation}

Now compute
\begin{equation}\label{eq:coord.v.geo.norm.2}
\begin{split}
&\: \nab^{\bf [n]}_{i_1} \cdots \nab^{\bf [n]}_{i_k} \mathcal T^{j_i\ldots j_m}_{b_1\ldots b_l} \\
= &\: \rd_{i_1}\cdots \rd_{i_k} \mathcal T^{j_i\ldots j_m}_{b_1\ldots b_l}  + \sum_{s=0}^{k-1} \sum_{e=1}^l \rd_{i_1} \cdots \rd_{i_s} [(\Gamma^{\bf [n]})_{i_{s+1} b_e}^{f} \rd_{i_{s+2}} \cdots \rd_{i_k}   \mathcal T^{j_i\ldots j_m}_{b_1\ldots b_{e-1} f b_{e+1}\ldots b_l} ] \\
&\: + \sum_{s=0}^{k-1} \sum_{e=s+2}^r \rd_{i_1} \cdots \rd_{i_s} [(\Gamma^{\bf [n]})_{i_{s+1} i_e}^{f} \rd_{i_{s+2}}\cdots \rd_{i_{e-1}} \rd_{i_f} \rd_{i_{e+1}} \cdots \rd_{i_r}   \mathcal T^{j_i\ldots j_m}_{b_1\ldots b_l} ] \\
&\: + \sum_{s=0}^{k-1} \sum_{e=1}^m \rd_{i_1} \cdots \rd_{i_s} [(\Gamma^{\bf [n]})_{i_{s+1} f}^{j_e} \rd_{i_{s+2}} \cdots \rd_{i_r}   \mathcal T^{j_i\ldots j_{e-1} f j_{e+1} j_m}_{b_1\ldots b_l} ]  + \cdots +  \underbrace{(\Gamma^{\bf [n]})\cdots (\Gamma^{\bf [n]})}_{\mbox{$k-1$ factors}} \rd \mathcal T\\
&\: +  \underbrace{(\Gamma^{\bf [n]})\cdots (\Gamma^{\bf [n]})}_{\mbox{$k-2$ factors}} (\rd \Gamma^{\bf [n]}) \mathcal T + \underbrace{(\Gamma^{\bf [n]})\cdots (\Gamma^{\bf [n]})}_{\mbox{$k$ factors}} \mathcal T,
\end{split}
\end{equation}
where we have suppressed the indices in terms where the exact contractions do not matter.

Our goal is to show that in the $|\cdot |_{g^{\bf [n]}}$ norm, each term in \eqref{eq:coord.v.geo.norm.2} can be bounded above by the RHS of \eqref{eq:coord.v.geo.norm.1}. By repeated application of the Cauchy--Schwarz inequality (with respect to $g^{\bf [n]}$), and using \eqref{eq:T.in.geo.norm.w.derivatives}, it suffices to prove 
\begin{equation}\label{eq:coord.v.geo.norm.goal}
|\partial_{i_1}\cdots \partial_{i_r} (\Gamma^{\bf [n]})_{jb}^\ell |_{g^{\bf [n]}} \leq C_n t^{-(1-\f{\ve}{2})(r+1)},
\end{equation}
which is the goal for the remainder of the proof.

We first make the easy observation that $|\rd_{i_1} \cdots \rd_{i_d} g^{\bf [n]}_{ab}| \leq C_n |\log t|^d t^{2\max\{p_a,\,p_b\}}$ and $|\rd_{i_1} \cdots \rd_{i_d} ((g^{\bf [n]})^{-1})^{bc}| \leq C_n |\log t|^d t^{-2\min\{p_b,\,p_c\}}$. In particular, 
\begin{equation*}
\begin{split}
|\rd_{i_1} \cdots \rd_{i_{d_1}} g^{\bf [n]}_{ab} \rd_{j_1} \cdots \rd_{j_{d_2}} ((g^{\bf [n]})^{-1})^{bc}| \leq  &\: C_n |\log t|^{d_1} t^{2\max\{p_a,\,p_b\}}|\log t|^{d_2} t^{-2\min\{p_b,\,p_c\}} \\
\leq &\: C_n |\log t|^{d_1+d_2} t^{2p_b} t^{-2p_b} = C_n |\log t|^{d_1+d_2}.
\end{split}
\end{equation*}

We now compute
\begin{equation*}
\begin{split}
&\: |\partial_{i_1}\cdots \partial_{i_r} (\Gamma^{\bf [n]})_{jb}^\ell |_{g^{\bf [n]}}^2 \\
= &\: \f 14 ((g^{\bf [n]})^{-1})^{i_1 i_1'}\cdots ((g^{\bf [n]})^{-1})^{i_r i_r'} ((g^{\bf [n]})^{-1})^{j j'} ((g^{\bf [n]})^{-1})^{b b'} g^{\bf [n]}_{\ell\ell'}\\
&\: \qquad \times \partial_{i_1'}\cdots \partial_{i_r'} [(g^{\bf [n]})^{-1})^{\ell' a'}(\rd_{j'} g^{\bf [n]}_{a'b'}+ \rd_{b'} g^{\bf [n]}_{a'j'} - \rd_{a'} g^{\bf [n]}_{b'j'})] \partial_{i_1}\cdots \partial_{i_r} [(g^{\bf [n]})^{-1})^{\ell a}(\rd_j g^{\bf [n]}_{ab}+ \rd_b g^{\bf [n]}_{aj} - \rd_a g^{\bf [n]}_{bj})].
\end{split}
\end{equation*}
Consider the example expression
$$ |((g^{\bf [n]})^{-1})^{b b'} g^{\bf [n]}_{\ell\ell'} [\rd \cdots \rd ((g^{\bf [n]})^{-1})^{\ell' a'}] [\rd \cdots \rd g^{\bf [n]}_{a'b'}] [\rd \cdots \rd (g^{\bf [n]})^{-1})^{\ell a}] [\rd \cdots \rd g^{\bf [n]}_{ab}]| .$$
We can pair up $g^{\bf [n]}$ and $(g^{\bf [n]})^{-1}$ with a common index and conclude that this expression is $\leq C_n |\log t|^{k+1}$.
All other terms are similar. Hence, we obtain
\begin{equation*}
\begin{split}
&\: |\partial_{i_1}\cdots \partial_{i_r} (\Gamma^{\bf [n]})_{jb}^\ell |_{g^{\bf [n]}}^2 \\
\leq  &\: C_n  |\log t|^{r+1} \max_{i_1,\,i_1',\dots, j,\,j'} |((g^{\bf [n]})^{-1})^{i_1 i_1'}\cdots ((g^{\bf [n]})^{-1})^{i_r i_r'} ((g^{\bf [n]})^{-1})^{j j'} | \\
\leq &\: C_n  |\log t|^{r+1} t^{-(r+1)(2-2\ve)} \leq C_n t^{-(2-\ve)(r+1)},
\end{split}
\end{equation*}
which is exactly \eqref{eq:coord.v.geo.norm.goal}.
\end{proof}

\begin{lemma}\label{lem:est.nabd}
For $r\leq s-2$,
$$\|\nab^{(d)}\|_{W^{r,\i}(\Sigma_t,g)} \leq C_0 ( \|g^{(d)}\|_{W^{r+1,\i}(\Sigma_t,g)} + \|(g^{-1})^{(d)}\|_{W^{r,\i}(\Sigma_t,g)}).$$
For $r\leq s$,
$$\|\nab^{(d)}\|_{H^r(\Sigma_t,g)} \leq C_0 ( \|g^{(d)}\|_{H^{r+1}(\Sigma_t,g)} + \|(g^{-1})^{(d)}\|_{H^{r}(\Sigma_t,g)}).$$
\end{lemma}
\begin{proof}
Note that
\begin{equation}\label{eq:diff.Gamma.id}
(\nab^{(d)})^\ell_{ij} = \f 12 [(g^{-1})^{\ell b} - ((g^{-1})^{(d)})^{\ell b}] (\nab_i g^{(d)}_{bj} + \nab_j g^{(d)}_{bi} - \nab_b g^{(d)}_{ij} ).
\end{equation}
The conclusion is then an immediate consequence of H\"older's inequality and the bootstrap assumptions \eqref{eq:BA2} and \eqref{eq:BA3}.
\end{proof}

\begin{lemma}\label{lem:Wki.comp}
For $r\leq s-1$,
$$C_0^{-1} \|\mathcal T\|_{W^{r,\i}(\Sigma_t,g^{\bf [n]})} \leq \|\mathcal T\|_{W^{r,\i}(\Sigma_t,g)} \leq C_0 \|\mathcal T\|_{W^{r,\i}(\Sigma_t,g^{\bf [n]})}.$$
\end{lemma}
\begin{proof}
This follows from the bootstrap assumption \eqref{eq:BA2} and Lemma~\ref{lem:est.nabd}.
\end{proof}

\begin{lemma}\label{lem:Hk.comp}
For $r\leq s+1$, 
$$C_0^{-1} \|\mathcal T\|_{H^r(\Sigma_t,g^{\bf [n]})} \leq  \|\mathcal T\|_{H^r(\Sigma_t,g)} \leq C_0 \|\mathcal T\|_{H^r(\Sigma_t,g^{\bf [n]})}.$$
\end{lemma}
\begin{proof}
This follows from the bootstrap assumptions \eqref{eq:BA2}, \eqref{eq:BA3} and Lemma~\ref{lem:est.nabd}.
\end{proof}

Note that Lemma~\ref{lem:Wki.comp} fails when $r = s,\,s+1$ as we do not control $\|g^{(d)}\|_{W^{r,\i}(\Sigma_t,g)}$. On the other hand, Lemma~\ref{lem:Hk.comp} by itself will not be sufficient for our purpose. Instead we need the following
\begin{lemma}\label{lem:Wsi.save}
The following holds for any $\alp >0$:
$$\|\nab^{(s)} \mathcal T\|_{L^{\infty}(\Sigma_t,g) + t^{\alp} L^2(\Sigma_t,g)} \leq C_0( \|\mathcal T\|_{W^{s,\infty}(\Sigma_t,g^{\bf [n]})} + t^{-\alp + \f 52} \|\mathcal T\|_{L^\i(\Sigma_t,g^{\bf [n]})}),$$
and 
$$\|\nab^{(s+1)} \mathcal T\|_{L^{\infty}(\Sigma_t,g) + t^{\alp} L^2(\Sigma_t,g)} \leq C_0 (\|\mathcal T\|_{W^{s+1,\infty}(\Sigma_t,g^{\bf [n]})} + t^{-\alp + \f 52} \|\mathcal T\|_{W^{1,\i}(\Sigma_t,g^{\bf [n]})}).$$

\end{lemma}
\begin{proof}
The main difference with Lemma~\ref{lem:Hk.comp} is that we may have terms which involve $s$ or $s+1$ derivatives of $g^{(d)}$. 

We first consider the term $\nab^{(s)} \mathcal T$. When writing $\nab^{(s)} \mathcal T$ in terms of $(\nab^{\bf [n]})^{(s)} \mathcal T$, there is the term
$$ \mathcal T [\nab^{(s-1)}\nab^{(d)}] $$
(meaning $s-1$ $\nab$ derivatives acting on the tensor $\nab^{(d)}$), together with other terms which are lower order and can be handled directly using the bootstrap assumptions \eqref{eq:BA2} and \eqref{eq:BA3}.
This term cannot be bounded in $L^\i$, and will instead be controlled in $L^2$. For this we note that by H\"older's inequality, Lemma~\ref{lem:est.nabd}, and the bootstrap assumptions \eqref{eq:BA2}, \eqref{eq:BA3},
\begin{equation*}
\begin{split}
\|\mathcal T [\nab^{(s-1)}\nab^{(d)}] \|_{L^2(\Sigma_t,g)} \leq &\: C_0 \|\mathcal T\|_{L^\i(\Sigma_t,g)} (\|g^{(d)}\|_{H^s(\Sigma_t,g)} + \| (g^{-1})^{(d)} \|_{H^{s-1}(\Sigma_t,g)}) \\
\leq &\: C_0 t^{\f 52}\|\mathcal T\|_{L^\i(\Sigma,g)} \leq C_0 t^{\f 52}\|\mathcal T\|_{L^\i(\Sigma,g^{\bf [n]})}. 
\end{split}
\end{equation*}
This gives the first inequality in the statement of the lemma.

The term $\nab^{(s+1)} \mathcal T$ is similar except for an additional derivative. Indeed, we need to control the terms
$$ [\nab \mathcal T][\nab^{(s-1)}\nab^{(d)}],\quad \mathcal T [\nab^{(s)}\nab^{(d)}].$$
Both of these can be controlled in $L^2(\Sigma_t,g)$ using H\"older's inequality, Lemma~\ref{lem:est.nabd}, and the bootstrap assumptions \eqref{eq:BA2}, \eqref{eq:BA3} as above.
\end{proof}

\subsubsection{An easy consequence of the bootstrap assumption}

%

\begin{lemma}\label{lem:kd.Li}
$$\| h^{(d)}\|_{W^{s-1,\infty}(\Sigma_t, g)}  + \| k^{(d)}\|_{W^{s-2,\infty}(\Sigma_t, g)} + \| \rd_t k^{(d)}\|_{W^{s-3,\infty}(\Sigma_t, g)}  \leq C_0. $$
\end{lemma}
\begin{proof}
This follows from Lemma~\ref{lem:Sobolev} (Sobolev embedding) and the bootstrap assumption \eqref{eq:BA4}.
\end{proof}

\subsubsection{Estimates for background quantities}

\begin{proposition}\label{prop:inhomo}
For each $n \in \mathbb N$, define 
$$I_{h^{[{\bf n}]}} := -\rd_t h^{\bf [n]} + |k^{\bf [n]}|^2,$$
$$(I_{k^{[{\bf n}]}})_i{ }^j:= -\rd_t^2 (k^{\bf [n]})_i{ }^j + \Delta_{g^{\bf[n]}} (k^{\bf [n]})_i{ }^j - (\nab_i \nab^j h)^{\bf [n]} + (k^{\bf [n]}\star k^{\bf [n]}\star k^{\bf [n]})_i{ }^j + (\rd_t k^{\bf [n]} \star k^{\bf [n]})_i{ }^j. $$

Given any $N \in \mathbb N$, there exists $n_{N,s} \in \mathbb N$ sufficiently large such that whenever $n \geq n_{N,s}$,
$$\sum_{r=0}^{s+1} t^r\|I_{h^{[{\bf n}]}}\|_{H^{r}(\Sigma_t,g)} + \sum_{r=0}^{s-1} t^r\|I_{k^{[{\bf n}]}}\|_{H^{r}(\Sigma_t,g)} \leq C_n t^{N + s}.$$

\end{proposition}
\begin{proof}
 {By Lemma~\ref{lem:coord.v.geo.norm}, it suffices to show that for any given polynomial rate, $n$ can be chosen sufficiently large so that $I_{h^{[{\bf n}]}}$, $(I_{k^{[{\bf n}]}})_i{ }^j$ and their coordinate derivatives tend to $0$ faster than the given polynomial rate.}

\pfstep{Step~1: Proving the estimates for $I_{h^{[{\bf n}]}}$} Recall that by definition $h^{\bf [n]} = (\kn)_\ell{ }^\ell$. By Proposition~\ref{prop:approx.G}, \eqref{eq:main.parametrix.k.and.II.1} and the expression for $Ric_{tt}(^{(4)}g)$ in \eqref{Rictt}, it follows that given any polynomial rate in $t$, we can choose $n\in \mathbb N$ sufficiently large so that $I_{h^{[{\bf n}]}} := -\rd_t h^{\bf [n]} + |k^{\bf [n]}|^2$ and its coordinate derivatives go to $0$ faster than the given polynomial rate in $t$.

\pfstep{Step~2: Proving the estimates for $I_{k^{[{\bf n}]}}$} By \eqref{eq:k.wave.prelim} and \eqref{eq:main.parametrix.k.and.II.1},
\begin{equation}\label{eq:Ik.in.terms.of.Ricci}
\begin{split}
(I_{k^{[{\bf n}]}})_i{ }^j = & -\partial_t Ric_i{}^j({^{(4)} g^{\bf [n]}}) + \nabla_i (\mathcal{G}^{\bf [n]})^j+\nabla^j (\mathcal{G}^{\bf [n]})_i
-3 (k^{[\bf n]})_i{}^mRic_m{}^j({^{(4)}g^{\bf [n]}}) \\
 &+ 2\de_i^j (k^{[\bf n]})_m{}^\ell Ric_{\ell}{}^m({^{(4)}g^{\bf [n]}})
 - (k^{[\bf n]})_\ell{}^j Ric_i{}^\ell ({^{(4)}g^{\bf [n]}})  + 2 (k^{[\bf n]})_\ell{}^\ell Ric_i{}^j({^{(4)}g^{\bf [n]}})\\
 &- ( (k^{[\bf n]})_\ell{}^\ell \de_i^j - (k^{[\bf n]})_i{}^j) Ric_m{}^m({^{(4)}g^{\bf [n]}}) + O(t^{L_n}),
\end{split}
\end{equation}
where $(\mathcal{G}^{\bf [n]})_i = Ric({^{(4)} g^{\bf [n]}})_{ti}$ and $L_n$ is linearly increasing in $n$.

By Proposition~\ref{prop:approx.G}, given any polynomial rate in $t$, we can choose $n\in \mathbb N$ sufficiently large so that the terms 
\begin{equation*}
\begin{split}
&-\partial_t Ric_i{}^j({^{(4)} g^{\bf [n]}}) + \nabla_i (\mathcal{G}^{\bf [n]})^j+\nabla^j (\mathcal{G}^{\bf [n]})_i
-3 (k^{[\bf n]})_i{}^mRic_m{}^j({^{(4)}g^{\bf [n]}}) \\
 &+ 2\de_i^j (k^{[\bf n]})_m{}^\ell Ric_{\ell}{}^m({^{(4)}g^{\bf [n]}})
 - (k^{[\bf n]})_\ell{}^j Ric_i{}^\ell ({^{(4)}g^{\bf [n]}})  + 2 (k^{[\bf n]})_\ell{}^\ell Ric_i{}^j({^{(4)}g^{\bf [n]}})\\
 &- ( (k^{[\bf n]})_\ell{}^\ell \de_i^j - (k^{[\bf n]})_i{}^j) Ric_m{}^m({^{(4)}g^{\bf [n]}})
\end{split}
\end{equation*}
and their coordinate derivatives go to $0$ faster than the given polynomial rate in $t$. 
\end{proof}

\begin{proposition}\label{prop:kn.higher}
For any $n\in \mathbb N$,
$$\sum_{r=1}^{s-1} t^{r} \| k^{\bf [n]}\|_{\dot{W}^{r,\infty}(\Sigma_t,g)} + \sum_{r=0}^{s-1} (t^{r+1} \|\nab^{\bf [n]} \kn \|_{\dot{W}^{r,\infty}(\Sigma_t,g)} + t^{r+2} \|\nab^{\bf [n]} \nab^{\bf [n]} \kn \|_{\dot{W}^{r,\infty}(\Sigma_t,g)}) \leq C_n t^{-1+\ve}, $$
and
$$\| \nab^{(s)} k^{\bf [n]}\|_{(L^\i + t^{s+ \f 52-\ve} L^2)(\Sigma_t,g)} \leq C_n t^{-s-1+\ve},\quad  \| \nab^{(s+1)} k^{\bf [n]}\|_{(L^\i + t^{s + \f 52 -\ve} L^2)(\Sigma_t,g)} \leq C_n t^{-s-2+\ve}.$$ 
\end{proposition}
\begin{proof}
This follows from Lemmas~\ref{lem:coord.v.geo.norm} {, \ref{lem:Wki.comp}} and \ref{lem:Wsi.save}, and the estimates for $\kn$ in coordinates given by \eqref{eq:main.parametrix.k.bd}.
\end{proof}

\begin{proposition}\label{prop:kn.lowest}
For any $n\in \mathbb N$,
$$\|k^{\bf [n]}\|_{L^\i(\Sigma_t,g)} \leq C_0 t^{-1} + C_n t^{-1+\ve}.$$
\end{proposition}
\begin{proof}
This is similar to the proof Proposition~\ref{prop:kn.higher}, except that we need to be more careful to check that the borderline $O(t^{-1})$ terms are \emph{independent} of $n$ (since Lemma~\ref{lem:coord.v.geo.norm} does not give an extra $t^\ve$ for the zeroth derivative). Nevertheless, by \eqref{eq:main.parametrix.k.bd}, it follows that the borderline contributions exactly come from $t^{-1}\kappa_i{ }^j$, which are manifestly independent of $n$.
\end{proof}

\begin{proposition}\label{prop:rdtkn}
For any $n\in \mathbb N$,
$$\sum_{r=1}^{s-1} t^{r} \|\rd_t k^{\bf [n]}\|_{\dot{W}^{r,\infty}(\Sigma_t,g)}  \leq C_n t^{-2+\ve},\quad \|\rd_t k^{\bf [n]}\|_{L^\i(\Sigma_t,g)} \leq C_0 t^{-2} + C_n t^{-2+\ve}.$$
\end{proposition}
\begin{proof}
This is a small variation to Propositions~\ref{prop:kn.higher} and \ref{prop:kn.lowest}. First, note that it suffices to control terms on the RHS of \eqref{eq:k.transport}. 
\begin{itemize}
\item For the term $Ric(g^{\bf [n-1]})_i{ }^j$, we use Lemmas~\ref{lem:coord.v.geo.norm} and \ref{lem:Wki.comp} and the estimate \eqref{eq:main.parametrix.Ric.bd}. (Note that there are no borderline terms in this estimate.)
\item For the term $(k^{\bf [n]})_\ell{ }^\ell(k^{\bf [n]})_i{ }^j$, we use Lemmas~\ref{lem:coord.v.geo.norm} and \ref{lem:Wki.comp} and the estimate \eqref{eq:main.parametrix.k.bd}. For the lowest order term, note that the borderline terms depend only on $t^{-1}\kappa_i{ }^j$ and are thus independent of $n$.
\end{itemize} \qedhere
\end{proof}

Once we obtain the estimates for $k^{\bf [n]}$, the estimates for $k$ can be controlled after using also the bootstrap assumptions \eqref{eq:BA4}.
\begin{proposition}\label{prop:k.general}
The following estimates hold for $k$:
$$\|k\|_{L^\i(\Sigma_t,g)}\leq C_0 t^{-1} + C_n t^{-1+\ve},\quad \sum_{r=1}^{s-2} t^r\|\nab^{(r)} k\|_{L^\i(\Sigma_t,g)} \leq C_n t^{-1+\ve}, $$
$$\|\nab^{(s-1)} k\|_{(L^\i + t^{s+ \f 52-\ve}L^2)(\Sigma_t,g)} \leq C_n t^{-s+\ve},\quad \|\nab^{(s)} k\|_{(L^\i + t^{s+ \f 52-\ve} L^2)(\Sigma_t,g)} \leq C_n t^{-s-1+\ve},$$
$$\|\rd_t k\|_{L^\i(\Sigma_t,g)}\leq C_0 t^{-2} + C_n t^{-2+\ve},\quad \sum_{r=1}^{s-3} t^r\|\nab^{(r)} \rd_t k\|_{L^\i(\Sigma_t,g)} \leq C_n t^{-2+\ve}, $$
$$\|\nab^{(s-2)} \rd_t k\|_{(L^\i + t^{s+ \f 52-\ve}L^2)(\Sigma_t,g)} \leq C_n t^{-s+\ve},\quad \|\nab^{(s-1)} \rd_t k\|_{(L^\i + t^{s+ \f 52-\ve} L^2)(\Sigma_t,g)} \leq C_n t^{-s-1+\ve}.$$
Moreover, the above estimates hold both when $k$ is replaced by $k^{\bf [n]}$ and $k^{(d)}$.
\end{proposition}
\begin{proof}
 That the estimates hold for $k^{\bf [n]}$ follows from Propositions~\ref{prop:kn.higher}, \ref{prop:kn.lowest} and \ref{prop:rdtkn}. That the estimates hold for $k^{(d)}$ follows from \eqref{eq:BA4} and Lemma~\ref{lem:kd.Li}.
 
Finally, since $k = k^{\bf [n]} + k^{(d)}$, the estimates also hold for $k$.
\end{proof}

\begin{proposition}\label{prop:hn}
For any $n\in \mathbb N$, 
$$\sum_{r=1}^{s-1} t^{r} \| h^{\bf [n]}\|_{\dot{W}^{r,\infty}(\Sigma_t,g)} + \sum_{r=0}^{s-1} (t^{r+1} \| \rd h^{\bf [n]}\|_{\dot{W}^{r,\infty}(\Sigma_t,g)} + t^{r+2} \|\nab^{\bf [n]} \rd h^{\bf [n]}\|_{\dot{W}^{r,\infty}(\Sigma_t,g)}) \leq C_n t^{-1+\ve},$$
and
$$\|h^{\bf [n]} \|_{L^\i(\Sigma_t,g)} \leq C_0 t^{-1} + C_n t^{-1+\ve}.$$
\end{proposition}
\begin{proof}
Recalling that we set $h^{\bf [n]} = (\kn)_\ell{ }^\ell$, this can be proven in the same way as Propositions~\ref{prop:kn.higher} and \ref{prop:kn.lowest}. \qedhere
\end{proof}




\begin{proposition}\label{prop:Ric.general}
The following estimates hold for $Ric(g)$:
$$\sum_{r = 0}^{s-3}  {t^r} \|Ric(g)\|_{W^{r,\infty}(\Sigma_t,g)} \leq C_nt^{-2+\ve},$$ 
$$\|\nab^{(s-2)} Ric(g)\|_{(L^\i + t^{s+\f 52 - \ve} L^2)(\Sigma_t,g)} \leq C_n t^{-s+\ve},\quad \|\nab^{(s-1)} Ric(g)\|_{(L^\i + t^{s+\f 52 - \ve} L^2)(\Sigma_t,g)} \leq C_n t^{-s-1+\ve}.$$
\end{proposition}
\begin{proof}
For simplicity, we write in this proof $Ric = Ric(g)$, $Ric^{\bf [n]} = Ric(g^{\bf [n]})$ and similarly for the Riemann curvature tensor.

First, notice that by Lemma~\ref{lem:coord.v.geo.norm} and \eqref{eq:main.parametrix.Ric.bd}, it follows that
$$\sum_{r=0}^{s-1}  {t^r} \|Ric^{\bf [n]}\|_{W^{r,\infty}(\Sigma,g^{\bf [n]})} \leq C_n t^{-2+\ve}.$$
As a result, Lemmas~\ref{lem:Wki.comp} and \ref{lem:Wsi.save} imply that all the desired estimates when $Ric$ is replaced by $Ric^{\bf [n]}$.

It thus remains to estimate the difference $Ric - Ric^{\bf [n]}$. We will bound the difference of the full Riemann curvature tensor; the bounds for the Ricci curvature tensor of course follow immediately. We compute
\begin{equation}
\begin{split}
&\: (Riem - Riem^{\bf [n]})^\ell_{ijk} \\
=&\:  \rd_i (\Gamma - \Gamma^{\bf [n]})^\ell_{jk} - \rd_j (\Gamma - \Gamma^{\bf [n]})^\ell_{ik} + \Gamma^p_{jk} \Gamma^\ell_{ip} - \Gamma^p_{ik} \Gamma^\ell_{jp} - (\Gamma^{\bf [n]})^p_{jk} (\Gamma^{\bf [n]})^\ell_{ip} + (\Gamma^{\bf [n]})^p_{ik} (\Gamma^{\bf [n]})^\ell_{jp} \\
=&\: \rd_i (\Gamma - \Gamma^{\bf [n]})^\ell_{jk} - \Gamma_{ij}^p (\Gamma - \Gamma^{\bf [n]})^\ell_{pk} - \Gamma^p_{ik} (\Gamma - \Gamma^{\bf [n]})^\ell_{jp} + \Gamma^\ell_{ip} (\Gamma - \Gamma^{\bf [n]})^p_{jk}\\
&\: -  \rd_j (\Gamma - \Gamma^{\bf [n]})^\ell_{ik} + \Gamma_{ij}^p (\Gamma - \Gamma^{\bf [n]})^\ell_{pk} + \Gamma^p_{jk} (\Gamma - \Gamma^{\bf [n]})^\ell_{ip}  - \Gamma^\ell_{jp}(\Gamma - \Gamma^{\bf [n]})^p_{ik} \\
&\: - (\Gamma- \Gamma^{\bf [n]})^p_{jk} (\Gamma - \Gamma^{\bf [n]})^\ell_{ip}  + (\Gamma- \Gamma^{\bf [n]})^\ell_{jp}(\Gamma - \Gamma^{\bf [n]})^p_{ik}\\
=&\: \nab_i (\Gamma - \Gamma^{\bf [n]})^\ell_{jk} - \nab_j  (\Gamma - \Gamma^{\bf [n]})^\ell_{ik} - (\Gamma- \Gamma^{\bf [n]})^p_{jk} (\Gamma - \Gamma^{\bf [n]})^\ell_{ip}  + (\Gamma- \Gamma^{\bf [n]})^\ell_{jp}(\Gamma - \Gamma^{\bf [n]})^p_{ik}.
\end{split}
\end{equation}
Combining this with \eqref{eq:diff.Gamma.id} and the bootstrap assumptions \eqref{eq:BA2} and \eqref{eq:BA3}, it is easy to see that $Riem - Riem^{\bf [n]}$ can be controlled by 
$$\sum_{r = 0}^{s-3}\|Riem - Riem^{\bf [n]}\|_{W^{r,\infty}(\Sigma_t,g)},\quad \sum_{r=s-2}^{s-1} \|Riem - Riem^{\bf [n]}\|_{H^r(\Sigma_t, g)} \leq C_0.$$
This concludes the proof of the proposition. \qedhere
\end{proof}

\subsubsection{Commutator estimates}

We will often use the commutator formula between the Lie derivative in $\partial_t$ and covariant derivatives in the spatial directions:

\begin{proposition}\label{prop:commutation.formula}
The following commutation formula holds for any $(m,l)$ $\Sigma$-tangent tensor $\mathcal T$:
\begin{equation*}
\begin{split}
[\rd_t, \, \nab_a ] \mathcal T^{j_1\ldots j_m}_{i_1\ldots i_l} = &\: - \sum_{r=1}^l ((g^{-1})^{be}g_{m (i_r|}\nabla_e  k_{|a)}{}^m - \nabla_{(a} k_{i_r)}{}^b - (g^{-1})^{be} g_{d(a}\nabla_{i_r)}k_e{}^d) \mathcal{T}^{j_1\ldots j_m}_{i_1\ldots \underset{r-\text{th index}}{b}\ldots i_l} \\
&\: + \sum_{r=1}^m ((g^{-1})^{j_r e}g_{m (b|}\nabla_e  k_{|a)}{}^m-\nabla_{(a} k_{b)}{}^{j_r} - (g^{-1})^{j_r e} g_{d(a}\nabla_{b)}k_e{}^{d})\mathcal{T}^{j_1\ldots\overset{r-\text{th index}}{b}\ldots j_m}_{i_1\ldots i_l}.
\end{split}
\end{equation*}

\end{proposition}
\begin{proof}
A direct computation shows
\begin{align*}
\partial_t\nabla_a\mathcal{T}^{j_1\ldots j_m}_{i_1\ldots i_l} = \nabla_a\partial_t\mathcal{T}^{j_1\ldots j_m}_{i_1\ldots i_l} - \sum_{r=1}^l\partial_t\Gamma^b_{a i_r}\mathcal{T}^{j_1\ldots j_m}_{i_1\ldots \underset{r-\text{th index}}{b}\ldots i_l} + \sum_{r=1}^m\partial_t\Gamma^{j_r}_{a b }\mathcal{T}^{j_1\ldots\overset{r-\text{th index}}{b}\ldots j_m}_{i_1\ldots i_l}.
\end{align*}
On the other hand, we compute  {using \eqref{eq:hyp.sys} and \eqref{eq:g-1.transport} that}
\begin{align}
\notag\partial_t\Gamma_{ac}^b=&\, (g^{-1})^{d(b} k_d{ }^{l)}(\partial_ag_{cl}+\partial_cg_{al}-\partial_lg_{ac}) + (g^{-1})^{bl}(\partial_l (g_{d(a} k_{c)}{ }^d)-\partial_a (g_{d(l} k_{c)}{ }^d)-\partial_c (g_{d(a} k_{l)}{ }^d))\\
\notag=&\, 2 (g^{-1})^{d(b} k_d{ }^{l)}\Gamma_{ac}^m g_{ml} +  (g^{-1})^{bl}(\partial_l (g_{d(a} k_{c)}{ }^d)-\partial_a (g_{d(l} k_{c)}{ }^d)-\partial_c (g_{d(a} k_{l)}{ }^d)) \\
=&\,(g^{-1})^{bl}(\nabla_l  (g_{d(a} k_{c)}{ }^d) - \nabla_a (g_{d(l} k_{c)}{ }^d) - \nabla_c (g_{d(a} k_{l)}{ }^d))\notag\\
=&\, (g^{-1})^{bl} g_{d(a|} \nab_l k_{|c)}{ }^d - \nab_{(a} k_{c)}{ }^b - (g^{-1})^{bl} g_{d(a} \nab_{c)} k_l{ }^d \notag.
\end{align}
%

Combining these computations yields the desired formula. \qedhere
%
%
%

\end{proof}

\begin{proposition}\label{prop:dt.comm.est}
Let $\mathcal T$ be an $(m,l)$ $\Sigma$-tangent tensor.

For $0\leq r \leq s-1$,
\begin{equation}\label{eq:comm.form.est.main.1}
\|[\rd_t, \nab_a] \mathcal T\|_{H^r(\Sigma_t,g)} \leq C_n \sum_{ r_1 + r_2 = r } t^{-2-r_1+\ve} \|\mathcal T \|_{H^{r_2}(\Sigma_t,g)}.
\end{equation}

Consequently, for $0\leq k\leq s$, iterated commutators can be bounded as follows:
\begin{equation}\label{eq:comm.form.est.main.2}
\|[\rd_t, \nab_{i_1}\cdots \nab_{i_k}] \mathcal T\|_{L^2(\Sigma_t,g)} \leq C_n \sum_{r'=0}^{k-1} t^{-2-r'+\ve} \|\mathcal T\|_{H^{k-r'-1}(\Sigma_t,g)}.
\end{equation}

Finally, if $\mathcal T$ is a \underline{scalar} function, then in fact \eqref{eq:comm.form.est.main.2} holds for $0\leq k \leq s+1$.

\end{proposition}
\begin{proof}
\pfstep{Step~1: Proof of \eqref{eq:comm.form.est.main.1}} Using Proposition~\ref{prop:commutation.formula}, we have the estimate
\begin{equation}\label{eq:comm.form.est.1}
\begin{split}
&\: \|[\rd_t, \nab_a] \mathcal T\|_{H^r(\Sigma_t,g)} \\
\leq &\: C_0 \sum_{ \substack { r_1 + r_2 = r \\ r_1 \leq s-3}} \|\nab^{(r_1)} \nab k\|_{L^\i(\Sigma_t,g)} \|\nab^{(r_2)} \mathcal T\|_{L^2(\Sigma_t,g)} \\
&\: + C_0 \sum_{ \substack { r_1 + r_2 = r \\ r_1 > s-3}} \|\nab^{(r_1)} \nab k\|_{(L^\i + t^{s+\f 52-\ve} L^2)(\Sigma_t,g)} \|\nab^{(r_2)} \mathcal T\|_{(L^2 \cap t^{-s-\f 52+\ve} L^\i)(\Sigma_t,g)}. 
\end{split}
\end{equation}

We estimate each of the terms in \eqref{eq:comm.form.est.1}. Using the estimates in  {Proposition~\ref{prop:k.general}}, the first term in \eqref{eq:comm.form.est.1} can be bounded above as follows:
\begin{equation}\label{eq:comm.form.est.2}
\begin{split}
 \sum_{ \substack { r_1 + r_2 = r \\ r_1 \leq s-3}} \|\nab^{(r_1)} \nab k\|_{L^\i(\Sigma_t,g)} \|\nab^{(r_2)} \mathcal T\|_{L^2(\Sigma_t,g)} 
\leq &\: C_n \sum_{ r_1 + r_2 = r } t^{-2-r_1+\ve} \|\mathcal T \|_{H^{r_2}(\Sigma_t,g)}.
\end{split}
\end{equation}

Before handling the second term in \eqref{eq:comm.form.est.1}, we first note make the following observations on the numerology:
\begin{itemize}
\item When $r_1 >s-3$, since we have $r_1 + r_2 = r \leq s - 1$, we have either $r_2 = 0$ or $r_2 = 1$. In particular, $r_2+2\leq r$. 
\end{itemize}
We can thus bound the second term in \eqref{eq:comm.form.est.1} using Proposition~\ref{prop:k.general}, \eqref{eq:Sobolev.used} and the above observations as follows:
\begin{equation}\label{eq:comm.form.est.3}
\begin{split}
&\: \sum_{ \substack { r_1 + r_2 = r \\ r_1 > s-3}} \|\nab^{(r_1)} \nab k\|_{(L^\i + t^{s+ \f 52-\ve} L^2)(\Sigma_t,g)} \|\nab^{(r_2)} \mathcal T\|_{(L^2 \cap t^{-s-\f 52+\ve} L^\i)(\Sigma_t,g)} \\
\leq &\: C_n \sum_{ \substack { r_1 + r_2 = r \\ r_1 > s-3}} t^{-2-r_1+\ve} (\sum_{r'=r_2}^{r_2+2} t^{r'-r_2}\|\mathcal T \|_{H^{r'}(\Sigma_t,g)}) \leq C_n \sum_{ r_1 + r_2 = r } t^{-2-r_1+\ve} \|\mathcal T \|_{H^{r_2}(\Sigma_t,g)},
\end{split}
\end{equation}
where the very last estimate follows simply after relabelling.

Combining \eqref{eq:comm.form.est.2} and \eqref{eq:comm.form.est.3} yields \eqref{eq:comm.form.est.main.1}.

\pfstep{Step~2: Proof of \eqref{eq:comm.form.est.main.2}}
When $0\leq k \leq s$, we compute using the triangle inequality and \eqref{eq:comm.form.est.main.1} to obtain
\begin{equation}\label{eq:iterated.com.est}
\begin{split}
&\: \| [\rd_t, \nab_{i_1} \cdots \nab_{i_k}] \mathcal T \|_{L^2(\Sigma_t,g)} \\
= &\: \| [\rd_t, \nab_{i_1}] \nab_{i_2} \cdots \nab_{i_k} \mathcal T + \cdots + \nab_{i_1} \cdots [\rd_t, \nab_{i_\ell}] \cdots \nab_{i_k} \mathcal T + \cdots + \nab_{i_1} \cdots \nab_{i_{k-1}} [\rd_t, \nab_{i_k}] \mathcal T \|_{L^2(\Sigma_t,g)} \\
\leq &\: C_0 \sum_{r=1}^{k} \|[\rd_t,\nab_{i_{r}}] \nab_{i_{r+1}} \cdots \nab_{i_k} \mathcal T \|_{H^{r-1}(\Sigma_t,g)}
\leq C_n \sum_{r=1}^{k} \sum_{r_1+r_2 = r -1} t^{-2-r_1+\ve} \|\nab_{i_{r+1}} \cdots \nab_{i_k} \mathcal T \|_{H^{r_2}(\Sigma_t,g)} \\
\leq &\: C_n \sum_{r=1}^{k} \sum_{r_1+r_2 = r -1} t^{-2-r_1+\ve} \|\mathcal T\|_{H^{r_2+k-r}(\Sigma_t,g)}
\leq C_n  \sum_{r'=0}^{k-1} t^{-2-r'+\ve} \|\mathcal T\|_{H^{k-r'-1}(\Sigma_t,g)}.
\end{split}
\end{equation}
This yields \eqref{eq:comm.form.est.main.2}.

Finally, for a scalar function $f$, $[\rd_t, \nab_i] f = 0$. Hence, in \eqref{eq:iterated.com.est}, we sum only up to $r = k-1$. As a result, we can take up to $k = s+1$. This gives the desired improvement for scalar functions.
\end{proof}

\subsubsection{Estimates for general equations}

\begin{proposition}[Transport estimates]\label{prop:transport.gen}
Let $\mathcal T$ be an $(m,l)$ $\Sigma$-tangent tensor. Then
\begin{equation}\label{eq:transport.gen.1}
\f{\ud}{\ud t} [t^{-M} \|\mathcal T\|_{L^2(\Sigma_t)}^2] + \f {M}{t}[t^{-M} \|\mathcal T\|_{L^2(\Sigma_t,g)}^2] - 2 t^{-M} \int_{\Sigma_t} |\langle \mathcal T, \rd_t \mathcal T\rangle_g| \,\mathrm{vol}_{\Sigma_t} \leq \f {C_0}{t}  [t^{-M} \|\mathcal T\|_{L^2(\Sigma_t,g)}^2].
\end{equation}
In particular, 
\begin{equation}\label{eq:transport.gen.2}
\f{\ud}{\ud t} [t^{-M} \|\mathcal T\|_{L^2(\Sigma_t)}^2] + \f {M}{t}[t^{-M} \|\mathcal T\|_{L^2(\Sigma_t,g)}^2]  \leq \f {C_0}{t}[t^{-M} \|\mathcal T\|_{L^2(\Sigma_t,g)}^2] + t^{-M+1} \|\rd_t \mathcal T\|_{L^2(\Sigma_t,g)}^2.
\end{equation}
\end{proposition}
\begin{proof}
We first note  {that by \eqref{eq:hyp.sys}}
\begin{equation}\label{eq:first.variation}
\f{\ud}{\ud t} \int_{\Sigma_t} f \, \mathrm{vol}_{\Sigma_t} = \int_{\Sigma_t} (\rd_t f -  { k_\ell{ }^\ell } ) f  \, \mathrm{vol}_{\Sigma_t}.
\end{equation}
We will apply \eqref{eq:first.variation} to $f = t^{-M} |\mathcal T|_g^2$. A direct computation shows that
\begin{equation}
\begin{split}
\rd_t f = &\: -M t^{-M-1} |\mathcal T|_g^2 + 2 t^{-M} \langle \mathcal T, \rd_t\mathcal T\rangle_g \\
&\: + 2 t^{-M} \sum_{r = 1}^\ell (g^{-1})^{i_1 i_1'} \cdots \{ (g^{-1})^{\ell (i_r} k_\ell{}^{i_r')}\} \cdots (g^{-1})^{i_\ell i_\ell'} g_{j_1 j_1'} \cdots g_{j_m j_m'}\mathcal T_{i_1\cdots i_\ell}^{j_1\cdots j_m} \mathcal T_{i_1'\cdots i_\ell'}^{j_1'\cdots j_m'} \\
&\: -2 t^{-M} \sum_{s = 1}^\ell (g^{-1})^{i_1 i_1'} \cdots (g^{-1})^{i_\ell i_\ell'} g_{j_1 j_1'} \cdots \{g_{\ell (j_s} k_{j_s')}{ }^\ell \} \cdots g_{j_m j_m'}\mathcal T_{i_1\cdots i_\ell}^{j_1\cdots j_m} \mathcal T_{i_1'\cdots i_\ell'}^{j_1'\cdots j_m'},
\end{split}
\end{equation} 
which implies, using Proposition~\ref{prop:k.general}, that
\begin{equation}\label{eq:transport.gen.1.0}
\f{\ud}{\ud t}[t^{-M} |\mathcal T|_g^2] + \f M t [ t^{-M} |\mathcal T|_g^2] \leq 2  t^{-M} |\langle \mathcal T, \rd_t \mathcal T\rangle_g | + C_0 t^{-M-1} |\mathcal T|_g^2.
\end{equation}
The pointwise inequality \eqref{eq:transport.gen.1.0} implies \eqref{eq:transport.gen.1} immediately after integrating over $\Sigma_t$, using \eqref{eq:first.variation}, and applying again the estimates in Proposition~\ref{prop:k.general}.

Finally, to derive \eqref{eq:transport.gen.2}, we simply note that by the Cauchy--Schwarz inequality,
$$2 t^{-M} \int_{\Sigma_t} |\langle \mathcal T, \rd_t \mathcal T\rangle_g| \,\mathrm{vol}_{\Sigma_t} \leq t^{-M-1} \|\mathcal T\|_{L^2(\Sigma_t,g)}^2 + t^{-M+1} \|\rd_t \mathcal T\|_{L^2(\Sigma_t,g)}^2.$$
\end{proof}

\begin{proposition}[Energy estimates for wave equations]\label{prop:EE.wave}
Let $\mathcal T$ be an $(m,l)$ $\Sigma_t$-tangent tensor such that $(-\rd_t^2 + \Delta_g)\mathcal T = \mathcal F$ for some $(m,l)$ $\Sigma_t$-tangent tensor $\mathcal F$.
Then
\begin{equation*}
\begin{split}
&\: \f{\ud}{\ud t} [t^{-M} (\|\rd_t \mathcal T\|_{L^2(\Sigma_t,g)}^2 + \|\nab \mathcal T\|_{L^2(\Sigma_t,g)}^2 + t^{-2} \|\mathcal T\|_{L^2(\Sigma_t,g)}^2)] \\
&\: + \f {M}{t}[t^{-M}(\|\rd_t \mathcal T\|_{L^2(\Sigma_t,g)}^2 + \|\nab \mathcal T\|_{L^2(\Sigma_t,g)}^2 + t^{-2} \|\mathcal T\|_{L^2(\Sigma_t,g)}^2)]\\ \leq &\:\f {(C_0 + C_n t^\ve)}{t}  [t^{-M} (\|\rd_t \mathcal T\|_{L^2(\Sigma_t,g)}^2 + \|\nab \mathcal T\|_{L^2(\Sigma_t,g)}^2) +  t^{-M-2} \|\mathcal T\|_{L^2(\Sigma_t,g)}^2] + t^{-M+1} \|\mathcal F\|_{L^2(\Sigma_t,g)}^2.
\end{split}
\end{equation*}
\end{proposition}
\begin{proof}
Denote by $E$ terms bounded by $\f{(C_0+ C_n t^\ve)}{t} [t^{-M} (\|\rd_t \mathcal T\|_{L^2(\Sigma_t,g)}^2 + \|\nab \mathcal T\|_{L^2(\Sigma_t,g)}^2) +  t^{-M-2} \|\mathcal T\|_{L^2(\Sigma_t,g)}^2]$.

\pfstep{Step~1: Controlling the first order terms} Applying  {\eqref{eq:transport.gen.1} in} Proposition~\ref{prop:transport.gen}, integrating by parts and using the Cauchy--Schwarz inequality,
\begin{equation}\label{eq:EE.wave.1}
\begin{split}
&\: \f{\ud}{\ud t} [t^{-M} (\|\rd_t \mathcal T\|_{L^2(\Sigma_t,g)}^2 + \|\nab \mathcal T\|_{L^2(\Sigma_t,g)}^2)] +  \f {M}{t}[t^{-M}(\|\rd_t \mathcal T\|_{L^2(\Sigma_t,g)}^2 + \|\nab \mathcal T\|_{L^2(\Sigma_t,g)}^2)] \\
= &\: 2 t^{-M} \int_{\Sigma_t} (\langle \rd_t \mathcal T,\, \rd_t^2 \mathcal T\rangle_g + \langle \nab \mathcal T,\, \rd_t \nab \mathcal T \rangle_g)\, \mathrm{vol}_{\Sigma_t} + E\\
=&\: 2 t^{-M} \int_{\Sigma_t} (- \langle \rd_t \mathcal T,\, \mathcal F\rangle_g + \langle \rd_t \mathcal T,\, \Delta_g \mathcal T \rangle_g + \langle \nab \mathcal T,\, \rd_t \nab \mathcal T \rangle_g)\, \mathrm{vol}_{\Sigma_t} + E \\
=&\: -2 t^{-M} \int_{\Sigma_t}  \langle \rd_t \mathcal T,\, \mathcal F\rangle_g \, \mathrm{vol}_{\Sigma_t} - 2 t^{-M} \int_{\Sigma_t} (\langle \nab \rd_t \mathcal T,\, \nab \mathcal T \rangle_g - \langle \nab \mathcal T,\, \rd_t \nab \mathcal T \rangle_g)\, \mathrm{vol}_{\Sigma_t}+ E \\
\leq &\: t^{-M-1} \|\rd_t \mathcal T\|_{L^2(\Sigma_t,g)}^2 + t^{-M+1} \|\mathcal F\|_{L^2(\Sigma_t,g)}^2 + E \leq t^{-M+1} \|\mathcal F\|_{L^2(\Sigma_t,g)}^2 + E,
\end{split}
\end{equation}
where we have used that by H\"older's inequality and the following commutator estimate (which uses  {Proposition~\ref{prop:dt.comm.est}})
\begin{equation*}
\begin{split}
&\: \left| \int_{\Sigma_t} \langle \nab \mathcal T,\, [\rd_t, \nab] \mathcal T \rangle_g \right| \leq C_n t^{-2+\ve} \|\nab \mathcal T\|_{L^2(\Sigma_t,g)}\|\mathcal T\|_{L^2(\Sigma_t,g)} \\
\leq &\: C_n t^{-1+\ve} \|\nab\mathcal T\|_{L^2(\Sigma_t,g)}^2 + C_n t^{-3+\ve} \|\mathcal T\|_{L^2(\Sigma_t,g)}^2.
\end{split}
\end{equation*}

\pfstep{Step~2: Controlling the zeroth order term} It remains to control the zeroth order term $\|\mathcal T\|_{L^2(\Sigma_t,g)}^2$. For this we simply use Proposition~\ref{prop:transport.gen}, and then use \eqref{eq:EE.wave.1} to obtain
\begin{equation}\label{eq:EE.wave.2}
\begin{split}
&\: \f{\ud}{\ud t} [t^{-M-2} \|\mathcal T\|_{L^2(\Sigma_t,g)}^2] + \f{M+2}{t} t^{-M-2}\|\mathcal T\|_{L^2(\Sigma_t,g)}^2 \\
\leq &\: \f{C_0}{t} t^{-M-2} \|\mathcal T\|_{L^2(\Sigma_t,g)}^2 + t^{-M-1} \|\rd_t\mathcal T\|_{L^2(\Sigma_t,g)}^2\leq E
\end{split}
\end{equation}
Summing \eqref{eq:EE.wave.1} and \eqref{eq:EE.wave.2}, we obtain the desired estimate.
\end{proof}

\subsection{Energy estimates for the wave equation for $k$}\label{subsec:mainest.k}

In this subsection we continue to work under the assumptions of Theorem~\ref{thm:bootstrap}. In particular, we assume the validity of the bootstrap assumptions \eqref{eq:BA1}--\eqref{eq:BA4}.

We insert \eqref{adkdhd.2} into \eqref{eq:hyp.sys} to obtain evolution equations for the difference $(k^{(d)})_i{}^j$:
\begin{equation}
\begin{split}
\label{eq:k.diff}\rd_t^2 (k^{(d)})_i{ }^j =&\: \Delta_g (k^{(d)})_i{ }^j + (k\star k \star k - \kn \star \kn \star \kn)_i{ }^j \\
&\: +(\rd_t k \star k - \rd_t \kn \star \kn)_i{ }^j
+(I_{k^{[{\bf n}]}})_i{}^j+\mathcal{B}_i{}^j,
\end{split}
\end{equation}
where the terms $(k\star k\star k)$, $(\partial_t k\star k)$ are as defined in \eqref{kstark}, $(I_{k^{[{\bf n}]}})_i{}^j$ is as defined in Proposition \ref{prop:inhomo}, and $\mathcal{B}_i{}^j$ denotes the following terms:
\begin{equation}\label{eq:Bij}
\mathcal B_i{ }^j = -\nab_i \nab^j h  + \nab_i^{\bf [n]} (\nab^{\bf [n]})^j h^{\bf [n]} + \Delta_g (k^{[\bf n]})_i{ }^j - \Delta_{g^{\bf [n]}} (k^{\bf [n]})_i{ }^j.
\end{equation}

The following is the main energy estimates for $k^{(d)}$:
\begin{proposition}\label{prop:main.wave.est}
Given $N\in \mathbb N$, let $n\in \mathbb N$ be sufficiently large so that the estimates in Proposition~\ref{prop:inhomo} hold. Then
\begin{equation*}
\begin{split}
&\: \f{\ud}{\ud t} [t^{-2N-2s}( \sum_{r=0}^{s-1} t^{2r+2} \|\rd_t \nab^{(r)} k^{(d)}\|^2_{L^2(\Sigma_t,g)} +  \sum_{r=0}^{s} t^{2r} \|k^{(d)}\|^2_{\dot H^r(\Sigma_t,g)} )]\\
&\: + \f{2N+2s}{t} [t^{-2N-2s}( \sum_{r=0}^{s-1} t^{2r+2} \|\rd_t \nab^{(r)} k^{(d)}\|^2_{L^2(\Sigma_t,g)} +  \sum_{r=0}^{s} t^{2r} \|k^{(d)}\|^2_{\dot H^r(\Sigma_t,g)} )] \\
\leq &\: (C_0 t^{-1} + C_n t^{-1+\ve} ) t^{-2N - 2s}\mathcal E_s(t) + C_n t^3,
\end{split}
\end{equation*}
where, as before, we have used the notation $\nab^{(r)} = \nab_{i_1} \cdots \nab_{i_r}$.
\end{proposition}
\begin{proof}
For $0\leq r \leq s-1$, we differentiate \eqref{eq:k.diff} by $\nab^{(r)}$ to obtain the following wave equation for $\nab^{(r)} k^{(d)}$: 
\begin{equation}\label{eq:higher.order.wave.gor.kd}
\begin{split}
&\: -\rd_t^2 \nab^{(r)}_{i_1\cdots i_r} (k^{(d)})_i{ }^j + \Delta_g \nab^{(r)}_{i_1\cdots i_r} (k^{(d)})_i{ }^j \\
=&\: - \nab^{(r)}_{i_1\cdots i_r} (I_{\kn})_i{ }^j - \nab^{(r)}_{i_1\cdots i_r} \mathcal B_i{ }^j 
 - \nab^{(r)}_{i_1\cdots i_r}(k\star k \star k - \kn \star \kn \star \kn)_i{ }^j \\
 &\: - \nab^{(r)}_{i_1\cdots i_r}(\rd_t k \star k - \rd_t \kn \star \kn)_i{ }^j 
 - [\rd_t^2, \nab^{(r)}_{i_1\cdots i_r}](k^{(d)})_i{ }^j  + [\Delta_g, \nab^{(r)}_{i_1\cdots i_r}](k^{(d)})_i{ }^j.
\end{split}
\end{equation}

For every $0\leq r \leq s-1$, our goal is to show that 
\begin{equation}\label{eq:main.k.inhomo}
\begin{split}
 t^r \|-\rd_t^2 \nab^{(r)}_{i_1\cdots i_r} (k^{(d)})_i{ }^j + \Delta_g \nab^{(r)}_{i_1\cdots i_r} (k^{(d)})_i{ }^j\|_{L^2(\Sigma_t,g)} 
\leq &\: (C_0 t^{-2} + C_n t^{-2+\ve} ) \mathcal E_s^{\f 12}(t) + C_n t^{N+s},
\end{split}
\end{equation}
after which we will apply Proposition~\ref{prop:EE.wave}. 

The proof of \eqref{eq:main.k.inhomo} will be achieved in Steps~1--5 below in which we bound each term on the RHS of \eqref{eq:higher.order.wave.gor.kd}.

\pfstep{Step~1: Bounding the inhomogeneous terms}
For $0\leq r \leq s-1$, by Proposition~\ref{prop:inhomo},
\begin{equation}
\|I_{\kn}\|_{H^r} \leq C_n t^{N+s-r}.
\end{equation}

\pfstep{Step~2: Bounding the terms in $\mathcal B_i{ }^j$} Recall from \eqref{eq:Bij} that $\mathcal B_i{ }^j$ consists of $h$ terms and $k$ terms. We first compute the exact form of the $h$ terms:
\begin{equation}\label{eq:h.terms.in.B.exact}
\begin{split}
&\: -(g^{-1})^{j\ell} \nab_i \partial_\ell h + ((g^{\bf [n]})^{-1})^{j\ell} \nab_i^{\bf [n]} \partial_\ell h^{\bf [n]} \\
= &\: - (g^{-1})^{j\ell} \nab_i \partial_\ell h^{(d)} - (g^{-1})^{j\ell} \nab^{(d)}_i \rd_\ell h^{\bf [n]} -  ((g^{-1})^{(d)})^{j\ell} \nab^{\bf [n]} \rd_\ell h^{\bf [n]}.
\end{split}
\end{equation}
From \eqref{eq:h.terms.in.B.exact}, the triangle inequality and H\"older's inequality, it follows that
\begin{equation}
\begin{split}
&\: \| -(g^{-1})^{j\ell} \nab_i \partial_\ell h + ((g^{\bf [n]})^{-1})^{j\ell} \nab_i^{\bf [n]} \partial_\ell h^{\bf [n]}\|_{\dot H^r(\Sigma_t,g)} \\
\ls &\:  \underbrace{\|\nab \rd_\ell h^{(d)}\|_{\dot H^{r}(\Sigma_t,g)}}_{=:\mathrm{I}}  + \underbrace{\sum_{\substack{ r_1+ r_2 = r }}  \| \nab^{(d)}\|_{\dot H^{r_1}(\Sigma_t,g)}  \| \rd_\ell h^{\bf [n]}\|_{\dot W^{r_2,\infty}(\Sigma_t,g)} }_{=:\mathrm{II}} \\
&\: + \underbrace{\sum_{\substack{ r_1+ r_2 = r }} \|(g^{-1})^{(d)} \|_{\dot H^{r_1}(\Sigma_t,g)} \|\nab^{\bf [n]} \rd_\ell h^{\bf [n]}\|_{\dot W^{r_2,\infty}(\Sigma_t,g)}}_{=:\mathrm{III}} .
\end{split}
\end{equation}

Term $\mathrm{I}$ can be directly estimated by the definition of $\mathcal E_s(t)$:
\begin{equation}
\mathrm{I} \leq \|h^{(d)}\|_{\dot H^{r+2}(\Sigma_t,g)} \leq t^{-r-2} \mathcal E_s^{\f 12}(t) .
\end{equation}

By Proposition~\ref{prop:hn}, Lemma~\ref{lem:est.nabd}, and the definition of $\mathcal E_s(t)$, we have
\begin{equation}
\begin{split}
\mathrm{II} \leq &\: C_n \sum_{\substack{ r_1+ r_2 = r}} t^{-2-r_2+\ve} (\|g^{(d)}\|_{H^{r_1+1}(\Sigma_t,g)} + \| (g^{-1})^{(d)} \|_{H^{r_1+1}(\Sigma_t,g)}) \\
\leq &\:  C_n \sum_{\substack{r_1 + r_2 = r}} t^{-2-r_2-r_1+\ve} \mathcal E_s^{\f 12}(t) \leq C_n t^{-r-2+\ve}\mathcal E_s^{\f 12}(t).
\end{split}
\end{equation}

For term $\mathrm{IV}$, we use Proposition~\ref{prop:hn} and the definition of $\mathcal E_s(t)$ to obtain
\begin{equation}
\begin{split}
\mathrm{III} \leq &\: C_n  \sum_{r_1+r_2 = r} t^{-3-r_2+\ve} \|(g^{-1})^{(d)} \|_{\dot H^{r_1}(\Sigma_t,g)} \leq C_n  \sum_{r_1+r_2 = r} t^{-2-r_1-r_2+\ve} \mathcal E_s^{\f 12}(t) \leq C_n t^{-r-2+\ve} \mathcal E_s^{\f 12}(t). 
\end{split}
\end{equation}

For the $k$ terms in $\mathcal B_i{ }^j$, we compute
\begin{equation}\label{eq:k.terms.in.B.exact}
\begin{split}
&\: \Delta_g (k^{[\bf n]})_i{ }^j - \Delta_{g^{\bf [n]}} (k^{\bf [n]})_i{ }^j \\
= &\: (g^{(d)})^{m \ell} \nab^{[\bf n]}_m \nab^{[\bf n]}_\ell (k^{\bf [n]})_i{ }^j + g^{m\ell} \nab_m^{(d)}  \nab^{[\bf n]}_\ell (k^{[\bf n]})_i{ }^j + g^{m\ell} \nab_m \nab^{(d)}_\ell (k^{[\bf n]})_i{ }^j.
\end{split}
\end{equation}

The terms in \eqref{eq:k.terms.in.B.exact} are similar to those in \eqref{eq:h.terms.in.B.exact} (with $k$ taking the place of $h$) except --- importantly --- that \eqref{eq:k.terms.in.B.exact} does not 
contain second derivative terms of $k^{(d)}$. This is important because while our energy controls up to $s+1$ derivatives of $h^{(d)}$, it only controls up to $s$ derivatives of $k^{(d)}$. Other than this difference, the remaining terms in \eqref{eq:k.terms.in.B.exact} can in fact be controlled very similarly as those in \eqref{eq:h.terms.in.B.exact}. We will therefore omit the details and simply give the final estimate:
\begin{equation}
\| \nab^{(r)}\mathcal B_i{ }^j \|_{L^2(\Sigma_t,g)} \ls (C_0 + C_n t^{\ve}) t^{-r-2} \mathcal E_s^{\f 12}(t)
\end{equation}

\pfstep{Step~3: Bounding the difference of the nonlinear terms} In this step we control the $H^r$ norm (for $0\leq r \leq s-1$) of $k\star k \star k - \kn \star \kn \star \kn$ and $\rd_t k \star k - \rd_t \kn \star \kn$.

We begin with $k\star k \star k - \kn \star \kn \star \kn$. For $0\leq r<s-1$, we use H\"older's inequality, Proposition~\ref{prop:k.general}, and the definition of $\mathcal E_s(t)$ to obtain
\begin{equation}\label{eq:k.wave.kkk-kkk}
\begin{split}
&\: \|k\star k \star k - \kn \star \kn \star \kn\|_{H^r(\Sigma_t,g)} \\
\leq &\: C_0 \sum_{\substack{r_1+r_2+r_3 = r \\ r_1,\,r_2\leq s-2 \\ \max\{ r_1,\,r_2\} \geq 1 }} \|\nab^{(r_1)} (\kn, \kd)\|_{L^\i(\Sigma_t,g)} \|\nab^{(r_2)} (\kn, \kd)\|_{L^\i(\Sigma_t,g)} \|\nab^{(r_3)} \kd\|_{L^2(\Sigma_t,g)} \\
&\: + C_0 \| (\kn, \kd)\|_{L^\i(\Sigma_t,g)} \| (\kn, \kd)\|_{L^\i(\Sigma_t,g)} \|\nab^{(r)} \kd\|_{L^2(\Sigma_t,g)}\\
\leq &\: C_n t^{-1- r_1} t^{-1-r_2} t^{\ve} t^{-r_3} \mathcal E_s^{\f 12}(t) + C_0 t^{-2} t^{-r} \mathcal E_s^{\f 12}(t)\leq (C_0 t^{-r-2} + C_n t^{-r-2+ \ve})\mathcal E_s^{\f 12}(t),
\end{split}
\end{equation}
where we have used the shorthand $\|\nab^{(r_1)} (\kn, \kd)\|_{L^\i(\Sigma_t,g)} = \|\nab^{(r_1)} \kn\|_{L^\i(\Sigma_t,g)} + \|\nab^{(r_1)} \kd\|_{L^\i(\Sigma_t,g)}$, etc.

When $r = s-1$, we have terms as in \eqref{eq:k.wave.kkk-kkk} which can be controlled similarly, but also the following extra term, which we in addition use Sobolev embedding in \eqref{eq:Sobolev.used} to obtain
\begin{equation}
\begin{split}
&\:  \|k\star k \star k - \kn \star \kn \star \kn\|_{H^r} \\
\leq &\: C_0 \|\nab^{(r)} (\kn, \kd)\|_{(L^\i + t^{s+\f 52 - \ve} L^2) (\Sigma_t,g)} \| (\kn, \kd)\|_{L^\i(\Sigma_t,g)} \|\kd\|_{(L^2\cap t^{-s-\f 52 + \ve}L^\i)(\Sigma_t,g)} \\
\leq  &\: C_0 \|\nab^{(r)} (\kn, \kd)\|_{(L^\i + t^{s+\f 52 - \ve} L^2) (\Sigma_t,g)} \| (\kn, \kd)\|_{L^\i(\Sigma_t,g)} \sum_{r'=0}^2 t^{r'} \|\kd\|_{\dot H^{r'}(\Sigma_t,g)} \\
\leq &\: C_n t^{-r-2+\ve} \mathcal E_s^{\f 12}(t). 
\end{split}
\end{equation}

We now turn to $\rd_t k \star k - \rd_t \kn \star \kn$. For $0\leq r <s-2$, we use H\"older's inequality, Proposition~\ref{prop:k.general}, and the definition of $\mathcal E_s(t)$ to obtain
\begin{equation}\label{eq:k.wave.kdtk.1}
\begin{split}
&\: \|\rd_t k\star k  - \rd_t \kn \star \kn\|_{H^r} \\
\leq &\: C_0 \sum_{\substack{r_1+r_2 = r \\ 1\leq r_1 \leq s-3}} \|\nab^{(r_1)} (\kn, \kd)\|_{L^\i(\Sigma_t,g)} \|\nab^{(r_2)} \rd_t \kd\|_{L^2(\Sigma_t,g)} \\
&\: + C_0 \sum_{\substack{r_1+r_2 = r \\ 1\leq r_1 \leq s-3}} \|\nab^{(r_1)} \rd_t (\kn, \kd)\|_{L^\i(\Sigma_t,g)} \|\nab^{(r_2)} \kd\|_{L^2(\Sigma_t,g)} \\
&\: + C_0 \| (\kn, \kd)\|_{L^\i(\Sigma_t,g)}  \|\nab^{(r)} \rd_t \kd\|_{L^2(\Sigma_t,g)} + C_0 \| \rd_t (\kn, \kd)\|_{L^\i(\Sigma_t,g)}  \|\nab^{(r)} \kd\|_{L^2(\Sigma_t,g)}\\
\leq &\: C_n \sum_{r_1+r_2=r} t^{-1- r_1+\ve}  t^{-r_3-1} \mathcal E_s^{\f 12}(t) + C_0 t^{-1} t^{-r-1} \mathcal E_s^{\f 12}(t)\leq (C_0 t^{-r-2} + C_n t^{-r-2+ \ve})\mathcal E_s^{\f 12}(t).
\end{split}
\end{equation}

For $r = s-2$, we have an additional term when all derivatives hit on $\rd_t (\kn, \kd)$ so that we cannot put it in $L^\i$. For this term we use Proposition~\ref{prop:k.general} and \eqref{eq:Sobolev.used} to obtain
\begin{equation}\label{eq:k.wave.kdtk.2}
\begin{split}
&\: \|\nab^{(r)} \rd_t (\kn, \kd)\|_{(L^\i + t^{s+\f 52 - \ve} L^2) (\Sigma_t,g)} \|\kd\|_{(L^2\cap t^{-s-\f 52 + \ve}L^\i)(\Sigma_t,g)} \\
\leq &\: C_n t^{-r-2+\ve} \sum_{r'=0}^2 t^{r'} \| \kd\|_{H^{r'}(\Sigma_t,g)} \leq C_n t^{-r-2+\ve} \mathcal E_s^{\f 12}(t).
\end{split}
\end{equation}

For $r = s-1$, we have additionally (compared to \eqref{eq:k.wave.kdtk.1} and \eqref{eq:k.wave.kdtk.2}) terms where (1) all but one derivatives hit on $\rd_t (\kn, \kd)$, (2) all derivatives hit on $(\kn,\kd)$, both of which cannot be put into $L^\i$. For these terms we use Proposition~\ref{prop:k.general} and \eqref{eq:Sobolev.used} to get
\begin{equation}
\begin{split}
&\:  \|\nab^{(r -1)} \rd_t (\kn, \kd)\|_{(L^\i + t^{s+\f 52 - \ve} L^2) (\Sigma_t,g)} \|\nab \kd\|_{(L^2\cap t^{-s-\f 52 + \ve}L^\i)(\Sigma_t,g)} \\
&\: + \|\nab^{(r)} (\kn, \kd)\|_{(L^\i + t^{s+\f 52 - \ve} L^2) (\Sigma_t,g)} \|\rd_t \kd\|_{(L^2\cap t^{-s-\f 52 + \ve}L^\i)(\Sigma_t,g)} \\
\leq &\: C_n t^{-r-1+\ve} (\sum_{r'=1}^3 t^{r'-1} \| \kd\|_{H^{r'}(\Sigma_t,g)} + \sum_{r'=0}^2 t^{r'} \|\rd_t \kd\|_{H^{r'}(\Sigma_t,g)}) \leq C_n t^{-r-2+\ve} \mathcal E_s^{\f 12}(t).
\end{split}
\end{equation}

\pfstep{Step~4: Bounding the commutator terms $[-\rd_t^2 + \Delta_g, \, \nab^{(r)}]$}

By (repeated applications of) Proposition~\ref{prop:commutation.formula}, $[-\rd_t^2, \, \nab^{(r)}] k^{(d)}$ consists exactly of terms of the form 
$\sum_{r_1+r_2 = r} \nab^{(r_1)} k\nab^{(r_2)} \rd_t \kd$, $\sum_{r_1+r_2 = r} \nab^{(r_1)} \rd_t k\nab^{(r_2)} \kd$ and $\sum_{r_1+r_2+r_3 = r} \nab^{(r_1)} k \nab^{(r_2)} k\nab^{(r_3)} \rd_t \kd$. Thus they can be controlled in exactly the same manner as in Step~3 to obtain
\begin{equation}
\begin{split}
&\: \|[\rd_t^2, \nab_{i_1} \cdots \nab_{i_r}] k^{(d)} \|_{L^2(\Sigma,g)} \leq (C_0 t^{-r-2} + C_n t^{-r-2+ \ve})\mathcal E_s^{\f 12}(t).
\end{split}
\end{equation}

On the other hand, the commutator $[\Delta_g, \,\nab^{(r)}]$ gives rise to curvature terms. In the $3$-dimensional $\Sigma_t$, the Riemann curvature tensor can be expressed in terms of the Ricci curvature and thus can be controlled using Proposition~\ref{prop:Ric.general} to obtain
\begin{equation}
\begin{split}
&\: \|[\Delta_g, \nab_{i_1} \cdots \nab_{i_r}] k^{(d)} \|_{L^2(\Sigma,g)} \leq C_n t^{-r-2+ \ve}\mathcal E_s^{\f 12}(t).
\end{split}
\end{equation}

\pfstep{Step~5: Putting everything together} Combining Steps~1--4, we have achieved \eqref{eq:main.k.inhomo}. 

Therefore, for every $0\leq r \leq s-1$, we apply Proposition~\ref{prop:EE.wave} with $M = 2N + 2s -2r -2$ to get 
\begin{equation*}
\begin{split}
&\: \f{\ud}{\ud t} [t^{-2N-2s+2r+2}( \|\rd_t \nab^{(r)} k^{(d)}\|^2_{L^2(\Sigma_t,g)} + \|k^{(d)}\|^2_{\dot H^{r+1}(\Sigma_t,g)} + t^{-2} \|k^{(d)}\|^2_{\dot H^{r}(\Sigma_t,g)})]\\
&\: + \f{2N+2s-2r -2}{t} [t^{-2N-2s+2r+2}( \|\rd_t \nab^{(r)} k^{(d)}\|^2_{L^2(\Sigma_t,g)} + \|k^{(d)}\|^2_{\dot H^{r+1}(\Sigma_t,g)} + t^{-2} \|k^{(d)}\|^2_{\dot H^{r}(\Sigma_t,g)})] \\
\leq &\: \f{C_0}{t} [t^{-2N-2s+2r+2}( \|\rd_t \nab^{(r)} k^{(d)}\|^2_{L^2(\Sigma_t,g)} + \|k^{(d)}\|^2_{\dot H^{r+1}(\Sigma_t,g)} + (C_0 + C_nt^{\ve}) t^{-2} \|k^{(d)}\|^2_{\dot H^{r}(\Sigma_t,g)})] \\
&\: +t^{-2N-2s+3} (t^{2r} \|-\rd_t^2 \nab^{(r)}_{i_1\cdots i_r} (k^{(d)})_i{ }^j + \Delta_g \nab^{(r)}_{i_1\cdots i_r} (k^{(d)})_i{ }^j\|_{L^2(\Sigma_t,g)}^2) \\ 
\leq &\:  (C_0 t^{-1} + C_n t^{-1+2\ve}) t^{-2N-2s} \mathcal E_s(t) + C_n t^{3},
\end{split}
\end{equation*}
where in the last line we have controlled $[t^{-2N-2s+2r+2}( \|\rd_t \nab^{(r)} k^{(d)}\|^2_{L^2(\Sigma_t,g)} + \|k^{(d)}\|^2_{\dot H^{r+1}(\Sigma_t,g)} + (C_0 + C_nt^{\ve}) t^{-2} \|k^{(d)}\|^2_{\dot H^{r}(\Sigma_t,g)})]$ using the energy and have used \eqref{eq:main.k.inhomo}.

Summing over $\sum_{r=0}^{s-1}$, we obtain the desired estimate. \qedhere
\end{proof}

\subsection{Transport estimates}\label{subsec:mainest.other}
In this subsection we continue to work under the assumptions of Theorem~\ref{thm:bootstrap}. In particular, we assume the validity of the bootstrap assumptions \eqref{eq:BA1}--\eqref{eq:BA4}.

We prove in this subsection estimates for $h^{(d)},g_{ij}^{(d)},((g^{-1})^{(d)})^{ij}$, which are all derived using the transport equations they obey. 

We insert \eqref{adkdhd.1} and \eqref{adkdhd.2} into \eqref{eq:hyp.sys} to obtain evolution equations for the differences $g^{(d)}_{ij},(k^{(d)})_i{}^j,h^{(d)}$, and $((g^{-1})^{(d)})^{ij}$:
\begin{align}
\label{eq:h.diff}\rd_t h^{(d)} =&\:  {2 (k^{[{\bf n}]})_j{}^i(k^{(d)})_i{}^j + (k^{(d)})_i{ }^j (\kd)_j{ }^i} +  I_{h^{[{\bf n}]}}\\
\label{eq:g.diff}\partial_tg_{ij}^{(d)}=&-2(k^{[\bf n]})_{(i}{}^\ell g^{(d)}_{j)\ell}-2(k^{(d)})_{(i}{}^\ell g_{j)\ell},\\
\label{eq:g-1.diff}\partial_t((g^{-1})^{(d)})^{ij}=&\,2(k^{[\bf n]})_\ell{}^{(j} ((g^{-1})^{(d)})^{i)l}+2(k^{(d)})_\ell{}^{(j} (g^{-1})^{i)l}.
\end{align}

We begin with the more straightforward, less than top-order, estimates for $h^{(d)},g_{ij}^{(d)},((g^{-1})^{(d)})^{ij}$. Commuting the equations \eqref{eq:h.diff}, \eqref{eq:g.diff} and \eqref{eq:g-1.diff} with $\nabla^r$, $r\leq s$, we obtain:
\begin{align}
\label{eq:hd.high}
\partial_t\nabla^{(r)}h^{(d)}=&\,2 \nabla^{(r)}[(k^{[{\bf n}]})_j{}^i(k^{(d)})_i{}^j] + \nabla^{(r)} (k^{(d)})_i{ }^j (\kd)_j{ }^i +\nabla^{(r)} I_{h^{[{\bf n}]}}
 + [\rd_t, \nab^{(r)}] h^{(d)}, \\
\label{eq:gd.high}\partial_t\nabla^{(r)} g_{ij}^{(d)}=&- 2\nabla^{(r)}[(k^{[\bf n]})_{(i}{}^\ell g^{(d)}_{j)\ell}] - 2 g_{\ell (j}\nabla^{(r)}(k^{(d)})_{i)}{}^\ell + [\rd_t, \nab^{(r)}] g_{ij}^{(d)}, \\
\label{eq:gdinv.high}\partial_t\nabla^{(r)} ((g^{-1})^{(d)})^{ij}=&\,2\nabla^{(r)} [(k^{[\bf n]})_\ell{}^{(j} ((g^{-1})^{(d)})^{i)\ell}]+2(g^{-1})^{\ell (i|}\nabla^{(r)} (k^{(d)})_\ell{}^{|j)} + [\rd_t, \nab^{(r)}] ((g^{-1})^{(d)})^{ij}.
\end{align}
We use \eqref{eq:hd.high}--\eqref{eq:gdinv.high} to obtain the following estimates.
\begin{proposition}\label{prop:transport.est.lower}
Given $N\in \mathbb N$, let $n\in \mathbb N$ be sufficiently large so that the estimates in Proposition~\ref{prop:inhomo} hold. Then
\begin{equation}\label{eq:transport.est.main.h}
\begin{split}
 &\: \f{\ud}{\ud t} [t^{-2N-2s} (\sum_{r=0}^{s} t^{2r} \|h^{(d)}\|^2_{\dot H^r(\Sigma_t,g)} )]
 + \f{2N}{t} [t^{-2N-2s}(\sum_{r=0}^{s} t^{2r} \|h^{(d)}\|^2_{\dot H^r(\Sigma_t,g)} )] \\
\leq &\: (C_0 t^{-1} + C_n t^{-1+\ve} ) t^{-2N - 2s}\mathcal E_s(t) + C_n  {t},
\end{split}
\end{equation}
and
\begin{equation}\label{eq:transport.est.main.g}
\begin{split}
 &\: \f{\ud}{\ud t} [t^{-2N-2s} (\sum_{r=0}^{s} t^{2r-2} (\|g^{(d)}\|^2_{\dot H^r(\Sigma_t,g)} + \|(g^{-1})^{(d)}\|^2_{\dot H^r(\Sigma_t,g)}) )] \\
&\: + \f{2N}{t} [t^{-2N-2s}(\sum_{r=0}^{s} t^{2r-2} ( \|g^{(d)}\|^2_{\dot H^r(\Sigma_t,g)} + \|(g^{-1})^{(d)}\|^2_{\dot H^r(\Sigma_t,g)}) )] \\
\leq &\: (C_0 t^{-1} + C_n t^{-1+\ve} ) t^{-2N - 2s}\mathcal E_s(t).
\end{split}
\end{equation}
\end{proposition}
\begin{proof}
We will only prove \eqref{eq:transport.est.main.h}; the bound \eqref{eq:transport.est.main.g} can be derived similarly (and is slightly simpler).

Applying Proposition~\ref{prop:transport.gen} for $\mathcal T = \nab^{(r)} h^{(d)}$ ($0\leq r \leq s$) and $M = 2N+2s-2r$, it suffices to show that
\begin{equation}\label{eq:transport.est.h.error.goal}
t^r\| \rd_t  \nab^{(r)} h^{(d)}\|_{L^2(\Sigma_t,g)} \leq (C_0t^{-1} + C_n t^{-1+\ve}) \mathcal E_s^{\f 12}(t) + t^{N+s}.
\end{equation}

To prove this we consider each term on the RHS of \eqref{eq:hd.high}. First, by  H\"older's inequality, Proposition~\ref{prop:k.general} and \eqref{eq:Sobolev.used}, we obtain
\begin{equation}\label{eq:transport.est.h.error.1}
\begin{split}
&\: \| (k^{[{\bf n}]})_j{}^i(k^{(d)})_i{}^j \|_{\dot H^r(\Sigma_t,g)} + \| (k^{(d)})_i{ }^j (\kd)_j{ }^i \|_{\dot H^r(\Sigma_t,g)} \\
\leq &\: C_0\sum_{ \substack{ r_1+r_2 = r \\ r_1 \leq s -2}} \| \nab^{(r_1)} (\kn, \kd) \|_{L^\i(\Sigma_t,g)} \|\nab^{(r_2)} \kd\|_{L^2(\Sigma_t,g)} \\
&\: + C_n \sum_{ \substack{ r_1+r_2 = r \\ r_1 > s -2}} \| \nab^{(r_1)} (\kn,\kd) \|_{(L^\i + t^{s+\f 52-\ve} L^2)(\Sigma_t,g)} (\sum_{r' = r_2}^{r_2+2} t^{r'-r_2} \|\nab^{(r')}  \kd \|_{L^2(\Sigma_t)}) \\
\leq &\: \sum_{r_1+r_2 = r} (C_0 t^{-r_1-1} + C_n t^{-r_1-1+\ve}) \|\kd\|_{\dot{H}^{r_2}(\Sigma_t,g)} 
\leq (C_0 t^{-1-r} + C_n t^{-1-r+\ve} ) \mathcal E_s^{\f 12}(t).
\end{split}
\end{equation} 

Next, the inhomogeneous term $I_{h^{\bf [n]}}$ can be bounded using Proposition~\ref{prop:inhomo} by
\begin{equation}\label{eq:transport.est.h.error.2}
t^r\| \nab^{(r)} I_{h^{\bf [n]}}\|_{L^2(\Sigma_t,g)} \leq C_n t^{N+s}.
\end{equation}

Finally, by Proposition~\ref{prop:dt.comm.est},
\begin{equation}\label{eq:transport.est.h.error.3}
\| [\rd_t ,\nab^{(r)}] h^{(d)}\|_{L^2(\Sigma_t,g)} \leq C_n \sum_{r'=0}^{r-1} t^{-2-r'+\ve} \|h^{(d)} \|_{H^{r-r'-1}(\Sigma_t,g)} \leq C_n t^{-1-r+\ve} \mathcal E_s^{\f 12}(t).
\end{equation}

Combining \eqref{eq:transport.est.h.error.1}--\eqref{eq:transport.est.h.error.3} yields \eqref{eq:transport.est.h.error.goal}.
\end{proof}

We next turn to the top order derivative estimates for $h^{(d)}$, $g^{(d)}$ and $(g^{-1})^{(d)}$. For this we first control the renormalized top-order quantities introduced in \eqref{def:hd.ren}--\eqref{def:g-1d.ren}. (Subsequently we will show using elliptic estimates that the renormalized top-order quantities indeed control all top-order derivatives; see already Lemma~\ref{lem:remove.modified.2}.) 
\begin{proposition}\label{prop:transport.est.main.top}
Given $N\in \mathbb N$, let $n\in \mathbb N$ be sufficiently large so that the estimates in Proposition~\ref{prop:inhomo} hold. Then
\begin{equation}\label{eq:transport.est.main.h.top}
\begin{split}
 &\: \f{\ud}{\ud t} [t^{-2N-2s+2(s+1)} \|\widetilde{\nab^{(s+1)}_{ren} h^{(d)}}\|^2_{L^2(\Sigma_t,g)} ]
 + \f{2N}{t} [t^{-2N-2s+2(s+1)} \|\widetilde{\nab^{(s+1)}_{ren} h^{(d)}}\|^2_{L^2(\Sigma_t,g)} ] \\
\leq &\: (C_0 t^{-1} + C_n t^{-1+\ve} ) t^{-2N - 2s}\mathcal E_s(t) + C_n  {t},
\end{split}
\end{equation}
and 
\begin{equation}\label{eq:transport.est.main.g.top}
\begin{split}
 &\: \f{\ud}{\ud t} [t^{-2N-2s+2(s+1)-2} (\|\widetilde{\nab^{(s+1)}_{ren} g^{(d)}} \|^2_{L^2(\Sigma_t,g)} + \|\widetilde{\nab^{(s+1)}_{ren} (g^{-1})^{(d)}} \|^2_{L^2(\Sigma_t,g)} )] \\
&\: + \f{2N}{t} [t^{-2N-2s+2(s+1)-2}(\|\widetilde{\nab^{(s+1)}_{ren} g^{(d)}} \|^2_{L^2(\Sigma_t,g)} + \|\widetilde{\nab^{(s+1)}_{ren} (g^{-1})^{(d)}} \|^2_{L^2(\Sigma_t,g)} )]  \\
\leq &\: (C_0 t^{-1} + C_n t^{-1+\ve} ) t^{-2N - 2s}\mathcal E_s(t).
\end{split}
\end{equation}
\end{proposition}
\begin{proof}\pfstep{Step~1: Proof of \eqref{eq:transport.est.main.h.top}}
The main difference with the estimates in Proposition~\ref{prop:transport.est.lower} is that there can potentially be $(s+1)$ derivatives of $k^{(d)}$, which is not controlled by our energy $\mathcal E_s(t)$. The quantity $\widetilde{\nab^{(s+1)}_{ren} h^{(d)}}$ is in fact designed exactly to avoid such terms after using the bounds for the wave equation for $k^{(d)}$. 

We begin our computations. First,
\begin{equation}\label{eq:h.mod.compute.1}
\begin{split}
\rd_t \Delta_g \nab^{(s-1)}_{i_1\cdots i_{s-1}} h^{(d)} = 2 (\kn + \kd)_i{ }^j \Delta_g \nab^{(s-1)}_{i_1\cdots i_{s-1}} (k^{(d)})_j{ }^i + \mathrm{error},
\end{split}
\end{equation}
where the error terms have at most $s$ derivatives hitting on $k^{(d)}$ and thus satisfy the estimates similar to that in the proof of Proposition~\ref{prop:transport.est.lower} (and their proofs are therefore omitted):
\begin{equation}\label{eq:h.mod.compute.1.5}
t^{s+1}\| \mathrm{error} \|_{L^2(\Sigma_t,g)} \leq  (C_0t^{-1} + C_n t^{-1+\ve}) \mathcal E_s^{\f 12}(t) + t^{N+s}.
\end{equation}

The term $2 (\kn + \kd)_i{ }^j \Delta_g \nab^{(s-1)}_{i_1\cdots i_{s-1}} (k^{(d)})_j{ }^i$, however, cannot be controlled. Nevertheless, continuing our computations, we see that
\begin{equation}\label{eq:h.mod.compute.2}
\begin{split}
&\: \rd_t ((\kn + \kd)_i{ }^j  \rd_t \nab^{(s-1)}_{i_1\cdots i_{s-1}} (k^{(d)})_j{ }^i) \\
= &\: (\kn + \kd)_i{ }^j \Delta_g\nab^{(s-1)}_{i_1\cdots i_{s-1}} (k^{(d)})_j{ }^i + (\kn + \kd)_i{ }^j (\rd_t^2 -\Delta_g)(k^{(d)})_j{ }^i \\
&\: +  \{ \rd_t(\kn + \kd)_i{ }^j\} \{\rd_t \nab^{(s-1)}_{i_1\cdots i_{s-1}} (k^{(d)})_j{ }^i\}.
\end{split}
\end{equation}

Note that this generates a term $(\kn + \kd)_i{ }^j \Delta_g\nab^{(s-1)}_{i_1\cdots i_{s-1}} (k^{(d)})_j{ }^i$ which can be used to cancel the uncontrollable term in \eqref{eq:h.mod.compute.1}. Hence, combining \eqref{eq:h.mod.compute.1}, \eqref{eq:h.mod.compute.1.5} and \eqref{eq:h.mod.compute.2}, we obtain
\begin{equation}\label{eq:h.mod.compute.3}
\begin{split}
&\: \|\rd_t (\Delta_g \nab^{(s-1)}_{i_1\cdots i_{s-1}} h^{(d)} - 2(\kn + \kd)_i{ }^j  \rd_t \nab^{(s-1)}_{i_1\cdots i_{s-1}} (k^{(d)})_j{ }^i ) \|_{L^2(\Sigma_t,g)}\\
\leq &\: 2\|(\kn + \kd)_i{ }^j (\rd_t^2 -\Delta_g)\nab^{(s-1)}(k^{(d)})_j{ }^i\|_{L^2(\Sigma_t,g)} \\
&\: + 2 \|\{ \rd_t(\kn + \kd)_i{ }^j\} \{\rd_t \nab^{(s-1)}_{i_1\cdots i_{s-1}} (k^{(d)})_j{ }^i\}\|_{L^2(\Sigma_t,g)}  + (C_0 t^{-2-s} + C_n t^{-2-s+\ve}) \mathcal E_s^{\f 12}(t) + t^{N-1}.
\end{split}
\end{equation}

We now handle to two terms in \eqref{eq:h.mod.compute.3}. By Proposition~\ref{prop:k.general}, \eqref{eq:main.k.inhomo} and H\"older's inequality,
\begin{equation}\label{eq:h.mod.compute.4}
\begin{split}
&\: \| (\kn + \kd)_i{ }^j (\rd_t^2 -\Delta_g)\nab^{(s-1)}(k^{(d)})_j{ }^i\|_{L^2(\Sigma_t,g)} \\
\leq &\: (\|(\kn, \kd)\|_{L^\i(\Sigma_t,g)} \| (\rd_t^2 -\Delta_g)\nab^{(s-1)}(k^{(d)})\|_{L^2(\Sigma_t,g)} \leq (C_0t^{-2-s} + C_n t^{-2-s+\ve})\mathcal E_s^{\f 12}(t) + C_n t^{N}.
\end{split}
\end{equation}
On the other hand, by Proposition~\ref{prop:k.general} and H\"older's inequality, 
\begin{equation}\label{eq:h.mod.compute.5}
\|\{ \rd_t(\kn + \kd)_i{ }^j\} \{\rd_t \nab^{(s-1)}_{i_1\cdots i_{s-1}} (k^{(d)})_j{ }^i\}\|_{L^2(\Sigma_t,g)} \leq (C_0t^{-2-s} + C_n t^{-2-s+\ve}) \mathcal E_s^{\f 12}(t).
\end{equation} 

Combining \eqref{eq:h.mod.compute.3}--\eqref{eq:h.mod.compute.5} and noticing that $\widetilde{\nab^{(s+1)}_{ren} h^{(d)}}$ is defined exactly to be (recall \eqref{def:hd.ren}) $\Delta_g \nab^{(s-1)}_{i_1\cdots i_{s-1}} h^{(d)} - 2(\kn + \kd)_i{ }^j  \rd_t \nab^{(s-1)}_{i_1\cdots i_{s-1}} (k^{(d)})_j{ }^i $, we thus obtain
\begin{equation}
\begin{split}
\| \rd_t \widetilde{\nab^{(s+1)}_{ren} h^{(d)}} \|_{L^2(\Sigma_t,g)} \leq (C_0t^{-2-s} + C_n t^{-2-s+\ve})\mathcal E_s^{\f 12}(t) + C_n t^{N-1}.
\end{split}
\end{equation}

The desired estimate \eqref{eq:transport.est.main.h.top} then follows directly from Proposition~\ref{prop:transport.gen} (for $M=2N+2s-2(s+1)$).

\pfstep{Step~2: Proof of \eqref{eq:transport.est.main.g.top}} The main idea is similar to Step~1, so we will be brief. The main difference is that for $g^{(d)}$, not only the derivatives of the inhomogeneous terms create $\nab^{(s+1)} k$, but the commutator terms also create similar terms, which have to be taken care of by a renormalization. More precisely, by \eqref{eq:g.diff} and Proposition~\ref{prop:commutation.formula},
\begin{equation}\label{eq:gd.top}
\begin{split}
&\: \partial_t\Delta_g\nabla^{(s-1)}_{i_1\cdots i_{s-2} a} g_{ij}^{(d)}\\
=&\: \Delta_g\nabla^{(s-2)}_{i_1\cdots i_{s-2}}\nabla_a\partial_tg_{ij}^{(d)} + [\rd_t, \Delta_g\nabla^{(s-1)}_{i_1\cdots i_{s-2} a}] g_{ij}^{(d)} \\
= &\: - 2 g_{\ell (j} \Delta_g \nabla^{(s-1)}_{i_1\cdots i_{s-2} a}(k^{(d)})_{i)}{}^\ell \\
&\: -\Delta_g\nabla^{(s-2)}_{i_1\cdots i_{s-2}} ((g^{-1})^{be} g_{m(i|}\nab_e k_{a)}{ }^m - \nabla_{(a} k_{i)}{ }^b - (g^{-1})^{be} g_{ {d}(a} \nab_{i)} k_e{}^d) g^{(d)}_{bj}\\
&\:-\Delta_g\nabla^{(s-2)}_{i_1\cdots i_{s-2}} ((g^{-1})^{be} g_{m(j|}\nab_e k_{a)}{ }^m - \nabla_{(a} k_{j)}{ }^b - (g^{-1})^{be} g_{ {d}(a} \nab_{j)} k_e{}^d) g^{(d)}_{ib} + \ldots,
\end{split}
\end{equation}
where the terms denotes by $\dots$ have at most $(s+1)$ derivatives on $\kn$, at most $s$ derivatives on $k$ and at most $(s+1)$ derivatives on $g^{(d)}$, and therefore can be bounded as in Proposition~\ref{prop:transport.est.lower} by 
$$\| \ldots\|_{L^2(\Sigma_t,g)} \leq (C_0 t^{-1-s} + C_n t^{-1-s+\ve}) \mathcal E_s^{\f 12}(t).$$

It thus remains to handle all the main terms appearing on the RHS of \eqref{eq:gd.top}. Now one observes that the quantity  $\widetilde{\nab^{(s+1)}_{ren} g^{(d)}}$ is designed exactly to remove this term (in a similar way as $\widetilde{\nab^{(s+1)}_{ren} h^{(d)}}$ is designed in Step~1) so that the additional error terms are controllable. It thus follows that 
\begin{equation}
\begin{split}
\| \rd_t  \widetilde{\nab^{(s+1)}_{ren} g^{(d)}}\|_{L^2(\Sigma_t,g)} \leq (C_0 t^{-1-s} + C_n t^{-1-s+\ve}) \mathcal E_s^{\f 12}(t),
\end{split}
\end{equation}
which implies the desired estimate for $\widetilde{\nab^{(s+1)}_{ren} g^{(d)}}$ in \eqref{eq:transport.est.main.g.top} after using Proposition~\ref{prop:transport.gen}..

The argument for $\widetilde{\nab^{(s+1)}_{ren} (g^{-1})^{(d)}}$ is similar and omitted. \qedhere
\end{proof}

We conclude this subsection by summarizing what we have achieved so far, namely that we have obtained an estimate for the modified energy by the energy:
\begin{proposition}\label{prop:intermediate.EE}
Given $N\in \mathbb N$, let $n\in \mathbb N$ be sufficiently large so that the estimates in Proposition~\ref{prop:inhomo} hold. Then for any $t\in [T_{\mathrm{aux}}, T_{\mathrm{Boot}})$,
\begin{equation*}
\begin{split}
t^{-2N-2s}\widetilde{\mathcal E_s}(t) + 2N \int_{T_{\mathrm{aux}}}^t \f{\tau^{-2N-2s}\widetilde{\mathcal E_s}(\tau)}{\tau}  \, \ud \tau 
\leq &\: \int_{T_{\mathrm{aux}}}^t (C_0 \tau^{-1} + C_n \tau^{-1 + \ve}) \tau^{-2N-2s}\mathcal E_s(\tau) \, \ud \tau + C_n  {t}. 
\end{split}
\end{equation*}
\end{proposition}
\begin{proof}
This is an immediate consequence of Propositions~\ref{prop:main.wave.est}, \ref{prop:transport.est.lower} and \ref{prop:transport.est.main.top}. \qedhere
\end{proof}

\subsection{Conclusion of the proof of Theorem~\ref{thm:bootstrap}}\label{sec:conclusion.of.bootstrap.thm}

In order to conclude the proof of Theorem~\ref{thm:bootstrap}, we finally need to relate $\mathcal E_s$ and $\widetilde{\mathcal E_s}$ (which will be achieved in Lemmas~\ref{lem:remove.modified.1} and \ref{lem:remove.modified.2}), and then use the energy inequality in Proposition~\ref{prop:intermediate.EE} to deduce our desired estimates.

Recalling now the difference between $\mathcal E_s$ and $\widetilde{\mathcal E_s}$ (as described immediately after their definitions in \eqref{eq:energy.def}--\eqref{def:g-1d.ren}), we need to 
\begin{itemize}
\item relate $\rd_t \nab^{(r)} k^{(d)}$ and $\nab^{(r)}\rd_t k^{(d)}$ (achieved using a commutator estimate; see Lemma~\ref{lem:remove.modified.1}), and
\item relate the renormalized top-order quantities and other top-order derivatives (achieved using elliptic estimates; see Lemma~\ref{lem:remove.modified.2}).
\end{itemize}

\begin{lemma}\label{lem:remove.modified.1}
The following estimate holds:
\begin{align*}
\sum_{r=0}^{s-1} t^{2r+2} \|\rd_t k^{(d)} \|_{\dot H^r(\Sigma_t,g)}^2 \leq (C_0+C_n t^\ve) \widetilde{\mathcal E_s}(t).
\end{align*}
\end{lemma}
\begin{proof}
We control the commutator $[\rd_t, \nab^{(r)}] k^{(d)}$ using Proposition~\ref{prop:dt.comm.est} to obtain
$$\sum_{r=0}^{s-1} t^{2r+2} \|\nab^{(r)}\rd_t k^{(d)}-\rd_t \nab^{(r)} k^{(d)}\|_{L^2(\Sigma_t,g)}^2 \leq C_n t^\ve \widetilde{\mathcal E_s}(t),$$
from which the desired estimate follows from the definition of $\widetilde{\mathcal E_s}$.
\end{proof}

\begin{lemma}\label{lem:elliptic}
Given any tensor $\xi$ tangential to $\Sigma_t$,
\begin{equation}\label{eq:main.elliptic}
\|\nab^{(2)} \xi \|_{L^2(\Sigma_t,g)}^2 \leq 2 \|\Delta_g \xi\|_{L^2(\Sigma_t,g)}^2 + C_n t^{-2+\ve} \|\nab \xi\|_{L^2(\Sigma_t,g)}^2 + C_n t^{-4+2\ve} \|\xi\|_{L^2(\Sigma_t,g)}^2.
\end{equation}
\end{lemma}
\begin{proof}
We compute
\begin{equation}
\begin{split}
&\|\Delta_g \xi \|_{L^2(\Sigma_t,g)}^2 \\
= &\int_{\Sigma_t} \underbrace{(g^{-1})^{a_1a_1'} \cdots (g^{-1})^{a_\ell a_\ell'} g_{b_1 b_1'} \cdots g_{b_m b_m'} (g^{-1})^{ii'} (g^{-1})^{jj'}}_{=:\mathfrak G_{b_1 \cdots b_m b_1' \cdots b_m' }^{a_1\cdots a_\ell a_1' \cdots a_\ell' i i' j j'}}\nabla_i\nabla_{i'}  \xi_{a_1 \cdots a_\ell}^{b_1 \cdots b_m} \nabla_{j'}\nabla_{j} \xi_{a_1' \cdots a_\ell'}^{b_1' \cdots b_m'} \, \mathrm{vol}_{\Sigma_t}\\
=& -\int_{\Sigma_t} \mathfrak G_{b_1 \cdots b_m b_1' \cdots b_m' }^{a_1\cdots a_\ell a_1' \cdots a_\ell' i i' j j'} \nabla_{i'}  \xi_{a_1 \cdots a_\ell}^{b_1 \cdots b_m} \nab_i \nabla_{j'}\nabla_{j} \xi_{a_1' \cdots a_\ell'}^{b_1' \cdots b_m'} \, \mathrm{vol}_{\Sigma_t}\\
=&\int_{\Sigma_t} \mathfrak G_{b_1 \cdots b_m b_1' \cdots b_m' }^{a_1\cdots a_\ell a_1' \cdots a_\ell' i i' j j'} \nabla_{j'} \nabla_{i'}  \xi_{a_1 \cdots a_\ell}^{b_1 \cdots b_m} \nabla_i\nabla_{j} \xi_{a_1' \cdots a_\ell'}^{b_1' \cdots b_m'} \, \mathrm{vol}_{\Sigma_t} + \mathrm{error}
= \|\nab^{(2)} \xi\|_{L^2(\Sigma_t,g)}^2 + \mathrm{error},
\end{split}
\end{equation}
where terms labelled $\mathrm{error}$ (different in the two instances) come from commuting covariant derivatives and obey an estimate
$$|\mathrm{error}| \leq C_0 \|Riem(g)\|_{L^\i(\Sigma,g)} \|\nab \xi\|_{L^2(\Sigma_t,g)}^2 +C_0   \|Riem(g)\|_{L^\i(\Sigma,g)} \|\nab^{(2)} \xi\|_{L^2(\Sigma_t,g)} \|\xi\|_{L^2(\Sigma_t,g)}.$$

As a consequence, since on the $3$-dimensional $\Sigma_t$, $Riem(g)$ can be expressed in terms of $Ric(g)$, we can use H\"older's inequality and Proposition~\ref{prop:Ric.general} to obtain
\begin{equation}
\begin{split}
&\: \|\nab^{(2)} \xi \|_{L^2(\Sigma_t,g)}^2 \\
\leq &\: \|\Delta_g \xi \|_{L^2(\Sigma_t,g)}^2 + C_0 \|Riem(g)\|_{L^\i(\Sigma,g)} \|\nab \xi\|_{L^2(\Sigma_t,g)}^2\\
&\: +C_0   \|Riem(g)\|_{L^\i(\Sigma,g)} \|\nab^{(2)} \xi\|_{L^2(\Sigma_t,g)} \|\xi\|_{L^2(\Sigma_t,g)} \\
\leq &\: \|\Delta_g \xi \|_{L^2(\Sigma_t,g)}^2 + C_n t^{-2+\ve} \|\nab \xi\|_{L^2(\Sigma_t,g)}^2 +C_n  t^{-2+\ve} \|\nab^{(2)} \xi\|_{L^2(\Sigma_t,g)} \|\xi\|_{L^2(\Sigma_t,g)},
\end{split}
\end{equation}
which implies \eqref{eq:main.elliptic} after using Young's inequality and absorbing $ \f 12\|\nab^{(2)} \xi \|_{L^2(\Sigma_t,g)}^2$ to the LHS. \qedhere
\end{proof}

\begin{lemma}\label{lem:remove.modified.2}
The top order part of the energy for $h^{(d)},g^{(d)},(g^{-1})^{(d)}$ is bounded by:
$$t^{2(s+1)}\| h^{(d)} \|_{\dot H^{s+1}(\Sigma_t, g)}^2 + t^{2s} (\| g^{(d)} \|_{\dot H^{s+1}(\Sigma_t, g)}^2 + \| (g^{-1})^{(d)} \|_{\dot H^{s+1}(\Sigma_t, g)}^2 )  \leq (C_0 + C_n t^\ve) \widetilde{\mathcal E_s}(t).$$
\end{lemma}
\begin{proof}
The key is to use Lemma~\ref{lem:elliptic}. Consider for instance $h^{(d)}$. We first note that $\Delta_g \nab^{(s-1)} h^{(d)}$ can be written as a linear combination of the renormalized top-order quantity $\widetilde{\nab^{(s+1)}_{ren} h^{(d)}}$ and terms which has at most $s$ derivatives of $k^{(d)}$ (and $\kn$) so that it can be checked that
$$t^{2(s+1)} \|\Delta_g \nab^{(s-1)} h^{(d)}\|_{L^2(\Sigma_t,g)}^2 \leq (C_0 + C_n t^\ve) \widetilde{\mathcal E_s}(t).$$
It then follows by the elliptic estimates in Lemma~\ref{lem:elliptic} and the lower order control for $h^{(d)}$ by $\widetilde{\mathcal E_s}(t)$ that
$$t^{2(s+1)}\| h^{(d)} \|_{\dot H^{s+1}(\Sigma_t, g)}^2  \leq (C_0 + C_n t^\ve) \widetilde{\mathcal E_s}(t).$$
The estimates for the top-order derivatives for $g^{(d)}$ and $(g^{-1})^{(d)}$ are similar. \qedhere
\end{proof}

Combining Lemmas~\ref{lem:remove.modified.1} and \ref{lem:remove.modified.2}, we obtain
\begin{proposition}\label{prop:remove.modified}
Given $N\in \mathbb N$, let $n\in \mathbb N$ be sufficiently large so that the estimates in Proposition~\ref{prop:inhomo} hold. Then for any $t\in [T_{\mathrm{aux}}, T_{\mathrm{Boot}})$,
$$\mathcal E_s(t) \leq (C_0 + C_n t^{\ve}) \widetilde{\mathcal E_s}(t).$$
\end{proposition}

We are now ready to conclude the proof of the bootstrap theorem (Theorem~\ref{thm:bootstrap}):
\begin{proof}[Proof of Theorem \ref{thm:bootstrap}]
Given any $N\in \mathbb N$, choose $n\in \mathbb N$ sufficiently large so that the estimates in Proposition~\ref{prop:inhomo} hold.

Combining Propositions~\ref{prop:intermediate.EE} and \ref{prop:remove.modified}, and integrating in $t$ (noting that we have trivial data at $T_{aux}$), we obtain that
\begin{equation}\label{eq:final.energy}
\begin{split}
\f{t^{-2N-2s}}{(C_0 + C_n t^{\ve})} \mathcal E_s(t) + 2N \int_{T_{\mathrm{aux}}}^t \f{\tau^{-2N-2s}\mathcal E_s(\tau)}{(C_0 + C_n \tau^{\ve}) \tau}  \, \ud \tau 
\leq &\: \int_{T_{\mathrm{aux}}}^t (C_0 \tau^{-1} + C_n \tau^{-1 + \ve}) \tau^{-2N-2s}\mathcal E_s(\tau) \, \ud \tau + C_n t. 
\end{split}
\end{equation}

We now choose our constants. First choose $N$ sufficiently large so that 
$$N \geq \max\{ 2 C_0(C_0+1),\,2(C_0+1),\,N_0,\,7\}.$$ We then fix an $n_{N_0,s}\in \mathbb N$ sufficiently large so that whenever $n \geq n_{N_0,s}$, \eqref{eq:final.energy} holds with the given $N$. After fixing $n$, we then choose $T_{N_0,s,n}$ so that $C_n T_{N_0,s,n}^\ve \leq 1$. Plugging $C_0 \leq \f{N}{2(C_0+1)}$ and $C_n T_{N_0,s,n}^\ve\leq 1$ into \eqref{eq:final.energy}, we then obtain
\begin{equation}\label{eq:final.energy.2}
\begin{split}
\f{t^{-2N-2s}}{(C_0 + 1)} \mathcal E_s(t) + 2N \int_{T_{\mathrm{aux}}}^t \f{\tau^{-2N-2s}\mathcal E_s(\tau)}{(C_0 + 1) \tau}  \, \ud \tau 
\leq &\: \int_{T_{\mathrm{aux}}}^t (\f{N}{2(C_0+1)} + 1) \f{ \tau^{-2N-2s}\mathcal E_s(\tau)}{\tau} \, \ud \tau + C_n t. 
\end{split}
\end{equation}
Notice that we have chosen $N$ so that  $(\f{N}{2(C_0+1)} + 1) \leq \f{N}{C_0+1}$. We can thus subtract $N \int_{T_{\mathrm{aux}}}^t \f{\tau^{-2N-2s}\mathcal E_s(\tau)}{(C_0 + 1) \tau}  \, \ud \tau $ from both sides of \eqref{eq:final.energy.2} to obtain
\begin{equation}\label{eq:final.energy.3}
\begin{split}
\f{t^{-2N-2s}}{(C_0 + 1)} \mathcal E_s(t) + N \int_{T_{\mathrm{aux}}}^t \f{\tau^{-2N-2s}\mathcal E_s(\tau)}{(C_0 + 1) \tau}  \, \ud \tau 
\leq C_n t,
\end{split}
\end{equation}
which immediately implies 
\begin{align}\label{gdkdhdest.no.0}
\mathcal E_s(t) \leq t^{2N+2s},
\end{align}
after choosing $T_{N_0,s,n}$ smaller if necessary. In particular, since we have chosen $N\geq N_0$ and $T_{N_0,s,n}\leq 1$, we obtain \eqref{gdkdhdest}. 

Finally, we check that we have improved the bootstrap assumption. For \eqref{eq:BA2}--\eqref{eq:BA4}, this is immediate from \eqref{gdkdhdest}. For \eqref{eq:BA1}, note that \eqref{gdkdhdest.no.0} and \eqref{eq:Sobolev.2} imply 
$$\|g - g^{\bf [n]}\|_{L^\i(\Sigma_t,g)}  \leq C_0 t^{N+s- {\f 32}}.$$
Now note that the smallest eigenvalue of $g^{-1}$ is $\geq C_0^{-1} t^{-2p_1} \geq C_0^{-1} t^{2}.$
Hence
$$t^8 |a_{ij} - a^{\bf [n]}_{ij}|^2 \leq C_0 t^{4 p_{\max\{i,j\}}} |a_{ij} - a^{\bf [n]}_{ij}|^2 \leq \max_{i,j} |g_{ij} - g^{\bf [n]}_{ij}| ^2 \leq C_0 (t^{-2})^2 \|g - g^{\bf [n]}\|_{L^\i(\Sigma_t,g)}^2 \leq C_0 t^{2N+2s-11}. $$
Now since $N\geq 7$ and $s\geq 4$, we have
$|a_{ij} - a^{\bf [n]}_{ij}| \leq C_0 t^{\f 32}$. Combining with \eqref{eq:main.parametrix.a.bd}, we thus obtain
\begin{equation}\label{eq:a-c.improvements}
|a_{ij} - c_{ij} | \leq C_0 t^{\ve},
\end{equation}
which improves over \eqref{eq:BA1} after taking $T_{N_0,s,n}$ to be sufficiently small. \qedhere
\qedhere
\end{proof}

As we discussed in Section~\ref{sec:step.2.bootstrap}, once we have proven Theorem~\ref{thm:bootstrap}, we now also obtain Corollary~\ref{cor:bootstrap}.

\subsection{Extracting a limit: proof of Proposition~\ref{prop:AA.main}}\label{sec:pf.AA.main}

In this final subsection, we prove Proposition~\ref{prop:AA.main}, which, as indicated in Section~\ref{sec:conclusion.of.actual.wave}, is the final step of the proof of Theorem~\ref{thm:main.reduced}.

We begin with some easy estimates, which will allow us to extract a limit. (Notice that these estimates are allowed to degenerate as $t\to 0$, but importantly they do not depend on $T_{aux}$.)
\begin{lemma}\label{lem:AA.0}
Let $s$, $N_0$, $n$ and $T_{N_0,s,n}$ be as in Theorem~\ref{thm:bootstrap}. For every $T'$, $T''$ satisfying $0< T' < T'' \leq T_{N_0,s,n}$, there exists a constant $C >0$ \underline{independent of $T_{\mathrm{aux}}$} such that the following holds (with definitions in \eqref{adkdhd.1} and \eqref{adkdhd.2}).

Let $T_{aux} \in (0, T']$ and suppose $(g^{\mathrm{aux}}, k^{\mathrm{aux}}, h^{\mathrm{aux}})$ is the solution to \eqref{eq:hyp.sys} on $[T_{\mathrm{aux}},T_{N_0,s,n}) \times \mathbb T^3$ given by Corollary~\ref{cor:bootstrap}. Then 
$$\sup_{t\in [T',T'']} \sup_{x\in \mathbb T^3} \sum_{r+|\alp| \leq 4} (|\rd_t^r \rd_x^{\alp} g^{(d)}_{ij}| + |\rd_t^r \rd_x^{\alp} ((g^{-1})^{(d)})^{ij} |)(t,x) + \sum_{r+|\alp|\leq 3} (|\rd_t^r \rd_x^\alp k^{(d)}_{ij}| + |\rd_t^r \rd_x^\alp h^{(d)}|)(t,x) \leq C.$$
\end{lemma}
\begin{proof}
When there is no $\rd_t$ derivative, this just follows from \eqref{gdkdhdest} and Sobolev embedding. To obtain the estimates with the $\rd_t$ derivatives, we use in addition the equations \eqref{eq:k.diff}, \eqref{eq:h.diff}, \eqref{eq:g.diff} and \eqref{eq:g-1.diff}.
\end{proof}

\begin{lemma}\label{lem:AA}
Let $s$, $N_0$, $n$ and $T_{N_0,s,n}$ be as in Theorem~\ref{thm:bootstrap}. There exists a sequence of auxiliary times $\{ T_{\mathrm{aux},I} \}_{I=1}^{+\infty} \subset (0,T_{N_0,s,n})$, $\lim_{I\to +\infty} T_{\mathrm{aux},I} = 0$ such that the corresponding solutions $\{ (g^{\mathrm{aux}}_I, k^{\mathrm{aux}}_I, h^{\mathrm{aux}}_I \}_{I=1}^{+\infty}$ given by Lemma~\ref{lem:locexist} converge locally in $C^3\times C^2\times C^2$ norm (as $I\to +\infty$) to a limit $(g, k ,h)$ which solves \eqref{eq:hyp.sys} in $(0,T_{N_0,s,n}]\times \mathbb T^3$. Moreover, after denoting $g^{(d)} = g - g^{\bf [n]}$, $(g^{-1})^{(d)} = g^{-1} - (g^{\bf [n]})^{-1}$, $k^{(d)} = k - k^{\bf [n]}$ and $h^{(d)} = h - h^{\bf [n]}$, the estimate \eqref{gdkdhdest} holds. 
\end{lemma}
\begin{proof}
The existence of a limit follows from Lemma~\ref{lem:AA.0}, the Arzela--Ascoli theorem, and a standard argument extracting a diagonal sequence. Since the limit is achieved locally in $C^3\times C^2\times C^2$, it follows that the limit satisfies the system \eqref{eq:hyp.sys}. 

Finally, we prove that the limit obeys the estimate \eqref{gdkdhdest}. First, note that the estimate \eqref{gdkdhdest} implies that for every $t$, there is a subsequence $\{ T_{\mathrm{aux},I_\ell} \}_{\ell=1}^{+\infty}$ for which  $\{ (g^{\mathrm{aux}}_{I_\ell}, k^{\mathrm{aux}}_{I_\ell}, h^{\mathrm{aux}}_{I_\ell} \}_{\ell=1}^{+\infty}$ has a weak limit satisfying \eqref{gdkdhdest}. This limit must coincide with the local $C^3\times C^2\times C^2$ limit, thus showing the bound \eqref{gdkdhdest}.
\end{proof}

The very final statement we need in order to complete the proof of Proposition~\ref{prop:AA.main} is that $g_{jj'} k_i{ }^{j'}$ is symmetric in $i$ and $j$. The key to such a statement is the following lemma.
\begin{lemma}\label{lem:eqn.for.anti.k}
Suppose $(g,k,h)$ solves \eqref{eq:hyp.sys}. Then the term $(g_{jj'} k_i{ }^{j'} - g_{ij'} k_j{ }^{j'})$ satisfies an inhomogeneous wave equation of the following form:
\begin{equation*}
\begin{split}
&\: (\rd_t^2 - \Delta_g)(g_{jj'} k_i{ }^{j'} - g_{ij'} k_j{ }^{j'}) \\
= &\: X^{a_1 b_1 c d}_{a_2 b_2 i j} k_{a_1}{ }^{a_2} k_{b_1}{ }^{b_2} (g_{c\ell} k_{d}{ }^\ell - g_{d\ell} k_{c}{ }^\ell) + Y^{a_1 c d}_{a_2 i j} \rd_t k_{a_1}^{a_2} (g_{c\ell} k_{d}{ }^\ell - g_{d\ell} k_{c}{ }^\ell) + Z^{a_1 c d}_{a_2 i j} k_{a_1}^{a_2} \rd_t (g_{c\ell} k_{d}{ }^\ell - g_{d\ell} k_{c}{ }^\ell),
\end{split}
\end{equation*}
where $X$, $Y$ and $Z$ are some tensor products of $g$, $g^{-1}$ and $\de$.
\end{lemma}
\begin{proof} \pfstep{Step~1: Easy reductions}
First, a direct computation shows that
\begin{equation*}
\begin{split}
&\: \rd_t^2 (g_{jj'} k_i{ }^{j'} - g_{ij'} k_j{ }^{j'}) \\
=&\: -\rd_t \{ (g_{jb} k_{j'}{ }^b - g_{j'b} k_j{ }^b) k_i{ }^{j'} - (g_{ib} k_{j'}{ }^b - g_{j'b} k_i{ }^b) k_j{ }^{j'}\} - (g_{j'b}k_j{ }^b - g_{jb} k_{j'}{ }^b) \rd_t k_i{ }^{j'} \\
&\: + (g_{jj'} \rd^2_t k_i{ }^{j'} - g_{ij'} \rd_t^2 k_j{ }^{j'}) - 2 g_{jj'} k_b{ }^{j'} \rd_t k_i{ }^b .
\end{split}
\end{equation*}
Notice that all the terms on the first line are of the form as required by the lemma.

It thus follows from \eqref{eq:hyp.sys} that
\begin{equation}
\begin{split}
&\: (\rd_t^2 - \Delta_g)(g_{jj'} k_i{ }^{j'} - g_{ij'} k_j{ }^{j'}) \\
=&\: g_{jj'} \{ (k\star k \star k)_i{ }^{j'} + (\rd_t k \star k)_i{ }^{j'}\}  - g_{ij'} \{ (k\star k \star k)_j{ }^{j'} + (\rd_t k \star k)_j{ }^{j'}\} - 2 g_{jj'} k_b{ }^{j'} \rd_t k_i{ }^b + \dots,
\end{split}
\end{equation}
where $\dots$ denotes terms which are of the form as required by the lemma. (Notice in particular that the Hessian of $h$ term drops of because it is symmetric.)

Investigating now the terms in $k\star k \star k$ and $\rd_t k \star k$, we only need to check that
$$Q_i{ }^{j'}g_{jj'} - Q_j{ }^{j'} g_{ij'},\quad \mbox{where }Q_i{ }^{j'} \in \{k_i{ }^{j'},\, \de_i{ }^{j'},\, k_i{ }^a k_a{ }^{j'},\, \rd_t k_i{ }^{j'},\, [(\rd_t k_i{ }^a) k_a{ }^{j'} - (\rd_t k_a{ }^{j'}) k_i{ }^a]\}$$ 
is of the form required by the lemma. (Note that the term $[(\rd_t k_i{ }^a) k_a{ }^{j'} - (\rd_t k_a{ }^{j'}) k_i{ }^a]$ comes from combining terms in $\rd_t k \star k$ and $- 2 g_{jj'} k_b{ }^{j'} \rd_t k_i{ }^b$.)

Now clearly if $Q_i{ }^{j'} \in \{k_i{ }^{j'},\, \de_i{ }^{j'}\}$,  {then} $Q_i{ }^{j'}g_{jj'} - Q_j{ }^{j'} g_{ij'}$ is of the desired form. 

For $Q_i^{j'} = k_i{ }^a k_a{ }^{j'}$, we compute
\begin{equation*}
\begin{split}
 k_i{ }^a k_a{ }^{j'} g_{jj'} - k_j{ }^a k_a{ }^{j'} g_{ij'} 
= &\: k_i{ }^a (k_a{ }^{j'} g_{jj'} - k_j{ }^{j'}g_{aj'}) - k_j{ }^a (k_a{ }^{j'} g_{ij'} - k_i{ }^{j'} g_{aj'}),
\end{split}
\end{equation*}
which is of the desired form.

For $Q_i^{j'} = \rd_t k_i{ }^{j'}$, we compute
\begin{equation*}
\begin{split}
 &\:  g_{jj'} \rd_t k_i{ }^{j'} -  g_{ij'} \rd_t k_j{ }^{j'}\\
= &\: \rd_t (g_{jj'}k_i{ }^{j'} - g_{ij'}k_j{ }^{j'}) + g_{jb} k_{j'}{ }^b k_i{ }^{j'} + g_{j'b} k_j{ }^b k_i{ }^{j'} - g_{ib} k_{j'}{ }^b k_j{ }^{j'} -g_{j'b} k_i{ }^b k_j{ }^{j'} \\
= &\: \rd_t (g_{jj'}k_i{ }^{j'} - g_{ij'}k_j{ }^{j'}) + (g_{jb} k_{j'}{ }^b - g_{j'b} k_j{ }^b) k_i{ }^{j'}  - (g_{ib} k_{j'}{ }^b -g_{j'b} k_i{ }^b) k_j{ }^{j'},
\end{split}
\end{equation*}
which is of the desired form.

For $Q_i^{j'} = [(\rd_t k_i{ }^a) k_a{ }^{j'} - (\rd_t k_a{ }^{j'}) k_i{ }^a]$, we compute
\begin{equation*}
\begin{split}
&\: [(\rd_t k_i{ }^a) k_a{ }^{j'} - (\rd_t k_a{ }^{j'}) k_i{ }^a] g_{jj'} - [(\rd_t k_j{ }^a) k_a{ }^{j'} - (\rd_t k_a{ }^{j'}) k_j{ }^a] g_{ij'} \\
= &\: \rd_t (k_i{ }^a g_{j'a} - k_{j'}{ }^a g_{ia}) k_j{ }^{j'} + \rd_t(k_{j'}{ }^a g_{ja} - k_j{ }^a g_{j' a}) + \rd_t k_i{ }^a (k_a{ }^{j'} g_{jj'} - k_j{ }^{j'} g_{aj'}) - \rd_t k_j{ }^a (k_a{ }^{j'} g_{ij'} - k_i{ }^{j'} g_{aj'}) \\
&\: + (g_{ba} k_{j'}{ }^a - g_{j'a} k_{b}{ }^a) k_j{ }^b k_i{ }^{j'} + (g_{j'a} k_i{ }^{j'} - g_{j'i} k_a{ }^{j'}) k_j{ }^b k_b{ }^a \\
&\: - (g_{ba} k_{j'}{ }^a - g_{j'a} k_{b}{ }^a) k_i{ }^b k_j{ }^{j'} - (g_{j'a} k_j{ }^{j'} - g_{j'j} k_a{ }^{j'}) k_i{ }^b k_b{ }^a, 
\end{split}
\end{equation*}
which is of the desired form. This concludes the proof. \qedhere

\end{proof}

We are now ready to show that $g_{jj'} k_i{ }^{j'}$ is symmetric in $i$ and $j$.
\begin{lemma}\label{lem:k.is.symmetric}
Given a limit $(g, k ,h)$ as in Lemma~\ref{lem:AA}, the limiting $k$ is in fact the second fundamental form, i.e.~$k_{ij}:= g_{jj'} k_i{ }^{j'} = -\f 12 \rd_t g_{ij}$.
\end{lemma}
\begin{proof}
Denoting $k_{ij}:= g_{jj'} k_i{ }^{j'} $, the equation for $g$ implies that $\rd_t g_{ij} = - k_{ij} - k_{ji}$. Hence, in order to prove the lemma, it suffices to show that $k_{ij}$ is symmetric in $i$ and $j$. 

To this end, we define $(k^{\mathrm{aux}}_I)_{ij}:= (g^{\mathrm{aux}}_I)_{jj'} (k^{\mathrm{aux}}_I)_i{ }^{j'}$, and first obtain an estimate for its anti-symmetric part.
 By Lemma~\ref{lem:eqn.for.anti.k}, $(k^{\mathrm{aux}}_I)_{ij} - (k^{\mathrm{aux}}_I)_{ji}$ satisfies a \emph{homogeneous} wave equation. By the choice of initial data for $k^{\mathrm{aux}}_I,g^{\mathrm{aux}}_I$ (recall Lemma~\ref{lem:locexist}) and Lemma~\ref{lem:D.est}, it follows that 
\begin{equation}\label{eq:k.limit.anti.sym.initial}
\|((k^{\mathrm{aux}}_I)_{ij} - (k^{\mathrm{aux}}_I)_{ji},\, t\rd_t ((k^{\mathrm{aux}}_I)_{ij} - (k^{\mathrm{aux}}_I)_{ji})) \restriction_{t = T_{\mathrm{aux},I}}\|_{H^1(\Sigma_{T_{\mathrm{aux},I}}, g)\times L^2(\Sigma_{T_{\mathrm{aux},I}},g)} \leq C t^{-1+(n+1)\ve}.
\end{equation}

We now perform energy estimates for $(k^{\mathrm{aux}}_I)_{ij} - (k^{\mathrm{aux}}_I)_{ji}$ using the wave equation in Lemma~\ref{lem:eqn.for.anti.k} (in a manner similar to the $k$ energy estimates in the proof of Theorem~\ref{thm:bootstrap}, only simpler). By choosing $n$ sufficiently large, the estimate \eqref{eq:k.limit.anti.sym.initial} allows one to take care the borderline terms and moreover show that for any $T_0 \in (0, T_{N_0,s,n})$, 
\begin{equation}\label{eq:lim.I.k.anti}
\lim_{I\to +\infty} \sup_{t\in [T_0,T_{N_0,s,n})} \|(k^{\mathrm{aux}}_I)_{ij} - (k^{\mathrm{aux}}_I)_{ji} \|_{H^1(\Sigma_t, g)}  = 0.
\end{equation}
 
Finally, since $k_{ij}$ is the pointwise limit of $(k^{\mathrm{aux}}_I)_{ij}$ as $I \to +\infty$ (by Lemma~\ref{lem:AA}), the estimate \eqref{eq:lim.I.k.anti} implies that $k_{ij}$ is symmetric in $i$ and $j$, which is what we wanted to prove. 
\end{proof}

\begin{proof}[Proof of Proposition~\ref{prop:AA.main}]
Proposition~\ref{prop:AA.main} follows directly from Lemmas~\ref{lem:AA} and \ref{lem:k.is.symmetric}. \qedhere
\end{proof}

\section{Vanishing of the Einstein tensor}\label{subsec:vanEVE}

The goal of this section is to show that the solution of \eqref{eq:hyp.sys}, constructed in Theorem~\ref{thm:main.reduced} in subsection \ref{subsec:locexist}, is in fact a solution to the Einstein vacuum equations.  {This then concludes the proof of Theorem~\ref{mainthm}; see the conclusion of the proof at the end of the section.}

We begin with the following:
\begin{proposition}\label{prop:h.is.trk}
There exists $N_h\in \mathbb N$ sufficiently large such that the following holds. 

Let $s \geq 5$ and $N_0 \geq N_h$. Then, for $n \geq n_{N_0,s}$, the solution $(g,h,k)$ to \eqref{eq:hyp.sys} given by Theorem~\ref{thm:main.reduced} satisfies
$$h = k_\ell{ }^\ell.$$
In particular, $^{(4)}Ric(^{(4)}g)_{tt} = 0$.
\end{proposition}
\begin{proof}
Once we establish that $h = k_\ell{ }^\ell$, it follows from the first equation in \eqref{eq:hyp.sys} that $\rd_t k_\ell{ }^\ell = |k|^2$. According to \eqref{Rictt}, this in turn implies that $^{(4)}Ric(^{(4)}g)_{tt} = 0$.

Taking the trace of the second equation in \eqref{eq:hyp.sys} and using the identity \eqref{eq:trace.of.nonlinear}, we obtain
$$\rd_t[\rd_t k_\ell{ }^\ell - |k|^2] = \Delta_g(k_\ell{ }^\ell - h) + 2 k_i{ }^i [\rd_t k_\ell{ }^\ell - |k|^2].$$
Since by \eqref{eq:hyp.sys} $\rd_t h = |k|^2$, it follows that
\begin{align}\label{boxtrk-h}
\rd_t^2(k_\ell{ }^\ell - h)= \Delta_g(k_\ell{ }^\ell -  {h}) + 2 k_i{ }^i \rd_t (k_\ell{ }^\ell -  {h}).
\end{align}
Note that this is a wave equation for $(k_\ell{ }^\ell - h)$.  We can then carry out a similar energy estimates as in the proof of Theorem~\ref{thm:main.reduced} to obtain
\begin{equation}\label{eq:h.is.trk.pf.1}
\begin{split}
&\: t^2 \|\rd_t (k_\ell{ }^\ell - h)\|_{L^2(\Sigma_t,g)}^2 + \sum_{r=0}^1 t^{2r} \| k_\ell{ }^\ell - h\|_{H^r(\Sigma_t,g)}^2 \\
\leq &\: \f{C_0 + C_n t^\ve}{t} (t^2 \|\rd_t (k_\ell{ }^\ell - h)\|_{L^2(\Sigma_t,g)}^2 + \sum_{r=0}^1 t^{2r} \| k_\ell{ }^\ell - h\|_{H^r(\Sigma_t,g)}^2),
\end{split}
\end{equation}
where we have used the estimates for $k$ given in Proposition \ref{prop:k.general}. Here, as in the previous section, we use $C_0$ to denote constants depending only on $s$, $c_{ij}$ and $p_i$, while $C_n$ can depend in addition on $n$ and $N_0$.

At the same time, by Theorem~\ref{thm:main.reduced} and the fact that $h^{\bf [n]} = (\kn)_\ell{}^\ell$,
\begin{equation}\label{eq:h.is.trk.pf.2}
\|k_\ell{ }^\ell - h\|_{H^1(\Sigma_t,g)}^2 + \|\rd_t(k_\ell{ }^\ell - h)\|_{L^2(\Sigma_t,g)}^2 \leq 2 t^{2N_0+2s-2}.
\end{equation}

In particular, choosing $N_h$ sufficiently large, the estimates \eqref{eq:h.is.trk.pf.1}, \eqref{eq:h.is.trk.pf.2} and Gr\"onwall's inequality implies that 
$$\|k_\ell{ }^\ell - h\|_{H^1(\Sigma_t,g)}^2 + \|\rd_t(k_\ell{ }^\ell - h)\|_{L^2(\Sigma_t,g)}^2  = 0,$$
which in turn implies the desired conclusion.
\end{proof}

\begin{proposition}\label{prop:hyp.est.for.prop.const}
 {There exists $N_G \geq N_h$ and $n_G$ sufficiently large such that the following holds.}

 {Let $s \geq 5$ and $N_0 \geq N_G$. For $n \geq \max\{ n_{N_0,s},\, n_G\}$, take the solution $(g,h,k)$ to \eqref{eq:hyp.sys} given by Theorem~\ref{thm:main.reduced}. Then $^{(4)} g = -\ud t^2 + g$ is} in fact a solution to the Einstein vacuum equations, i.e.~$ {Ric({ }^{(4)}g) = 0}$,  {and $k$ is the corresponding second fundamental form of the constant-$t$ hypersurfaces.}
\end{proposition}
\begin{proof}
 {For this proof, we denote $\mathcal G_i = G_{ti}(^{(4)}g)$ and $\mathfrak G_{ij} = G_{ij}(^{(4)}g)$, both thought of as $\Sigma_t$-tangent tensors. We also use the notation that $\nab$ is the Levi--Civita connection for the spatial metric $g$.}

\pfstep{Step~1: Derivation of a system of equations}  {By \eqref{eq:k.wave.prelim} and the wave equation \eqref{eq:k.wave}, we have}
\begin{align}\label{eq:Ric.high}
\begin{split}
\partial_t Ric_i{}^j({^{(4)}g})=&\,\nabla_i \mathcal{G}^j+\nabla^j\mathcal{G}_i
-3k_i{}^mRic_m{}^j({^{(4)}g}) + 2\de_i^jk_m{}^\ell Ric_{\ell}{}^m({^{(4)}g}) -k_\ell{}^j Ric_i{}^\ell({^{(4)}g})\\
&+ 2k_\ell{}^\ell Ric_i{}^j({^{(4)}g})
-(k_\ell{}^\ell \de_i^j - k_i{}^j) Ric_m{}^m({^{(4)}g}) {.}
\end{split}
\end{align}
Taking the trace of \eqref{eq:Ric.high} and using the fact that $Ric_{tt}({^{(4)}g})=0$, we also have:
\begin{align}\label{eq:R.high}
\partial_t R({^{(4)}g})=&\,2\nabla_j \mathcal{G}^j
+2k_m{}^\ell Ric_\ell{}^m({^{(4)}g}).
\end{align}
The combination of \eqref{eq:Ric.high}  {and} \eqref{eq:R.high} implies the following equation for the Einstein tensor $G_i{}^j({^{(4)}g})$:
\begin{align}\label{eq:G.high}
\notag\partial_t  G_i{}^j({^{(4)}g}):=&\,\partial_tRic_i{}^j({^{(4)}g})-\frac{1}{2}\delta_i{}^j\partial_tR({^{(4)}g})\\
=&\,\nabla_i \mathcal{G}^j+\nabla^j\mathcal{G}_i-\delta_i{}^j\nabla_\ell \mathcal{G}^\ell
-3k_i{}^mRic_m{}^j({^{(4)}g}) + \de_i^jk_m{}^\ell Ric_{\ell}{}^m({^{(4)}g}) -k_\ell{}^j Ric_i{}^\ell({^{(4)}g})\\
\notag& +2k_\ell{}^\ell Ric_i{}^j({^{(4)}g})
-(k_\ell{}^\ell \de_i^j - k_i{}^j) Ric_m{}^m({^{(4)}g}).
\end{align}
Note that $Ric_i{}^j({^{(4)}g})$ can be written in terms of $G_i{}^j({^{(4)}g})$: $Ric_i{}^j({^{(4)}g})=G_i{}^j({^{(4)}g})+\frac{1}{2}\delta_i{}^j R({^{(4)}g})$, where $R({^{(4)}g}):=-R_{tt}({^{(4)}g})+R_\ell{}^\ell({^{(4)}g}) = R_\ell{}^\ell({^{(4)}g})$ by Proposition~\ref{prop:h.is.trk}. Taking the trace we get $Ric_i{ }^i = G_i{ }^i + \f 32R_\ell{}^\ell({^{(4)}g})$ so that $R_\ell{}^\ell({^{(4)}g}) = 2G_i{}^i({^{(4)}g})$. It follows that 
\begin{equation}\label{eq:Ric.in.terms.of.G}
Ric_i{}^j({^{(4)}g})= \mathfrak G_i{}^j({^{(4)}g})+\de_i{ }^j \mathfrak G_\ell{}^\ell.
\end{equation}

We can thus rewrite \eqref{eq:G.high} as
\begin{align}\label{eq:G.high.2}
\partial_t \mathfrak G_i{}^j
=&\,\nabla_i \mathcal{G}^j+\nabla^j\mathcal{G}_i-\delta_i{}^j\nabla_\ell \mathcal{G}^\ell + (k \star \mathfrak G)_i{ }^j,
\end{align}
where $(k \star \mathfrak G)_i{ }^j$ is some quadratic contraction of $k$ and $\mathfrak G$ whose exact form is unimportant.

On the other hand, by the contracted second Bianchi equations and the fact that $D_t G_{ti} (^{(4)} g) = \rd_t \mathcal G_i + k_i{}^j \mathcal G_j$, and $D_j G_i{}^j({^{(4)}g})=\nabla_j \mathfrak G_i{}^j +k_j{}^j\mathcal{G}_i+k_i{}^j\mathcal{G}_j$, we obtain
\begin{align}\label{eq:mathcalGiODE}
\partial_t\mathcal{G}_i
=&\,k_j{}^j\mathcal{G}_i+\nabla_j\mathfrak G_i{}^j({^{(4)}g}) {.}
\end{align}
Taking $\rd_t$ of \eqref{eq:mathcalGiODE}, applying \eqref{eq:G.high.2}, and using the commutation formula in Proposition~\ref{prop:commutation.formula}, we obtain the wave equation
\begin{equation}\label{eq:wave.eqn.for.G}
\begin{split}
\rd_t^2 \mathcal G_i = &\: \rd_t (k_j{ }^j \mathcal G_i) + \nab_j (\nabla_i \mathcal{G}^j+\nabla^j\mathcal{G}_i-\delta_i{}^j\nabla_\ell \mathcal{G}^\ell + (k\star \mathfrak G)_i{ }^j)  + [\rd_t, \nab_j] \mathfrak G_i{ }^j \\
= &\:  \Delta_g \mathcal G_i + k \star k \star \mathcal G + \rd_t k \star \mathcal G + k \star \rd_t \mathcal G + \nab k\star \mathfrak G + k \star \nab \mathfrak G,
\end{split}
\end{equation}
where in the last equality we have also used that the curvature tensor $Riem(g)$ can be expressed in terms of $\mathfrak G$, $k$ and $\rd_t k$ using \eqref{Riem=Ric}, \eqref{Ricij},  and \eqref{eq:Ric.in.terms.of.G}, so that
$$\nab_j (\nabla_i \mathcal{G}^j+\nabla^j\mathcal{G}_i-\delta_i{}^j\nabla_\ell \mathcal{G}^\ell) = \nab_j \nab_i \mathcal G^j + \Delta_g \mathcal G_i - \nab_i \nab_j \mathcal G^j = \Delta_g \mathcal G_i + k \star k \star \mathcal G + \rd_t k \star \mathcal G.$$
Here, $k \star k \star \mathcal G$, etc.~are in principle explicit, but we do not carry out the computations as the exact form is unimportant.

\pfstep{Step~2: Energy estimates and vanishing of the Einstein tensor} Our goal now is to perform energy estimates using \eqref{eq:G.high.2} and \eqref{eq:wave.eqn.for.G} so as to show that $\mathfrak G$ and $\mathcal G$ are both $\equiv 0$. Investigating the terms in \eqref{eq:G.high.2} and \eqref{eq:wave.eqn.for.G}, we note that the RHS of \eqref{eq:wave.eqn.for.G} has terms with one derivative of $\mathcal G$, which apparently leads to a loss of derivatives. Nevertheless, this can be treated in \emph{exactly} the same manner as \eqref{eq:hyp.sys}.

Define the energy
\begin{equation}
\begin{split}
E(t) = &\: \sum_{r=0}^1 t^{2r}\|\rd_t \mathcal G \|_{\dot H^r(\Sigma_t,g)}^2 + \sum_{r=0}^2 t^{-2+2r} \|\mathcal G \|_{\dot H^r(\Sigma_t,g)}^2 + \sum_{r=0}^2 t^{-2+2r} \|\mathfrak G\|_{\dot H^r(\Sigma_t,g)}^2,
\end{split}
\end{equation}
and modified energy
\begin{equation}
\begin{split}
\widetilde{E}(t) = \sum_{r=0}^1 t^{2r}\|\rd_t \nab^{(r)} \mathcal G \|_{L^2(\Sigma_t,g)}^2 + \sum_{r=0}^2 t^{-2+2r} \|\mathcal G \|_{\dot H^r(\Sigma_t,g)}^2 + \sum_{r=0}^1 t^{-2+2r} \|\mathfrak G\|_{\dot H^r(\Sigma_t,g)}^2 + t^2\|\widetilde{\nab^{(2)} \mathfrak G} \|_{L^2(\Sigma_t,g)}^2,
\end{split}
\end{equation}
where
$$(\widetilde{\nab^{(2)} \mathfrak G})_i{ }^j:= \Delta_g \mathfrak G_i{ }^j - \rd_t \nab_i \mathcal G^j - \rd_t \nab^j \mathcal G_i + \de_i{ }^j \rd_t \nab_\ell \mathcal G^\ell.$$

We now carry out energy estimates for the wave-transport system \eqref{eq:G.high.2} and \eqref{eq:wave.eqn.for.G} in a manner similar to that for \eqref{eq:hyp.sys} in Theorem~\ref{thm:main.reduced}. Note that we in particular need to use the elliptic estimates in Lemma~\ref{lem:elliptic}. Nevertheless, the present case is much easier because of the linearity of the system. We omit the proof and give the estimates
\begin{equation}\label{eq:E.for.G}
\f{\ud}{\ud t} E(t) \leq \f{C_0 + C_n t^\ve }{t} E(t),
\end{equation}
where we again used the convention that $C_0$ depends only on $s$, $c_{ij}$ and $p_i$, while $C_n$ can depend in addition on $n$ and $N_0$. We now fix $C_0$ and $C_n$ so that \eqref{eq:E.for.G} holds. 

We now need to show, using \eqref{eq:E.for.G}, that $E(t) \equiv 0$. For this purpose it suffices to check that 
\begin{equation}\label{eq:E.for.G.goal}
\lim_{t\to 0^+} t^{-C_0} E(t) = 0,
\end{equation} 
so that 
we can apply Gr\"onwall's inequality to $\f{\ud}{\ud t} (t^{-C_0} E(t)) \leq \f{C_n }{t^{1-\ve}} (t^{-C_0} E(t))$.

Define $\mathcal G^{\bf [n]} = G_{ti}(^{(4)} g^{\bf [n]})$ and $\mathfrak G_{ij}^{\bf [n]} = G_{ij}(^{(4)} g^{\bf [n]})$. Then by Proposition~\ref{prop:approx.G}, there exists $n_G \in \mathbb N$ such that if $n\geq n_G$, then
\begin{equation}\label{eq:E.for.G.goal.1}
\lim_{t\to 0^+} t^{-C_0} \left( \sum_{r=0}^1 t^{2r}\|\rd_t \mathcal G^{\bf [n]} \|_{\dot H^r(\Sigma_t,g)}^2 + \sum_{r=0}^2 t^{-2+2r} \|\mathcal G^{\bf [n]} \|_{\dot H^r(\Sigma_t,g)}^2 + \sum_{r=0}^2 t^{-2+2r} \|\mathfrak G^{\bf [n]}\|_{\dot H^r(\Sigma_t,g)}^2\right) = 0.
\end{equation}

On the other hand, by \eqref{eq:main.reduced.est.in.thm} in Theorem~\ref{thm:main.reduced}, if $N_G$ is sufficiently large and $N_0\geq N_G$, then
\begin{equation}\label{eq:E.for.G.goal.2}
\begin{split}
\lim_{t\to 0^+} t^{-C_0} &\: \left( \sum_{r=0}^1 t^{2r}\|\rd_t (\mathcal G - \mathcal G^{\bf [n]}) \|_{\dot H^r(\Sigma_t,g)}^2  \right. \\
&\: \left. + \sum_{r=0}^2 t^{-2+2r} \|(\mathcal G -\mathcal G^{\bf [n]}) \|_{\dot H^r(\Sigma_t,g)}^2 + \sum_{r=0}^2 t^{-2+2r} \|\mathfrak G^{\bf [n]}\|_{\dot H^r(\Sigma_t,g)}^2\right) = 0.
\end{split}
\end{equation}

Therefore, choosing $N_0\geq N_G$ and $n\geq \max\{n_{N_0,s},n_G\}$, we obtain \eqref{eq:E.for.G.goal} by using \eqref{eq:E.for.G.goal.1} and \eqref{eq:E.for.G.goal.2}. This gives $E(t) \equiv 0$. Together with Proposition~\ref{prop:h.is.trk}, this gives that the Einstein tensor vanishes identically. \qedhere

\end{proof}

We end the section with the conclusion of the proof of Theorem~\ref{mainthm}:

\begin{proof}[Proof of Theorem~\ref{mainthm}]
This follows immediately from Theorem~\ref{thm:main.reduced} and Proposition~\ref{prop:hyp.est.for.prop.const}. \qedhere
\end{proof}

\section{Uniqueness and smoothness of solutions: proofs of Theorems~\ref{thm:uniq}  {and \ref{thm:smooth}}}\label{sec:uni}

We prove Theorems~\ref{thm:uniq} and \ref{thm:smooth} in Sections~\ref{sec:thm.uniq} and \ref{sec:thm.smooth} respectively.

\subsection{Uniqueness of solutions}\label{sec:thm.uniq}

\begin{proof}[Proof of Theorem~\ref{thm:uniq}]
Let ${^{(4)}g},{^{(4)}\tilde g}$ be two solutions to the Einstein vacuum equations \eqref{eq:EVE} satisfying the assumptions of Theorem~\ref{thm:uniq}. 

In this proof, \textbf{we use $C$ to denote positive constants depending only on $c_{ij}$ and $p_j$, and use $C'$ to denote positive constants which depend \underline{in addition} on the implicit constants in \eqref{uniqueness.condition.0}, \eqref{uniqueness.condition.1} and  \eqref{uniqueness.condition.2}.}

Notice that it suffices to prove uniqueness on a sub-domain $(0,T']\times \mathbb T^3$ (for some $0<T'<T$) since in the region $[T',T]\times \mathbb T^3$, we are away from the singularity, and uniqueness will follow from standard uniqueness results. For this reason, \textbf{we will take $T'$ sufficiently small so as to assume $C' (T')^\ve \leq 1$}.

\pfstep{Step~1: Estimating $k$ and $\tilde{k}$} Using the estimates \eqref{uniqueness.condition.0} and \eqref{uniqueness.condition.1}, and arguing as in Propositions~\ref{prop:kn.higher} and \ref{prop:kn.lowest}, we obtain
\begin{equation}\label{eq:k.est.in.uniq}
\sum_{r=0}^2 t^{r} ( \|\nab^{(r)} k\|_{L^\i(\Sigma_t,g)} + \|\tilde{\nab}^{(r)} \tilde{k}\|_{L^\i(\Sigma_t,g)}) \leq Ct^{-1},
\end{equation}
and 
\begin{equation}\label{eq:dtk.est.in.uniq}
\sum_{r=0}^1 t^{r} ( \| \nab^{(r)}\rd_t k \|_{L^\i(\Sigma_t,g)} + \| \tilde{\nab}^{(r)} \rd_t\tilde{k} \|_{L^\i(\Sigma_t,g)}) \leq Ct^{-2},
\end{equation}
where $\tilde{\nab}$ denotes the Levi--Civita connection of $\tilde{g}$.

\pfstep{Step~2: Estimating the convergence rate as $t\to 0^+$} Let 
\begin{equation}\label{eq:def.h.th}
 h = k_\ell{ }^\ell,\quad \tilde{h} = \tilde{k}_\ell{ }^\ell.
 \end{equation}
Define the variables 
\begin{align*}
g^{(d)}:=g-\tilde{g}, \quad (g^{-1})^{(d)} := g^{-1} - \tilde{g}, \quad h^{(d)} := h - \tilde{h},\quad k^{(d)} := k - \tilde{k}.
\end{align*}

Given any $M_u' \in \mathbb N$ we can choose $M_u$ sufficiently large so that by \eqref{uniqueness.condition.0} and \eqref{uniqueness.condition.2},
\begin{equation}\label{eq:diff.decay.in.uniq}
\|g^{(d)}\|_{H^2(\Sigma_t,g)} + \|(g^{-1})^{(d)}\|_{H^2(\Sigma_t,g)} + \|h^{(d)}\|_{H^2(\Sigma_t,g)} + \| k^{(d)} \|_{H^1(\Sigma_t,g)} \leq C' t^{M_u'}.
\end{equation}

Moreover, given any $M_u'' \in \mathbb N$ we can choose $M_u$ even larger so that by \eqref{uniqueness.condition.2},
\begin{equation}\label{eq:diff.Ric.decay.in.uniq}
\| Ric(g) - Ric(\tilde{g}) \|_{L^2(\Sigma_t,g)} \leq C' t^{M_u''}.
\end{equation}
Now since both ${^{(4)}g}$ and ${^{(4)}\tilde g}$ solve \eqref{eq:EVE}, the RHS of \eqref{Ricij} vanishes for both metrics. Hence, using  \eqref{eq:k.est.in.uniq}, \eqref{eq:diff.decay.in.uniq} and \eqref{eq:diff.Ric.decay.in.uniq}, we obtain
\begin{equation}\label{eq:diff.decay.in.uniq.2}
\|\rd_t k^{(d)} \|_{L^2(\Sigma_t,g)}\leq C' \max\{ t^{M_u'-1},\,t^{M_u''}\}.
\end{equation}

\pfstep{Step~3: Energy estimates} We now carry out energy estimates for $(g^{(d)},h^{(d)},k^{(d)})$. First, we note that they satisfy a system of equations analogous to \eqref{eq:k.diff}, \eqref{eq:h.diff}, \eqref{eq:g.diff}, \eqref{eq:g-1.diff} as follows.
\begin{itemize}
\item By definition of $k$ and $\tilde{k}$, we immediate obtain the transport equation $\rd_t g^{(d)} = -2\tilde k_{(i}{}^\ell g^{(d)}_{j)l}-2(k^{(d)})_{(i}{}^\ell g_{j)l}$.
\item By \eqref{eq:def.h.th} and \eqref{Rictt}, $h^{(d)}$ satisfies a transport equation $\rd_t h^{(d)} = |k|^2 - |\tilde{k}|^2$
\item Arguing as in Section~\ref{sec:actual}, it follows that both $k$ and $\tilde{k}$ satisfy the wave equation \eqref{eq:k.wave} (with the corresponding metric $g$ and $\tilde{g}$). We take the difference to obtain a wave equation for $k^{(d)}$.
\end{itemize}
Note that these equations are similar to but simpler than \eqref{eq:k.diff}, \eqref{eq:h.diff} and \eqref{eq:g.diff} in the sense that the system is  \emph{homogeneous}.

We can thus carry out energy estimates in exactly the same way as in the proof of Theorem~\ref{thm:main.reduced}, including using a modified energy together with elliptic estimates. In particular, defining 
\begin{equation*}
\begin{split}
\mathcal E_u(t) = &\: \sum_{r=0}^2 [ t^{-2+2r}(\|g^{(d)}\|_{H^r(\Sigma_t,g)}^2 + \|(g^{-1})^{(d)}\|_{H^r(\Sigma_t,g)}) + t^{2+2r} \|h^{(d)}\|_{H^r(\Sigma_t,g)}]\\
&\: + \sum_{r=0}^1 t^{2r} \|k^{(d)}\|_{L^2(\Sigma_t,g)} + t^2 \|\rd_t k^{(d)}\|_{L^2(\Sigma_t,g)}, 
\end{split}
\end{equation*}
we can run the energy estimates in Theorem~\ref{thm:main.reduced} using the bounds established in Steps~1 and 2 above.
\begin{itemize}
\item Estimates \eqref{eq:k.est.in.uniq} and \eqref{eq:dtk.est.in.uniq} in Step~1 guarantee that
\begin{equation}\label{eq:uniq.energy.1}
\f{\ud}{\ud t} \mathcal E_u(t) \leq \f Ct  \mathcal{E}_{u}(t)
\end{equation}
for some fixed constant $C>0$ depending only on the constants in \eqref{uniqueness.condition.1}.
\item Taking $C$ as in \eqref{eq:uniq.energy.1}, estimates \eqref{eq:diff.decay.in.uniq} and \eqref{eq:diff.decay.in.uniq.2} in Step~2 guarantee that if $M_u$ sufficiently large, then
\begin{equation}\label{eq:uniq.energy.2}
\limsup_{t\to 0^+} t^{-C} \mathcal E_u(t) =0.
\end{equation}
\end{itemize}

The bounds \eqref{eq:uniq.energy.1} and \eqref{eq:uniq.energy.2} immediately imply that $\mathcal E_u\equiv 0$, which in particular implies $g \equiv \tilde g$, which is what we wanted to prove.
\end{proof}

\subsection{Regularity of solutions}\label{sec:thm.smooth}
Our goal in this subsection is to prove Theorem~\ref{thm:smooth}. As already mentioned in the introduction, for the proof we rely on our uniqueness result. 

We first introduce a piece of notation for the rest of this subsection. Let $s \geq 5$ and $N_0 \in \mathbb N$. For $n\geq n_{N_0,s}$, Theorem~\ref{thm:main.reduced} and Proposition~\ref{prop:hyp.est.for.prop.const} give a solution to the Einstein vacuum equations of the form \eqref{metricansatz} which satisfies the estimates  \eqref{eq:main.reduced.est.in.thm}. \textbf{We denote such a solution by $g_{N_0,s,n}$ and denote the corresponding second fundamental form by $k_{N_0,s,n}$.}

We need the following lemma, which checks the conditions \eqref{uniqueness.condition.0} and \eqref{uniqueness.condition.1} in Theorem~\ref{thm:uniq}.
\begin{lemma}\label{lem:uniq.check.cond}
Let $n_{N_0,s}$ be as in Theorem~\ref{thm:bootstrap} and $N_G$, $n_G$ be as in Proposition~\ref{prop:hyp.est.for.prop.const}. There exists $N_c \geq N_G$ sufficiently large such that if $N_0\geq N_c$, $s\geq 5$ and $n \geq \max\{ n_{N_0,s}, n_G\}$, then for $g = g_{N_0,s,n}$ and $k = k_{N_0,s,n}$, there exists $C>0$ depending on $N_0$, $s$, $n$, $c_{ij}$ and $p_i$ such that
\begin{equation}\label{uniqueness.condition.1.0}
\sum_{|\alp|\leq 2} |\rd_x^\alp(a_{ij} - c_{ij})| \leq C t^\ve,
\end{equation}
\begin{equation}\label{uniqueness.condition.1.1}
\sum_{r=0}^1 \sum_{|\alp|\leq 2-r} t^r |\rd_t^r \rd_x^\alp (k_i{ }^j - t^{-1} \kappa_i{ }^j) | \leq C t^{-1+\ve}.
\end{equation}
\end{lemma}
\begin{proof}
In this proof, we allow the implicit constants $C>0$ to depend on $N_0$, $s$, $n$, $c_{ij}$ and $p_i$.

We first prove \eqref{uniqueness.condition.1.0}. Since $s\geq 5$, by \eqref{eq:main.reduced.est.in.thm} and \eqref{eq:Sobolev.2}, we have
$$\sum_{r=0}^2 t^r \|g - g^{\bf [n]}\|_{W^{r,\i}(\Sigma_t,g)}  \leq C t^{N_0 +s-\f 32}.$$
Now note that the smallest eigenvalue of $g^{-1}$ is $\geq C^{-1} t^{-2p_1} \geq C^{-1} t^{2}.$
Hence,
$$|(g - g^{\bf [n]})_{ij}| + t^2|\nab_\ell (g - g^{\bf [n]})_{ij}| + t^4|\nab_b \nab_\ell (g - g^{\bf [n]})_{ij}| \leq C t^{N_0 +s-\f 72}. $$
Writing the covariant derivatives in terms of coordinate derivatives, using $g_{ij} - g^{\bf [n]}_{ij} = t^{2 p_{\max\{i,j\}}} (a_{ij} - a^{\bf [n]}_{ij})$, and choosing $N_c$ sufficiently large, we thus obtain 
\begin{equation}\label{eq:check.uniq.a.diff}
\sum_{|\alp|\leq 2}  |\rd_x^\alp(a_{ij} - a^{\bf [n]}_{ij})| \leq C t^\ve.
\end{equation}
The estimate \eqref{uniqueness.condition.1.0} then follows from \eqref{eq:check.uniq.a.diff}, \eqref{eq:main.parametrix.a.bd} and the triangle inequality.

The proof of \eqref{uniqueness.condition.1.1} is similar, where we first use \eqref{eq:main.reduced.est.in.thm} and \eqref{eq:Sobolev.2} to obtain
$$\sum_{r=0}^2 t^r \|k - k^{\bf [n]}\|_{W^{r,\i}(\Sigma_t,g)} + \sum_{r=0}^1 t^{r+1}\|\rd_t (k - \kn) \|_{W^{r,\i}(\Sigma_t,g)} \leq C t^{N_0 +s-\f 52}.$$
Then, after choosing $N_c$ sufficiently large, we can obtain the desired \eqref{uniqueness.condition.1.1} using \eqref{eq:main.parametrix.k.bd} and the triangle inequality. \qedhere

\end{proof}

We are now ready to prove Theorem~\ref{thm:smooth}:
\begin{proof}[Proof of Theorem~\ref{thm:smooth}]
Given $M_u$ as in Theorem~\ref{thm:uniq}, the following holds:
\begin{itemize}
\item There exists $n_r \in \mathbb N$ sufficiently large such that if $n,\,n'\geq n_r$, then
\begin{equation}\label{eq:proof.smooth.1}
\sum_{r=0}^{1} \sum_{|\alp|\leq 3-r} |\rd_t^r \rd_x^\alp (g^{\bf [n]} - g^{\bf [n']})| = O(t^{M_u}).
\end{equation}
This is because of the estimates \eqref{eq:k.diff.est} and \eqref{eq:a.diff.est} derived in the proof of Theorem~\ref{thm:parametrix}.
\item There exists $N_r \geq N_c$ (where $N_c$ is as in Lemma~\ref{lem:uniq.check.cond}) sufficiently large such that the following holds. Suppose $s \geq 5$, $N_0 \geq N_r$ and $n\geq n_{N_0,s}$, then
\begin{equation}\label{eq:proof.smooth.2}
\sum_{r=0}^{1} \sum_{|\alp|\leq 3-r} |\rd_t^r \rd_x^\alp (g_{N_0,s,n} - g^{\bf [n]})| = O(t^{M_u}).
\end{equation}
This is a direct consequence of \eqref{eq:main.reduced.est.in.thm} and Sobolev embedding.
\end{itemize}

Fix $(g_{N_0 = N_r,s=5,n_0}, k_{N_0 = N_r,s=5,n_0})$ on $(0,T_{N_0 = N_r,s=5,n_0}] \times \mathbb T^3$, where $n_0 \geq \max\{ n_{N_0 = N_r,s=5},\, n_r,\,n_G\}$. We want to show that this particular solution is in fact smooth. Let $s_0 \geq 5$ be arbitrary. By Theorem~\ref{thm:main.reduced} we obtain a solution $(g_{N_0=N_r,s=s_0,n},k_{N_0,s=s_0,n})$ on $(0,T_{N_0 = N_r,s=s_0,n}] \times \mathbb T^3$ for some $n \geq \max\{n_{N_0 = N_r,s=s_0},\, n_r,\,n_G\}$. We now claim that in fact on the common domain of existence $(0, \min \{ T_{N_0 = N_r,s=5,n_0},\, T_{N_0 = N_r,s=s_0,n} \} ]\times \mathbb T^3$, we have
\begin{equation}\label{eq:uniq.claim.in.reg}
g_{N_0 = N_r,s=5,n_0} \equiv g_{N_0=N_r,s=s_0,n}.
\end{equation}
To prove the claim, it suffices to verify the conditions of Theorem~\ref{thm:uniq}:
\begin{itemize}
\item Since $s\geq 5$ and $N_0 = N_r \geq N_c$, the conditions \eqref{uniqueness.condition.0} and \eqref{uniqueness.condition.1} hold because of Lemma~\ref{lem:uniq.check.cond}.
\item By \eqref{eq:proof.smooth.1}, \eqref{eq:proof.smooth.2} and the triangle inequality, our choice of $n_0$, $n$, $N_0$, $s$ implies that
$$\sum_{r=0}^{1} \sum_{|\alp|\leq 3-r} |\rd_t^r \rd_x^\alp (g_{N_0 = N_r,s=5,n_0} - g_{N_0=N_r,s=s_0,n} )| = O(t^{M_u}),$$
i.e.~\eqref{uniqueness.condition.2} holds.
\end{itemize}
This establishes \eqref{eq:uniq.claim.in.reg}.

As a result of \eqref{eq:uniq.claim.in.reg}, it follows that the fixed solution $(g_{N_0 = N_r,s=5,n_0}, k_{N_0 = N_r,s=5,n_0})$ is in $H^{s_0+1}\times H^{s_0}$ for every $t\in (0, \min \{ T_{N_0 = N_r,s=5,n_0},\, T_{N_0 = N_r,s=s_0,n} \} ]$. Now we use energy estimates as in the proof of Theorem~\ref{thm:main.reduced} to show \emph{propagation of regularity}: it then follows that the solution is in $H^{s_0+1}\times H^{s_0}$ for every $t$ in the original time interval, i.e.~for every $t\in (0,  T_{N_0 = N_r,s=5,n_0}]$.

Since $s_0$ can be arbitrarily large, it follows from Sobolev embedding and the equations \eqref{eq:hyp.sys} that the fixed solution $(g_{N_0 = N_r,s=5,n_0}, k_{N_0 = N_r,s=5,n_0})$ is in fact smooth in $(0,  T_{N_0 = N_r,s=5,n_0}]\times \mathbb T^3$. This concludes the proof of the theorem.
\end{proof}

\def\cprime{$'$} \def\cprime{$'$} \def\cprime{$'$}

\end{document}